%% file: MainArxiv.tex
\documentclass[11pt,a4paper]{amsart}
\usepackage[english]{babel}
\usepackage[utf8]{inputenc}
\usepackage[foot]{amsaddr}
\usepackage{amsmath,amsfonts}
\usepackage{bbm}
\usepackage{amssymb,amsthm}
\usepackage{thmtools}
\usepackage[format=plain, indention=.5cm]{caption}
\usepackage{subcaption}
\usepackage{graphicx}
\usepackage{color}
\usepackage[dvipsnames]{xcolor}
\usepackage{textcomp}
\usepackage[a4paper]{geometry} 
\usepackage{mathtools}
\usepackage[
final,
colorlinks,
linkcolor = BrickRed,
citecolor = OliveGreen,
filecolor = black,
urlcolor = BrickRed,
bookmarks=true,
bookmarksnumbered=true,
bookmarksopen=true,
bookmarksopenlevel = 0,
hypertexnames=false,
]{hyperref}

\usepackage{lmodern} 
\usepackage{microtype}

\usepackage{csquotes}

\usepackage{siunitx}
\newtheorem{lemma}{Lemma}
\newtheorem{remark}[lemma]{Remark}

\usepackage{float}
\usepackage{algorithm}
\makeatletter
\def\plist@algorithm{Alg.\space}
\makeatother
\usepackage{algpseudocodex}
\algrenewcommand\algorithmicrequire{\textbf{Input:}}
\algrenewcommand\algorithmicensure{\textbf{Output:}}

\usepackage{makecell}
\usepackage{multirow}
\usepackage{booktabs}

\usepackage{xcolor}
\usepackage{tikz}
\usetikzlibrary{patterns}
\usepackage{pgfplots}

\definecolor{unia-green}{RGB}{0, 101, 97}
\definecolor{unia-pink}{RGB}{173, 0, 124}
\definecolor{unia-yellow}{RGB}{246, 168, 0}
\definecolor{unia-orange}{RGB}{235, 105, 11}
\definecolor{unia-red}{RGB}{212, 0, 45}
\definecolor{unia-lightblue}{RGB}{0, 174, 207}
\definecolor{unia-blue}{RGB}{0, 135, 193}
\definecolor{unia-lightgreen}{RGB}{72, 147, 36}

\AddToHook{cmd/appendix/before}{\crefalias{section}{appendix}}

\usepackage[capitalise,noabbrev,nameinlink]{cleveref}

\title[Effective anisotropic permeabilities]{Effective permeabilities for flow through anisotropic microscopic geometries}

\author[]{L.~Balazi$^{*}$, F.~Holzberger$^{\S}$, S.~B.~Lunowa$^{\S}$, M.~A.~Peter$^{\dagger}$, D.~Peterseim$^{\dagger}$, B.~Wohlmuth$^{\S}$}

\address{${}^{*}$ Institute of Mathematics, University of Augsburg, Universit\"atsstr.~12a, 86159 Augsburg, Germany}
\address{${}^{\S}$ Department of Mathematics, School of Computation, Information and Technology, Technical University of Munich, 85748 Garching b.~München, Germany}
\address{${}^{\dagger}$ Institute of Mathematics \& Centre for Advanced Analytics and Predictive Sciences (CAAPS), University of Augsburg, Universit\"atsstr.~12a, 86159 Augsburg, Germany}

\email{loic.balazi@uni-a.de}
\email{malte.peter@uni-a.de}
\email{daniel.peterseim@uni-a.de}
\email{fabian.holzberger@tum.de}
\email{stephan.lunowa@tum.de}
\email{wohlmuth@cit.tum.de}

\date{\today}

\begin{document}

\thanks{F.~Holzberger, S.~B.~Lunowa and B.~Wohlmuth gratefully acknowledge the financial support provided by the German Research Foundation (Deutsche Forschungsgemeinschaft, \href{http://dx.doi.org/10.13039/501100001659}{DOI 10.13039/501100001659}) under project number 465242983 within the priority programme \enquote{SPP 2311: Robust coupling of continuum-biomechanical in silico models to establish active biological system models for later use in clinical applications -- Co-design of modeling, numerics and usability}. \\ 
\indent L.~Balazi, M.~A.~Peter and D.~Peterseim gratefully acknowledge the support of the Centre for Future Production of University of Augsburg.}

\begin{abstract}
This work develops a computational and theoretical framework for determining effective permeabilities in anisotropic microscopic geometries containing dense, fibre-like obstacles, motivated by the need to model flow in coiled aneurysm domains accurately. Building on homogenisation theory and fully resolved simulations in Representative Elementary Volumes (REVs), we validate the permeability model introduced in [C. Boutin, Study of permeability by periodic and self-consistent homogenisation. Eur. J. Mech. A Solids, 19(4):603–632, 2000] and propose a systematic methodology for capturing the directional variations induced by fibre orientation. The resulting permeability tensors are incorporated into macroscopic flow simulations based on the Darcy equation, enabling direct comparison of anisotropic and isotropic permeability models across several benchmark configurations.

Our findings show that anisotropy has a significant impact on local flow direction and magnitude, generating directional permeability contrasts which cannot be reproduced by classical isotropic approximations. By integrating coil-induced microstructural effects into continuum-scale hemodynamic models, the proposed approach enables more realistic assessment of post-treatment aneurysm flow behaviour. Beyond this clinical application, the framework is broadly applicable to other biomedical and engineering systems involving fibrous or filamentous porous microstructures.
\end{abstract}

\maketitle

{\scriptsize{\bf Key words.} Effective permeability, Anisotropic microstructures, Homogenisation, Darcy equation, Coiled aneurysms}\\
\indent
{\scriptsize{\bf AMS subject classifications.} {\bf 76M50}, {\bf 35B27}, {\bf 76Z05}, {\bf 92C35}, {\bf 92C50}}

\section{Introduction} 
The precise computation of flow through porous media is essential for many engineering applications, such as groundwater transport through soil \cite{bear_modeling_2018}. In this work, we consider the underlying concepts in the context of a medical application, that is blood flow through thin wired devices such as endovascular coils. Endovascular coiling is one of the common treatment methods for cerebral aneurysms. Understanding and predicting the flow through such coils is of high current interest \cite{adamou2021endovascular, hu2019advances}. 

A cerebral aneurysm is a saccular malformation of a brain artery, which develops due to the gradual dilation of the vessel wall over time. The primary clinical risk is that these walls become progressively thinner and weaker, up to the point where they can rupture even under relatively low stresses. In this case, a subarachnoid haemorrhage occurs, leading to blood entering the subarachnoid space.
The sudden increase in intracranial pressure and disturbance of the normal circulation of the cerebrospinal fluid can lead to severe brain damage and death. 
To mitigate these risks, minimally invasive endovascular procedures have seen significant advancement in the last decades. Among these, endovascular coiling involves the insertion of thin metallic wires via a micro catheter into the aneurysm, where they coil up, creating a form of plug. This procedure directly occludes the aneurysm from the blood flow, lowering the risk of rupture due to sudden haemodynamic changes and is often even applied when rupture has already occurred. Furthermore, the presence of the metallic coils promotes blood clotting within the aneurysm sac, leading to even better occlusion and, in the best case, to the gradual healing and reconstruction of the vessel wall \cite{brinjikji2015mechanisms}.
The choice of coil length, shape and stiffness currently relies largely on the experience of the medical professional. Consequently, it often remains unclear which specific configuration provides the most favourable conditions for healing, or whether an alternative treatment would be superior \cite{seibert2011intracranial}.
Therefore, it is of utmost importance to gain a better understanding of the post-insertion hemodynamics through numerical flow simulations, allowing the effects of different coil designs to be assessed a priori. 

The first approach to study the flow in these kind of media is the direct computation of flow problems (e.g.~the associated Stokes or Navier--Stokes problems). However, simulating the flow in a multiscale medium with many obstacles is very challenging.
Typically, resolving such fine structures with mesh-based solvers results either in poor element quality or high computational cost, due to the large element number needed to resolve the device sufficiently.  In order to overcome these difficulties, a popular class of numerical methods has been developed in the literature, which is based on the strategy of averaging or \enquote{upscaling}, i.e.~it approximates microscale models by macroscopic models which incorporate local informations about the unresolved microscale features. Possibly the most popular approach to achieve this is based on Representative Elementary Volumes (REVs) and spatial averaging \cite{quintard_two-phase_1988}. The idea of this approach can be summarised as follows, see e.g.~\cite{hornung_homogenization_1997}. Let $u$ be a real-valued function on a domain $\Omega$ which describes certain physical quantities with rapid spatial oscillations. To smoothen this function, one considers local averages 
in a small neighbourhood $V(x)$ about the point $x$ representative for the material --- the (local) REV. If the oscillations of $u$ reflect the behaviour of the physical quantity in question on a \enquote{microscale}, the averaged
function $\langle u \rangle$ is assumed to describe its properties on the (much larger) observation scale, i.e.~on the \enquote{macroscale}. Mathematical homogenisation is another method allowing to \enquote{upscale} differential equations \cite{hornung_homogenization_1997}. While the ratio of characteristic microscopic length and characteristic macroscopic length, $\varepsilon > 0$, is a small number determined by the geometry of the problem, the idea of homogenisation is to consider a family of functions $u_\varepsilon$ and to determine the limit 
\begin{equation}
    u = \lim_{\varepsilon \to 0} u_\varepsilon.
\label{eq:lim}
\end{equation}
This limit is then considered the upscaled unknown and the objective of the homogenisation procedure is to determine the system of (differential) equations satisfied by the limit $u$ and proving that \eqref{eq:lim} in fact holds.
In this context, the well-known Darcy law for flow problems, which was first introduced empirically by \cite{Darcy}, plays an important role, by giving a relation between the flow and the characteristics of the porous medium via the permeability tensor $\mathbf{K}$, i.e. 
\begin{equation*}
  v = \frac{1}{\mu} \mathbf{K} (f-\nabla p),
\end{equation*}
where $v$ is the Darcy velocity, $p$ the pressure, $f$ an external force and $\mu>0$ is the viscosity.
As a result of periodic homogenisation, the entries of the permeability tensor $\mathbf{K}$ are given by certain averages over the periodicity cells involving the solutions to auxiliary flow problems stated on these cells.
Analytical expressions for the permeabilities for an arrangement of cylinders (representing fibrous media with parallel cylindrical fibres) were derived in \cite{Boutin} using homogenisation in periodic media and the self-consistent method. 
It has to be noted that the empirical Darcy law does not rely on any periodicity assumption. 
Nonetheless, in what follows, we give a mathematical justification of the Darcy law in periodic porous media, which allows also to give an explicit formulation for the permeability tensor $\mathbf{K}$ in terms of rotation matrices and transversal and longitudinal permeabilities. 

An approach based on rotation matrices and analytical expressions for transversal and longitudinal permeabilities was employed for simulation of blood flow through coiled aneurysms in \cite{romero_bhathal_modeling_2023,romero_bhathal_towards_2023,Horvat2024}. However, in these works only the diagonal terms of the permeability matrix $\mathbf{K}$ were considered. The objective of this work is first to assess if the effective permeability can be modelled by certain explicit formulas for the transversal and longitudinal permeabilities, and to what extent the rotational transformation associated with analytical expressions for permeability allows to model effectively the permeability tensor of  anisotropic fibre-like microscopic geometries, with a particular emphasis on coiled aneurysms. The validation permeabilities are computed from expressions for the permeability based on solutions to cell problems (potentially with slight modifications, namely oversampling, see below) assuming that the fibre-like anisotropic media exhibit a quasi-periodic behaviour. In \cref{sec:permeability}, we propose a numerical validation of the permeability for anisotropic microscopic geometries, and apply it to a realistic coil in \cref{sec.TrueCoil}, while in \cref{sec:Numerical}, we use the modelled permeability matrix $\mathbf{K}$ to compute the flow in complex geometries using the Darcy equation.

\section{Derivation and numerical validation of the permeability for anisotropic microscopic geometries}
\label{sec:permeability}

In this section, we validate tensorial permeabilities resulting from anisotropic microstructures. As test case, we employ a coil packing which is one of the standard medical devices used  for  cerebral aneurysm treatment.  
\subsection{Flow in porous media}

Consider the Stokes problem in a heterogenous domain $\Omega_\varepsilon$.
The steady-state Stokes problem in the domain $\Omega_\varepsilon$ with homogeneous Dirichlet boundary condition is to find the velocity $u_\varepsilon: \Omega_{\varepsilon} \to \mathbb{R}^N$ and the pressure $p_\varepsilon:\Omega_{\varepsilon} \to \mathbb{R}$, solution to
\begin{equation}
    \begin{array}{rcll}
     - \mu \Delta v_\varepsilon + \nabla p_\varepsilon &= &f &\text{in}  \ \Omega_\varepsilon,\\
    \nabla \cdot v_\varepsilon &=& 0 &\text{in} \ \Omega_\varepsilon,\\
    v_\varepsilon &=& 0 & \text{on} \ \partial\Omega_\varepsilon,
    \end{array}
\label{eq:eps_probl}
\end{equation}
where $\mu>0$ is the viscosity and $f$ an applied force with the usual regularity. In the literature, some authors consider a different scaling of \eqref{eq:eps_probl}, replacing the velocity $v_\varepsilon$ by $v_\varepsilon = \varepsilon^2 v_\varepsilon$; however, this does not change the methodology nor the results, up to an $\varepsilon^2$ factor in certain places. To proceed, let $\Omega \subset\mathbb{R}^N$ be a regular bounded open set divided into a fixed solid part $\mathfrak{B}_\varepsilon$ and its complementary fluid part $\Omega_\varepsilon$.
We denote by $(\cdot,\cdot)$ the usual scalar product in $L^2(\Omega_\varepsilon)$; we use the same notation for vector-valued functions.
We introduce the classical velocity space $ V=[H^1_0(\Omega_\varepsilon)]^N = \{u \in [H^1(\Omega_\varepsilon)]^N {\rm such \ that} \ u \lvert_{\partial \Omega_\varepsilon}= 0 \}$ and pressure space $M = L^2_0(\Omega_\varepsilon) = \{q \in L^2(\Omega_\varepsilon) \ {\rm such \ that} \ \int_{\Omega_\varepsilon} q =0 \}$ as well as the bilinear forms $a:V \times V \to \mathbb{R}$ and $b:M \times V \to \mathbb{R}$,
\begin{equation*}
a(u, v) = \mu(\nabla u, \nabla v) , \qquad b(p,v) = -(p, \nabla \cdot v).
\end{equation*}
Then, assuming that $f \in [L^2(\Omega)]^N$, a weak formulation of the Stokes problem \eqref{eq:eps_probl} reads as follows: find $u_\varepsilon \in V$ and $p_\varepsilon \in M$ such that
\begin{equation}
\left \{
\begin{array}{rclrl}
 a(u_\varepsilon, v) + b(p_\varepsilon, v) &=& (f, v), 	&&  v \in V, \\
b(q,u_\varepsilon) &=& 0,  && q \in M .
\end{array}
\right.
\label{eq:WeakStokes}
\end{equation}
It is well-known that if $\Omega_\varepsilon$ is connected (to ensure that the zero-average condition of the pressure is sufficient to remove the undetermined constant of the pressure) there exists a unique weak solution to \eqref{eq:WeakStokes} \cite{girault_finite_1986}.

Denoting the unit cell (periodicity cell of unit side length) by $Y = (0,1)^N$, we assume that the domain $\Omega$ is perforated by a periodically arranged set of closed obstacles $\mathfrak{B}_\varepsilon$ of size $\varepsilon$, i.e.~$\mathfrak{B}_\varepsilon = \Omega \cap \cup_{k\in\mathbb{Z}} \varepsilon (\mathcal{O}+k)$, where $\mathcal{O}$ is the obstacle in the unit cell $Y$, so that $\Omega_\varepsilon = \Omega \setminus {\mathfrak{B}_\varepsilon}$.
The homogenisation limit of the Stokes equations, i.e.~finding the limit system satisfied by the limit of $(u_\varepsilon,p_\varepsilon)$ as $\varepsilon$ tends to zero, was first investigated by \cite{sanchez-palencia_non-homogeneous_1980, tartar:80, allaire_homogenization_1989}.
A review of these results can be found in \cite{hornung_homogenization_1997,CPS}.
There exists an extension of the solution $(v_\varepsilon, p_\varepsilon)$ of \eqref{eq:eps_probl} such that the velocity converges weakly in $[L^2(\Omega)]^N$ to $v$ and the pressure converges strongly in $L^2_0(\Omega)$ to $p$, where $(v, p)$ is the unique solution to the homogenised (Darcy) problem,
\begin{equation}
\left\{
    \begin{array}{rcll}
    v &=& \displaystyle \frac{1}{\mu} \mathbf{K} (f-\nabla p)& \text{in} \ \Omega,\\
    \nabla \cdot v &=& 0 & \text{in} \ \Omega,\\
    v \cdot\nu &=& 0 & \text{on} \ \partial\Omega,
    \end{array}
\right. 
\label{eq:DarcyEquations}
\end{equation}
with $\nu$ the outward unit normal vector to $\partial \Omega$.  Moreover, the sharp convergence rate $O(\sqrt{\varepsilon})$ for the energy norm of the difference between the Stokes and Darcy velocities was shown in \cite{Balazi}. Defining the cell problems for $i=1, \dotsc, N$ by finding $w_i : Y \to \mathbb{R}^N$ and $\pi_i : Y \to \mathbb{R}$, $Y$-periodic and solutions to 
\begin{equation}
\left\{
    \begin{array}{rcll}
     - \Delta w_i +\nabla \pi_i &=& e_i  & \text{in} \ Y\setminus{\mathcal{O}},\\
    \nabla \cdot w_i &=& 0 & \text{in} \ Y\setminus{\mathcal{O}},\\
    w_i&=&0 & \text{in} \ \mathcal{O},
    \end{array}
\right.
\label{eq:cellpb}
\end{equation}
it follows that the permeability tensor $\mathbf{K}$ is given by
\begin{equation}
    \mathbf{K}_{ij} = \frac{1}{\lvert Y \rvert} \int_{Y\setminus{\mathcal{O}}} \nabla w_i \cdot \nabla w_j \,\mathrm{d}y =  \frac{1}{\lvert Y \rvert} \int_{Y\setminus{\mathcal{O}}} w_i \cdot e_j \, \mathrm{d}y, \quad 1 \leq i,j\leq N,
\label{eq:MatrixPermeability}
\end{equation}
where the last equality is obtained by integration by parts, and where $\{e_j\}_{j=1,\dotsc, N}$, is the canonical basis of $\mathbb{R}^N$.
It has to be noted that such an approach holds beyond the periodic case, for example in the case of perforated domains defined as a local perturbation of the periodically perforated domain, see e.g.~\cite{Wolf}, or even for non-periodically evolving microstructures \cite{wiedemann2024,DarcyMemory2024}.
Thus, various engineering applications involving flow through complex porous media require accurate knowledge of the permeability tensor $\mathbf{K}$.
However, the computation of $\mathbf{K}$ is often non-trivial owing to the heterogeneous and anisotropic nature of such materials.
The modelling of the permeability tensor therefore represents a significant area of research in engineering and geosciences.
A common engineering approach is to express $\mathbf{K}$ as a function of the porosity, which is the volume fraction of the fluid phase within the medium.
A benchmark and comparison of different methods to estimate the permeability of regular porous structures can be found in \cite{benchmark} from the well-known Kozeny--Carman relation \cite{Kozeny, Carman} to more advanced approaches such as pore-scale simulations, although limited to a relatively simple geometry, that is circular obstacles in two dimensions and cylinders in three dimensions.

In this present work, motivated by the study of flow in coiled aneurysms which exhibit a fibre-like structure, we pay particular attention to the derivation of the permeabilities of \cite{Boutin}.
Indeed, analytical permeabilities for an arrangement of parallel cylinders (representing fibrous media) were derived in \cite{Boutin} using homogenisation of periodic media and the self-consistent method.
Let $R$ be the radius of the cylinders, let $\phi$ be the porosity (the volume fraction of the fluid phase) of the fibrous medium and let $\rho=1-\phi$ (the volume fraction of the solid phase).
Using the subscripts $\perp$ and $\parallel$ for the transversal and longitudinal directions, respectively, and the subscripts p and v for the methods based on statically continuous fields and  kinematically continuous fields, which act as upper and lower bounds for the actual permeabilities, respectively, it was shown that the transversal and longitudinal permeabilities are given by
\begin{equation}
    K_{\perp, \rm p} = -\frac{R}{8\sqrt{\rho}}\left(\log(\rho) + \frac{1-\rho^2}{1+\rho^2}\right), \quad  K_{\perp,\rm v} = -\frac{R}{8\sqrt{\rho}}\left(\log(\rho)+ \frac{2 (1-\rho)}{1+\rho} \right),
   \label{eq:Kperp} 
\end{equation}
and
\begin{equation}
       K_{\parallel,\rm p}= -\frac{R^2}{4\rho}\left(\log(\rho)+\frac{(1-\rho)(3-\rho)}{2} \right),  \quad  K_{\parallel,\rm v} = -\frac{R^2}{4\rho}\left(\log(\rho) + \frac{2 (1-\rho)}{1+\rho} \right).
    \label{eq:Klong}
\end{equation} 
As noted in \cite{Boutin}, in case of transverse flow, the gradient is perpendicular to the cylinder axis, and the problem is fully determined in a $(r, \theta)$-plane of zero thickness in parallel direction.
Then, the intrinsic permeability is related to a flow per length perpendicular to the pressure gradient in the $(r, \theta)$-plane, instead of a usual flow per area.
This is the reason why $K_{\perp,\rm v}$ and $K_{\perp,\rm p}$ in \eqref{eq:Kperp} have the dimension of a length instead of a squared length. In order to recover the usual dimension of the intrinsic permeability, we have to consider a unit length of material in the cylinder direction.
However, we may replace this length by an arbitrary length, such as $R/\sqrt{\rho}$, which leads to 
\begin{equation}
    K_{\perp,\mathrm{p}}^\prime  = -\frac{R^2}{8\rho}\left(\log(\rho) + \frac{1-\rho^2}{1+\rho^2}\right), \quad 
    K_{\perp,\mathrm{v}}^\prime = -\frac{R^2}{8\rho}\left(\log(\rho) + \frac{2 (1-\rho)}{1+\rho}\right).
    \label{eq:Kperp_mod}
\end{equation}
Given these dimensional considerations, we henceforth restrict our attention to the permeabilities defined in \eqref{eq:Klong} and \eqref{eq:Kperp_mod}. 

For a given fibre (modelled as a cylinder) within a unit cell, let $\mathbf{R}_\theta$ be the rotation matrix which maps the vertical axis $e_z$ to the fibre direction. Now, assuming that the direction given by the angle $\theta$ represents the principal direction of the flow, we propose to estimate the permeability tensor as
\begin{equation*}
    \mathbf{K} = \mathbf{R}_\theta \begin{pmatrix}
        K_{\perp} & 0 & 0 \\
        0 & K_{\perp} & 0  \\
        0 & 0  & K_{\parallel}
    \end{pmatrix} \mathbf{R}_\theta^\top, 
\end{equation*}
where the superscript $\top$ denotes the transpose, and where we propose to model the permeabilities $K_{\parallel}$ and $K_{\perp}$ using \eqref{eq:Klong} and \eqref{eq:Kperp_mod}.
 
For the sake of comparison, in this paper, we also consider an isotropic Kozeny--Carman permeability, i.e.~$\mathbf{K}_{\rm iso}=  K_{\rm iso}\mathbf{I}_3$, with $\mathbf{I}_3$ the identity matrix of $\mathbb{R}^3$ and $K_{\rm iso}$  given by \cite{Schulz} as (using the same notations as before)
\begin{equation}
    K_{\rm iso} = \frac{(2R)^2 (1-\rho)^3}{150 \rho^2}.
\label{eq:Kiso}
\end{equation}

\subsection{Preliminary results}
In this section, we introduce some theoretical results linked to scaling and rotation of the unit cell. 

\begin{lemma}\label{lemma:over}
For any real number $\kappa>0$, consider a scaled unit cell $Y^\kappa =(0, \kappa)^N$.
Introducing the scaled variable $z:=\kappa y$, $y \in Y$, the scaled permeability matrix satisfies
\begin{equation*}
    \mathbf{K}^\kappa = \kappa^2 \mathbf{K}, 
\end{equation*}
where  $\mathbf{K}$ is given by \eqref{eq:MatrixPermeability} and $\mathbf{K}^\kappa$ is defined by the same relation but replacing $Y$ by $Y^\kappa$ in \eqref{eq:cellpb} and \eqref{eq:MatrixPermeability}.
\end{lemma} 

\begin{proof}
When computing the homogenised coefficients in the domain $Y^\kappa$, one must average over the periodicity cell $Y^\kappa$, whose volume is $\kappa^N$.
Let $\omega_i^\kappa$ denote the cell solution associated with the periodicity cell $Y^\kappa$.
By introducing the change of variable $z:=\kappa y$, $y \in Y$, it follows that the rescaled function satisfies $\omega_i^\kappa(z) = \kappa \omega_i\left(\frac{z}{\kappa}\right)$, where $\omega_i$ denotes the solution to the reference cell problem defined on the unit cell. 
\end{proof}

\begin{remark}
The scaling law of \Cref{lemma:over} also applies to the next order Stokes cell problems, namely $\gamma_{ij}^\kappa(z) = \kappa^2 \gamma_{ij}\left(\frac{z}{\kappa}\right)$, $1 \leq i,j \leq N$.
\end{remark}

\begin{lemma}\label{lemma:preservation_symmetry}
Let $Y =(0,1)^N$ be the unit cell and $\mathcal{O}$ be a unit cylinder oriented along an axis $e_i$, $1\leq i \leq N$, of $\mathbb{R}^N$. Furthermore, let 
\begin{equation*}
    \mathfrak{B}_\varepsilon = \bigcup_{k \in \mathbb{Z}^N} \varepsilon(\mathcal{O}+k)
\end{equation*}
be the periodic arrangements of obstacles and let $\mathbf{R}_\theta\in SO(N)$ be a rotation.
If and only if there exists a scalar $\lambda>0$ such that $\lambda \mathbf{R}_\theta \in \mathbb{Z}^{N \times N}$, then there exists an invertible integer matrix $\mathbf{M} \in \mathbb{Z}^{N \times N}$ such that 
\begin{equation*}
    \mathbf{R}_\theta(\mathbb{Z}^N) = \frac{1}{\lambda}\mathbf{M} (\mathbb{Z}^N)
\end{equation*}
and there holds the identity
\begin{equation*}
    \mathbf{R}_\theta(\mathfrak{B}_\varepsilon) = \bigcup_{\ell \in \frac{1}{\lambda}\mathbf{M}(\mathbb{Z}^N)} \varepsilon(\mathbf{R}_\theta(\mathcal{O})+\ell) .
\end{equation*}
\end{lemma}

\begin{proof}
If $\lambda \mathbf{R}_\theta \in \mathbb{Z}^{N \times N}$ for some $\lambda>0$, then setting $\mathbf{M} := \lambda \mathbf{R}_\theta$ gives an invertible integer matrix such that 
\begin{equation*}
  \mathbf{R}_\theta(\mathbb{Z}^N) = \frac{1}{\lambda} \mathbf{M}(\mathbb{Z}^N).
\end{equation*}
Consequently, the rotation of the periodic arrangement satisfies
\begin{equation*}
    \mathbf{R}_\theta(\mathfrak{B}_\varepsilon) = \bigcup_{\ell \in \frac{1}{\lambda}\mathbf{M}(\mathbb{Z}^N)} \varepsilon (\mathbf{R}_\theta(\mathcal{O}) + \ell).
\end{equation*}
Conversely, if such an integer matrix $\mathbf{M}$ exists, then $\lambda \mathbf{R}_\theta = \mathbf{M} \in \mathbb{Z}^{N \times N}$, establishing the equivalence.
\end{proof}

\begin{remark}
\cref{lemma:preservation_symmetry} is essentially a lattice preservation property under a rotation of the cylinders and ensures that the unit cell is changed by a single uniform scaling $\lambda$ and defines $\frac{1}{\lambda}\mathbf{M}(\mathbb{Z}^N)$ as the new lattice. 
Intuitively, this condition restricts $\mathbf{R}_\theta$ to \enquote{rational rotations} that map the integer lattice $\mathbb{Z}^N$ onto a scaled copy of itself, ensuring the periodic structure is preserved.
For the rotation of an angle $\theta$ in a two-dimensional plane, this condition is equivalent to the existence of a $\lambda>0$ such that $\lambda \cos(\theta)$ and $\lambda \sin(\theta) \in \mathbb{Z}$.
So all angles $\theta$ satisfying this condition are exactly those for which $\tan(\theta)$ is a rational number.
For example, with $\theta=\frac{\pi}{4}$ with $\lambda = \sqrt{2}$ gives
\begin{equation*}
    \lambda \mathbf{R}_{\pi/4} = \sqrt{2}\begin{pmatrix}
        \frac{1}{\sqrt{2}} & -\frac{1}{\sqrt{2}} \\
        \frac{1}{\sqrt{2}} & \frac{1}{\sqrt{2}}
    \end{pmatrix}
= \begin{pmatrix}
        1 & -1 \\
        1 &1
    \end{pmatrix} \in \mathbb{Z}^{2\times 2}.
\end{equation*}
\label{rem:Rational}
\end{remark}

From \eqref{eq:MatrixPermeability}, it follows that the tensor $\mathbf{K}$ is symmetric and positive semi-definite.
Consequently, there exists an orthogonal matrix $\mathbf{O}$ and a diagonal tensor $\mathbf{K}_{\rm eff}$, such that
\begin{equation}
    \mathbf{K} = \mathbf{O} \mathbf{K}_{\rm eff}\mathbf{O}^\top.
\label{eq:O}
\end{equation}
$\mathbf{K}_{\rm eff}$ is called the effective permeability tensor.
In the following lemma, we give a characterisation of the matrix $\mathbf{O}$ for a particular case.

\begin{lemma}
Let $\mathbf{R}_\theta$ be a rotation satisfying \cref{lemma:preservation_symmetry}, i.e.~there exists a scalar $\lambda>0$ such that $\lambda \mathbf{R}_\theta \in \mathbb{Z}^{N \times N}$, and consider the rotated variable $z := \mathbf{R}_\theta y$.
Let us consider an inclined fibre $\tilde{\mathcal{O}}$ inside the unit cell $Y$, i.e.~$\tilde{\mathcal{O}} = \mathbf{R}_\theta(\mathcal{O})$, where $\mathcal{O}$ is the reference fibre oriented along the vertical axis $e_N$.
Then, there exists a diagonal tensor $\mathbf{K}_{\rm eff}$ such that
\begin{equation}
   \mathbf{K}  = \mathbf{R}_\theta \mathbf{K}_{\rm eff}\mathbf{R}_\theta^\top, 
 \label{eq:Rotation}   
\end{equation}
that is $\textbf{O}= \mathbf{R}_\theta$, where the matrix $\textbf{O}$ is defined in \eqref{eq:O}.
\label{lemma:Rotation}
\end{lemma}

\begin{proof}
Equation \eqref{eq:Rotation} follows directly from the classical transformation of the cell problems \eqref{eq:cellpb} under scaling and rotation.
\end{proof}

\begin{remark}
Although \eqref{eq:Rotation} is strictly valid only for certain rotation angles due to \cref{lemma:preservation_symmetry}, numerical experiments indicate that it remains a reliable, albeit approximate, description even for angles which break the cell symmetry.
Up to a reordering of the eigenvalues of $\mathbf{K}_{\rm eff}$, the effective permeability tensor can be identified with a scaled version of the reference permeability $\mathbf{K}_0$, where $\mathbf{K}_0$ corresponds to the case in which the fibres are aligned with the vertical axis.
\end{remark}

\subsection{Numerical validation of the permeabilities}
This section has two main objectives.
The first is to validate the use of the analytical expressions for the permeabilities \eqref{eq:Klong} and \eqref{eq:Kperp_mod} using numerical experiments.
The second objective is motivated by the study of the actual coil geometry, in which the coil within the REVs, namely the unit cell, do not align on opposite faces of the unit cube, preventing the use of periodic boundary conditions.
To address this, we embed the REV within a larger computational cell (oversampling), enabling the application of periodic boundary conditions. Consequently, we propose a calibration procedure to determine the true permeability corresponding to the oversampled configuration.

\subsection*{Experiment 1} 
In order to validate the use of analytical expressions  \eqref{eq:Klong} and \eqref{eq:Kperp_mod} for the permeabilities, we perform two studies.
First, we consider configurations within a unit cell $Y$ containing between 1 and 9 cylinders, each of radius $R=0.15$ and aligned with the  $e_z$-axis. This choice of radius is motivated by its close correspondence to the coil radius analysed in subsequent sections.  Each configuration is built from a regular $3\times 3$ grid of potential cylinder locations, giving nine subsquares in total. For a given number of cylinders $n \in \{1,\dots,9 \}$, a configuration is obtained by randomly choosing $n$ subsquares from the grid. Each selected subsquare contains exactly one cylinder. For each selected subsquare, the position of the cylinder centre is chosen randomly (uniform distribution), subject to the constraint that the cylinder boundary keeps a distance from the boundary of the subsquare of at least $\delta=5\cdot 10^{-3}$. To ensure a representative sampling of the configuration space, 40 random configurations are generated for each value of $n$ (except for $n=1$, in which case only a single configuration is enough), noting that no further visible changes have been observed any more by the inclusion of the last samples. For each configuration, the analytical permeabilities  \eqref{eq:Klong}, \eqref{eq:Kperp_mod} and \eqref{eq:Kiso}, are compared with the entries of the computed effective permeability tensor. To this end, the cell problems \eqref{eq:cellpb} are solved using the Finite Element method with a Taylor–Hood $\mathbb{P}_2$–$\mathbb{P}_1$ discretization on a mesh sufficiently fine to resolve the geometry. The computations are carried out in FreeFEM \cite{FreeFEM}, a finite element software, with careful verification of mesh convergence to ensure the accuracy of the numerical results. The permeability tensor is then evaluated using \eqref{eq:MatrixPermeability}. The resulting permeability matrix $\mathbf{K}$ is either diagonal, in which case $\mathbf{K}_{\rm eff} =\mathbf{K}$, or symmetric positive semi-definite. In the latter case, $\mathbf{K}_{\rm eff}$ is computed as in \eqref{eq:O}. Finally, we obtain a tensor of the form 
\begin{equation*}
    \mathbf{K}_{\rm eff} \approx \begin{pmatrix}
        K_{\perp, \rm exp} & 0 & 0 \\
        0 & K_{\perp, \rm exp} & 0  \\
        0 & 0  & K_{\parallel, \rm exp}
    \end{pmatrix},
\end{equation*}
where, as previously, the subscripts $\perp$ and $\parallel$ denote the transversal and longitudinal permeability, respectively (we use the same notation for simplicity, although the two transversal components may differ). The subscript \enquote{exp} denotes the experimentally obtained permeabilities, emphasising their numerical (as opposed to analytical) aspect. We then average the computed permeabilities over all realisations. Second, we repeat the same procedure, but considering only one cylinder of different radii, oriented along the $e_z$-axis, in a unit cell (in this case, the computed permeability tensor is diagonal). The resulting effective permeabilities are presented in \cref{fig:PermBoutin0} and in \cref{fig:PermBoutin01}, in which the region of interest for the application to a coil in an aneurysm is shaded in grey.

\begin{figure}[tbp]
    \centering
    \begin{subfigure}[b]{0.49\textwidth}
        \include{figures/PLOTS/PermBoutinPaper_TRANS_RAND.tex}
    \end{subfigure}
    \begin{subfigure}[b]{0.49\textwidth}
        \include{figures/PLOTS/PermBoutinPaper_LONG_RAND.tex}
    \end{subfigure}
\caption{Comparisons between the averaged computed effective permeabilities and analytical permeabilities \eqref{eq:Klong}, \eqref{eq:Kperp_mod} and~\eqref{eq:Kiso} for arrangements of cylinders of radius $R=0.15$. The standard deviation of the averaged permeabilities are also represented. The number of cylinders vary from 1 to 9, from the right to the left. The region of interest for the application to a coil in an aneurysm is shaded in grey.}
  \label{fig:PermBoutin0}
\end{figure}

\begin{figure}[tbp]
    \centering
    \begin{subfigure}[b]{0.49\textwidth}
        \include{figures/PLOTS/PermBoutinRadius_TRANS}
    \end{subfigure}
    \begin{subfigure}[b]{0.49\textwidth}
        \include{figures/PLOTS/PermBoutinRadius_LONG.tex}
    \end{subfigure}
\caption{
Comparisons between computed permeabilities and analytical permeabilities \eqref{eq:Klong}, \eqref{eq:Kperp_mod} and \eqref{eq:Kiso} for one cylinder of different radii, ranging from $R=0.1$ to $R=0.4$.}
  \label{fig:PermBoutin01}
\end{figure}

We observe that the experimental permeabilities are well approximated by the analytical expressions \eqref{eq:Klong} and \eqref{eq:Kperp_mod} in general, with minor deviations in the longitudinal permeability for the case when the number of cylinders increases, for which the one based on statically continuous fields is closer to the experimental permeabilities compared to ones based on kinematically continuous fields or isotropy, both of which almost coincide for all values of porosity. This deviation can be attributed to several reasons. First, it should be noted that the permeability expressions derived in \cite{Boutin} are based on an idealised arrangement of parallel cylinders that completely fill space, and the bounds are obtained by interpreting the periodic cell problem as an equivalent inclusion in a Darcy medium using the self-consistent method. This approach provides rigorous estimates of the effective permeability while accounting for the influence of microstructural geometry.  Second, the expression we use to compute permeability, namely \eqref{eq:MatrixPermeability}, is rigorous only in the limit of a periodic arrangement of cylinders with vanishing size and period. While it would be possible, as is common in permeability modelling, to introduce a heuristic correction factor to match the analytical model to numerical results more closely, we do not pursue this direction but follow a different avenue in Section \ref{sec:Numerical}, noting that we are not interested in the precise values of the permeabilities but the resulting fluid flow in particular.
Nevertheless, this experiment confirms the usefulness of the widely used model of Boutin~\cite{Boutin} for fibre-like media to our case.

\subsection*{Experiment 2}

Afterwards, as explained previously, to investigate the actual coil geometry (see \cref{sec.TrueCoil}), where the obstacles within the REVs do not align on the opposite faces of the unit cube, thus preventing the use of periodic boundary conditions, we embed the REV within a larger computational cell to enable their application. We then propose a procedure to establish the relationship between the true permeability and that obtained from the oversampled configuration. To find this relation, we repeat the first two studies, but now the cylinders are embedded in a larger cell of size $\kappa=1.1$, denoted by $Y^\kappa$. In this configuration, to compute the permeability tensor $\mathbf{K}^\kappa$, we solve the cell problems \eqref{eq:cellpb} in the cell $Y^\kappa$ and compute the permeability as 
\begin{equation}
    \mathbf{K}^\kappa_{ij} = \frac{1}{\lvert Y^\kappa \rvert} \int_{Y^\kappa} w_i \cdot e_j \, \mathrm{d}y, \quad 1\leq i,j \leq N.
\label{eq:MatrixPermeabilityOver}
\end{equation}
As previously, the resulting permeability matrix $\mathbf{K}^\kappa$ is either diagonal, in which case $\mathbf{K}_{\rm eff}^\kappa =\mathbf{K}^\kappa$, or symmetric positive semi-definite. In the latter case, $\mathbf{K}_{\rm eff}^\kappa$ is computed as in \eqref{eq:O}. Finally, we obtain permeability matrix is thus of the form 
\begin{equation*}
    \mathbf{K}_{\rm eff}^\kappa \approx \begin{pmatrix}
        K_{\perp, \rm exp}^\kappa & 0 & 0 \\
        0 & K_{\perp, \rm exp}^\kappa & 0  \\
        0 & 0  & K_{\parallel, \rm exp}^\kappa
    \end{pmatrix}.
\end{equation*}

At this stage, the objective is to establish a relationship between $\mathbf{K}^\kappa_{\rm eff}$ (the oversampled permeability) and $\mathbf{K}_{\rm eff}$ (the true permeability in the unit cell) for this specific setting, namely an REV composed of an arrangement of cylinders. Building on \cref{lemma:over}, we introduce the factor $\kappa^2$ to account for the scaled domain. Unlike in \cref{lemma:over}, where the obstacles are also scaled by 
$\kappa$, in our oversampling approach the obstacles remain unchanged, which requires an additional correction. Heuristically, for moderate values of the oversampling parameter, we find that the suitable scaling for this particular case is given by 
\begin{equation}
    \widetilde{\mathbf{K}}_{\rm exp} \approx \frac{\phi}{\kappa^2} \mathbf{K}^\kappa.
\label{eq:relation_over}
\end{equation} 
The results from this ``oversampling and then scaling'' procedure are also shown in \cref{fig:PermBoutin0} and \cref{fig:PermBoutin01}.  First, in the case where only a single cylinder is present within the cell, we observe that, under this scaling, the \enquote{oversampled} and \enquote{unit} permeabilities are in good agreement. A similarly good agreement is observed for the longitudinal permeability in configurations with several cylinders.
However, the situation becomes more complex for the transversal permeability in the presence of multiple cylinders. In this case, good agreement is maintained only for porosities larger than 0.6 (corresponding to fewer than six cylinders). We note also that the variance of the \enquote{unit} transversal permeability is larger for the oversampled one. The deterioration of the agreement for smaller porosities can be explained by the oversampling procedure, which introduces excessive void space compared to the amount initially present in the unit cell. In particular, additional void regions are added upstream and downstream along the cylinder direction, which are not part of the original configuration. If this regime is of interest, an alternative scaling would be required in this region. However, as will be shown below, this regime is not relevant for our purposes, and we therefore do not pursue further adjustments to the scaling in this range. Finally, it should be noted that in the proposed scaling, the factor $\kappa^2$ is supported by homogenisation theory, whereas the factor $\phi$ is configuration-dependent and may require adjustment for other geometries.

\subsection*{Experiment 3}
Lastly, we validate \cref{lemma:Rotation} numerically.
To this end, we compute the permeability, using again \eqref{eq:MatrixPermeability}, in a rotated cell (of unit side length) obtained by rotation of the whole periodic pattern by an angle $\theta$, which is \enquote{rational} in the sense of \cref{rem:Rational}, about the $e_z$-axis (note that we use different cylinder radii for the different cases for ease of generation of the unit cell).
As result, we obtain a semi-definite positive permeability matrix $\mathbf{K}$, which we diagonalise as $\mathbf{K}= \mathbf{O}\mathbf{K}_{\rm eff}\mathbf{O}^\top$, with $\mathbf{O}$ an orthogonal matrix.
The objective is now to compare the obtained matrix $\mathbf{O}$ with the theoretical rotation matrix of the angle $\theta$ about the $e_z$-axis, denoted by $\mathbf{R}_\theta$. Numerically, we find that $\mathbf{O}$ is of the form $\begin{psmallmatrix}
    0 & -c & s \\
    0 &  s & c \\
    1 & 0 &  0 
\end{psmallmatrix}$.
A comparison between $c$, $s$ and $\cos(\theta)$, $\sin(\theta)$, respectively, is performed. The corresponding results are presented in \cref{table:Rotations}.

\begin{table}[htpb]
    \caption{Validation of \cref{lemma:Rotation} for different \enquote{rational rotations}.}
    \label{table:Rotations}
    \centering
    \begin{tabular}{ccccc}
    \toprule
    \multirow{2}{*}{Angle of rotation $\theta$} & \multicolumn{2}{c}{Experimental values} & \multicolumn{2}{c}{Theoretical values} \\
                                                & $c$ & $s$                               & $\cos(\theta)$ & $\sin(\theta)$ \\
    \midrule
    $\arctan(1)$ & $0.707$ & $0.707$ & $1/\sqrt{2} \approx 0.707$ & $1/\sqrt{2} \approx 0.707$ \\[1ex]
    $\arctan(\frac{4}{5})$ & $0.780$ & $0.624$ & $5/\sqrt{41} \approx 0.780$ & $4/\sqrt{41} \approx 0.624$ \\[1ex]
    $\arctan(\frac{3}{4})$ & $0.800$ & $0.599$ & $4/5 = 0.8$ & $3/5 = 0.6$ \\[1ex]
    $\arctan(\frac{1}{2})$ & $0.894$ & $0.447$ & $2/\sqrt{5} \approx 0.894$ & $1/\sqrt{5} \approx 0.447$ \\[1ex]
    $\arctan(\frac{1}{4})$ & $0.970$ & $0.242$ & $4/\sqrt{17} \approx 0.970$ & $1/\sqrt{17} \approx 0.242$ \\[1ex]
    $\arctan(\frac{1}{5})$ & $0.980$ & $0.196$ & $5/\sqrt{26} \approx 0.980$ & $1/\sqrt{26} \approx 0.196$ \\
    \bottomrule
    \end{tabular}
\end{table}

After validating the analytical expressions for the permeabilities and calibrating the oversampling approach, the next section focuses on the actual geometry of a coil inside an aneurysm.
In this context, we employ the relation \eqref{eq:relation_over} to determine the true permeability from the oversampled configuration.

\section{Application to a coil in an aneurysm}

\label{sec.TrueCoil}
In this section, we investigate the permeability of a coil packing.
To this end, we compute the permeability of several REVs using homogenisation theory and compare the results with the analytical permeability expressions given in \eqref{eq:Klong} and \eqref{eq:Kperp_mod}.
The goal is to assess whether the permeability of the coil packing can be accurately described by these analytical models.
In the present work, as benchmark example, we consider the coil presented in \cref{fig:coil}.
It was generated using a mechanical simulation of the insertion process in a cerebral aneurysm, as described in \cite{holzberger2024comprehensive}. 
In the simulation, we first insert a stiffer framing coil which provides structural stability.
Subsequently, a soft filling coil is inserted into the cage formed by the framing coil, to increase packing density, which often correlates to better occlusion of the aneurysm.
The parameters of the inserted coils can be found in \cref{tab:coil_parameters} and closely correspond to those of coils commonly used in practise. With these parameters, we reach a packing density of approximately 30\% in the aneurysm.

\begin{figure}[H]
\centering
\includegraphics[width=0.2\linewidth]{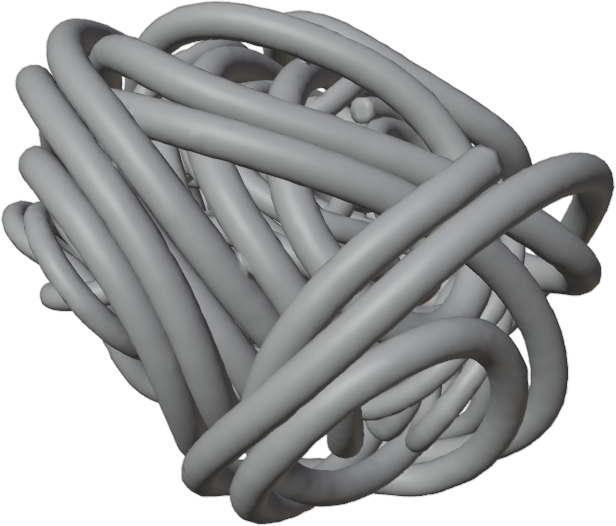}
\caption{The coil under study.}
\label{fig:coil}
 \end{figure}

\begin{table}[htpb]
    \caption{Parameters of the framing and filling coils.}
    \label{tab:coil_parameters}
    \begin{tabular}{lccc}
    \toprule
    Parameter & Unit & Framing Coil & Filling Coil \\
    \midrule
    Insertion velocity         & [mm/s] & 30   & 30   \\
    Internal wire radius       & [mm]   & 0.0381 & 0.025 \\
    Coil radius $R_{\rm coil}$ & [mm]   & 0.1778 & 0.127  \\
    Coil length                & [mm]   & 150  & 200 \\
    Tertiary coil diameter $D$ & [mm]   & 4.6  & 1.75 \\
    Shape (in $D$)             & ---    & Complex & Complex \\
    \bottomrule
    \end{tabular}
\end{table}

At this stage, the idea is to select several REVs within the coil of \cref{fig:coil}. In this work, the coil is provided as a surface mesh in \texttt{.obj} format with quadrilateral elements. Our goal is to extract multiple REVs and generate corresponding volume meshes in a format compatible with FreeFEM (the finite element software used for the permeability computation), namely the \texttt{.mesh} (Medit) format. The main steps of the REV generation process are summarised in \cref{alg:rev_generatiom} (see \cref{app.Algo}). 

From the coil geometry illustrated in \cref{fig:coil}, a total of 12 representative elementary volumes (REVs) of unit size have been selected.
The complete set of selected REVs is presented in \cref{fig:AllREVs}.  These REVs have been sampled from different regions of the packing: some are located near the boundary, where the porosity is relatively high, while others are taken from the interior, which is characterised by lower porosity. Also, we include some REVs from transitional regions which exhibit an intermediate porosity. In more complex porous media, it would be necessary to have a finer discretisation of the porosity from transitional zones to represent spatial variability completely.  However, in the present case, the clear distinction between high and low-porosity regions justifies the chosen sampling strategy, which adequately captures the dominant range of porosity within the system, even with a limited number of samples.

\begin{figure}[htbp]
    \centering
    
      \begin{subfigure}[b]{0.2\textwidth}
        \includegraphics[width=0.9\linewidth]{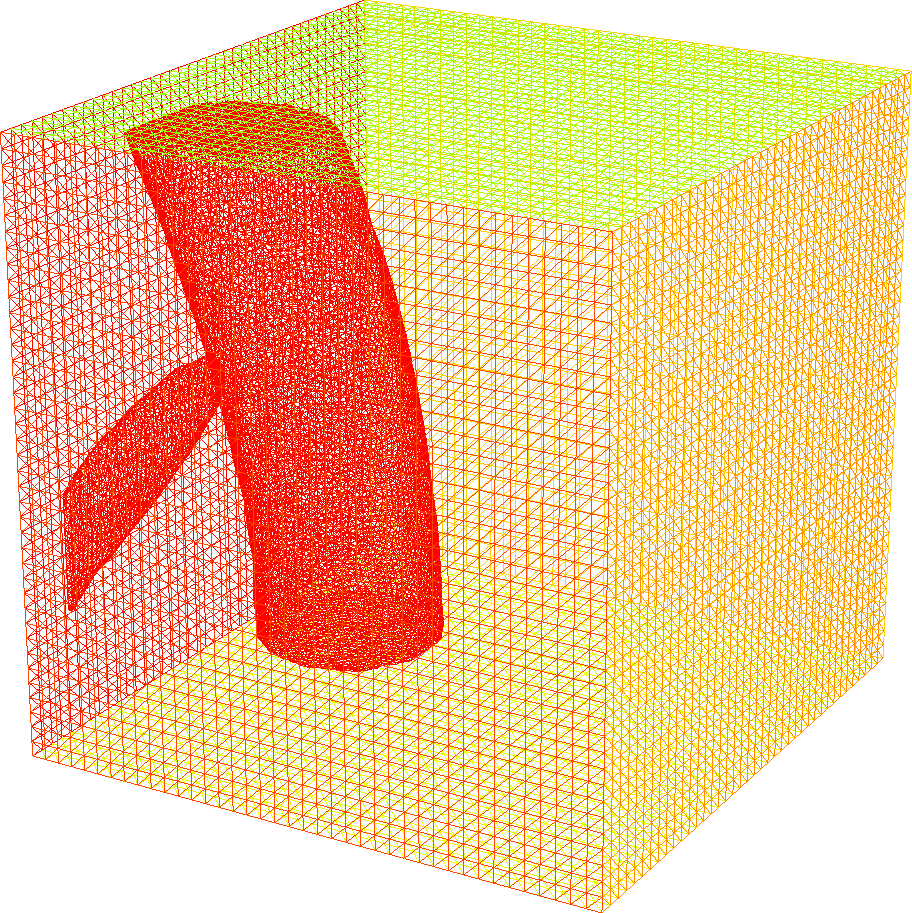}
        \caption{ \tiny REV 1 ($\phi \approx 0.91$).}
    \end{subfigure}\hfill
     \begin{subfigure}[b]{0.2\textwidth}
        \includegraphics[width=0.9\linewidth]{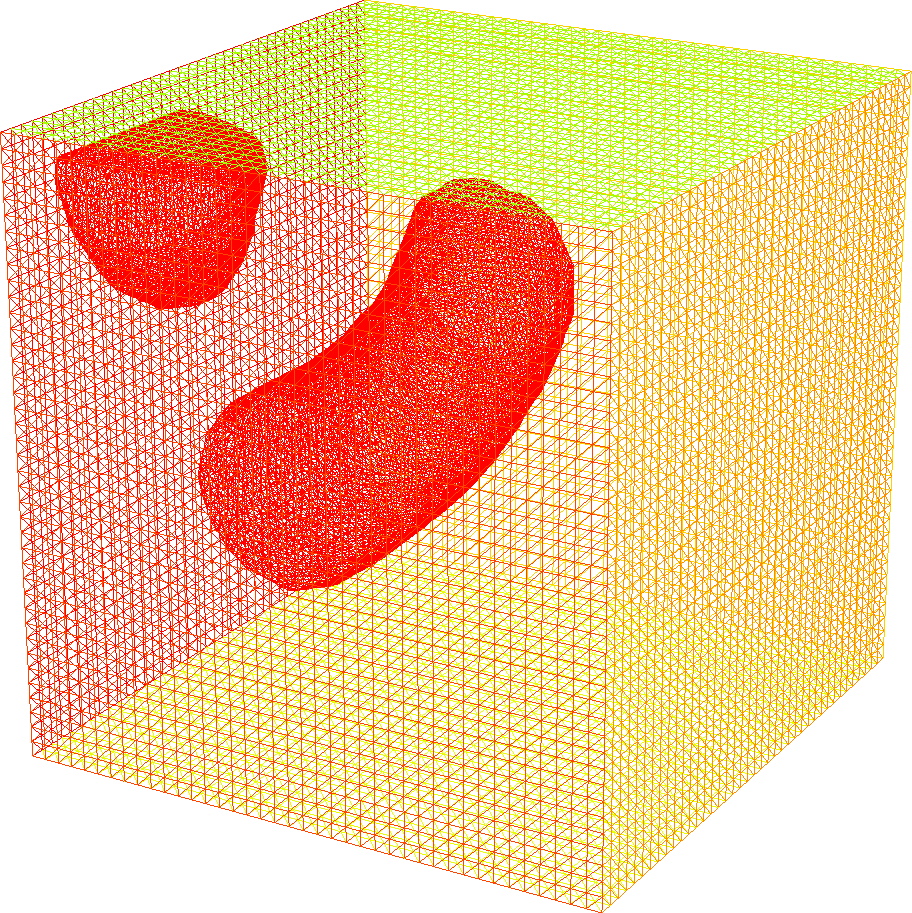}
        \caption{\tiny REV 2 ($\phi \approx 0.89$).}
    \end{subfigure}\hfill
     \begin{subfigure}[b]{0.2\textwidth}
        \includegraphics[width=0.9\linewidth]{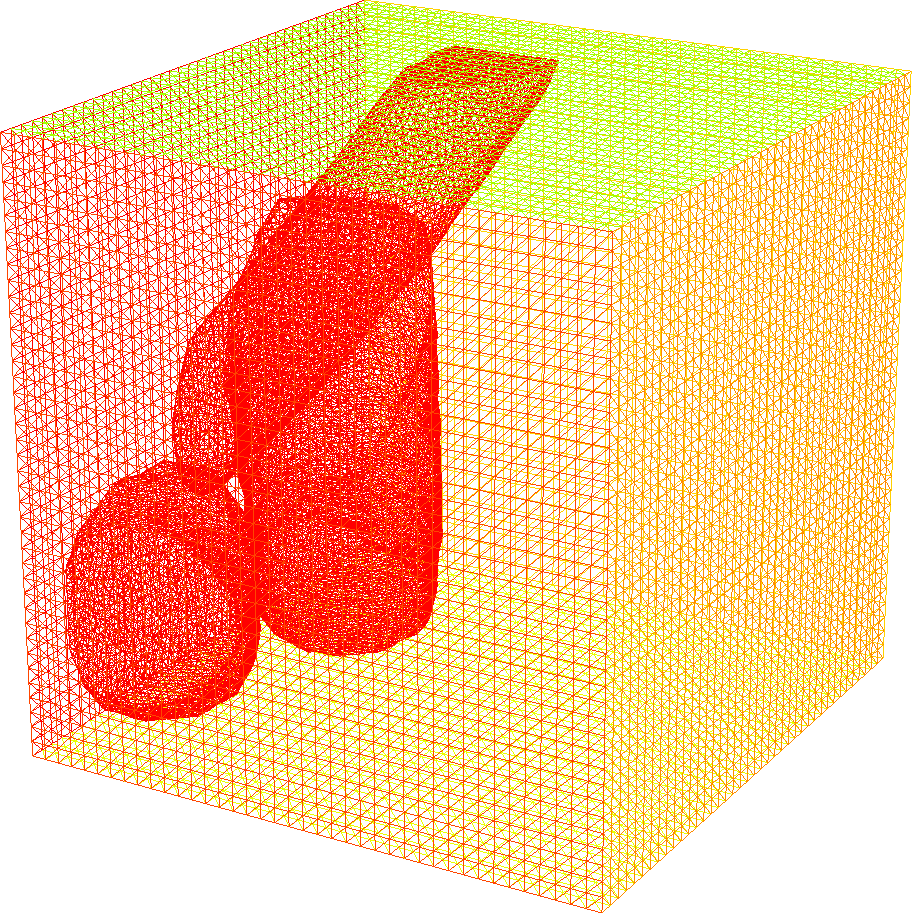}
        \caption{\tiny REV 3 ($\phi \approx 0.84$).}
    \end{subfigure}\hfill
   \begin{subfigure}[b]{0.2\textwidth}
        \includegraphics[width=0.9\linewidth]{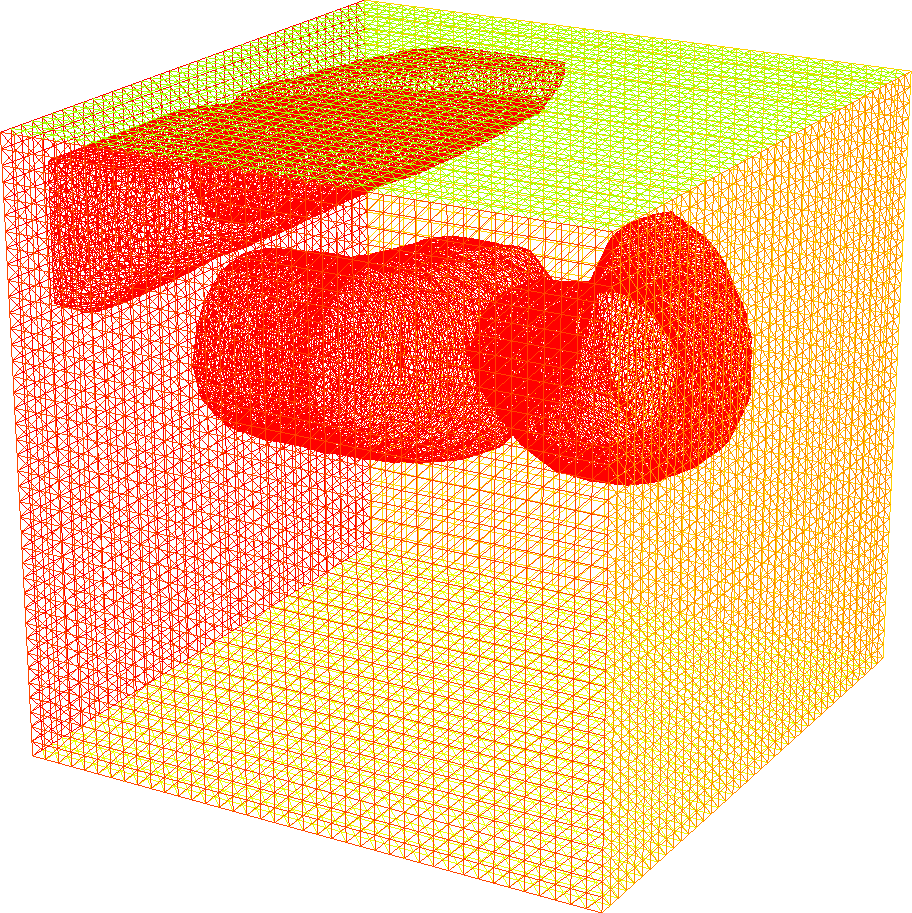}
        \caption{\tiny REV 4 ($\phi \approx 0.80$).}
    \end{subfigure}
\hspace{0.5\linewidth}
      
      \begin{subfigure}[b]{0.2\textwidth}
        \includegraphics[width=0.9\linewidth]{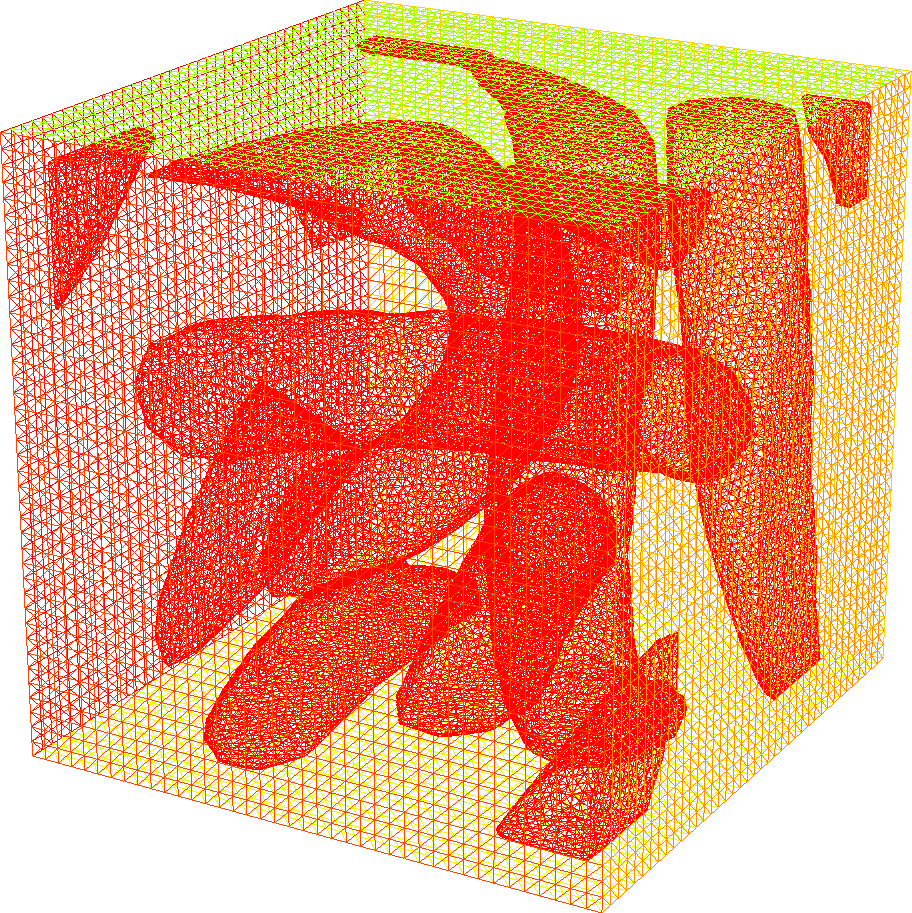}
        \caption{\tiny REV 5 ($\phi \approx 0.75$).}
    \end{subfigure} \hfill
    \begin{subfigure}[b]{0.2\textwidth}
         \includegraphics[width=0.9\linewidth]{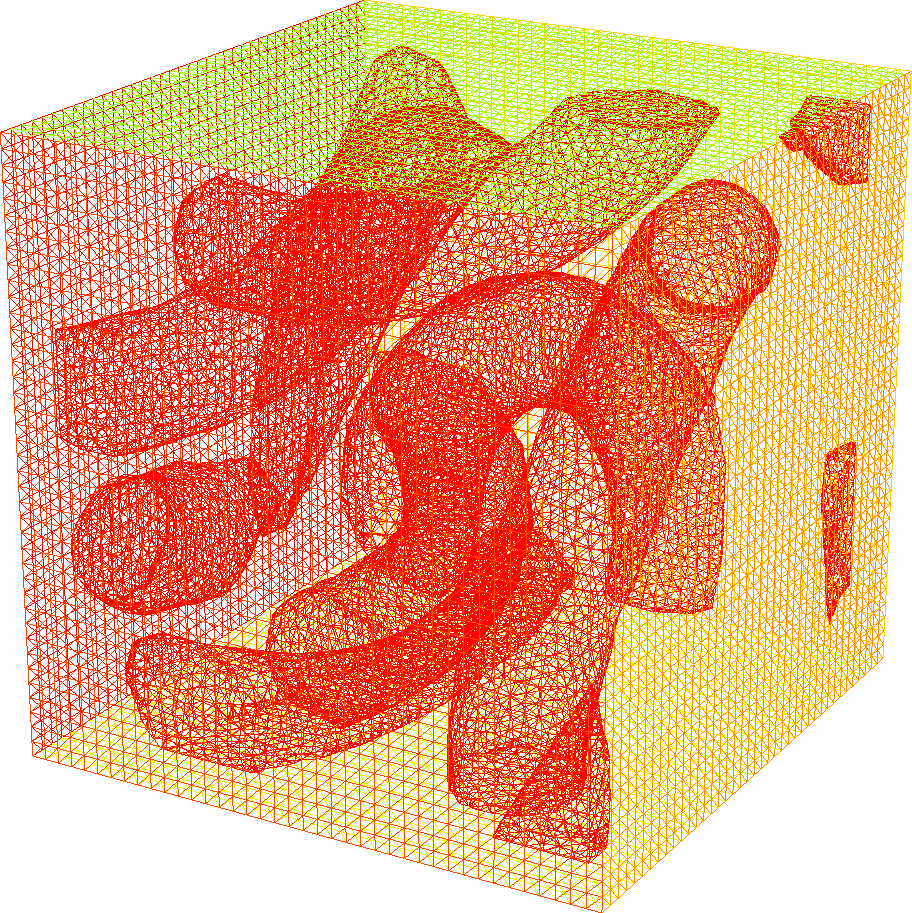}
        \caption{\tiny REV 6 ($\phi \approx 0.73$)}
    \end{subfigure} \hfill 
    \begin{subfigure}[b]{0.2\textwidth}
    \includegraphics[width=0.9\linewidth]{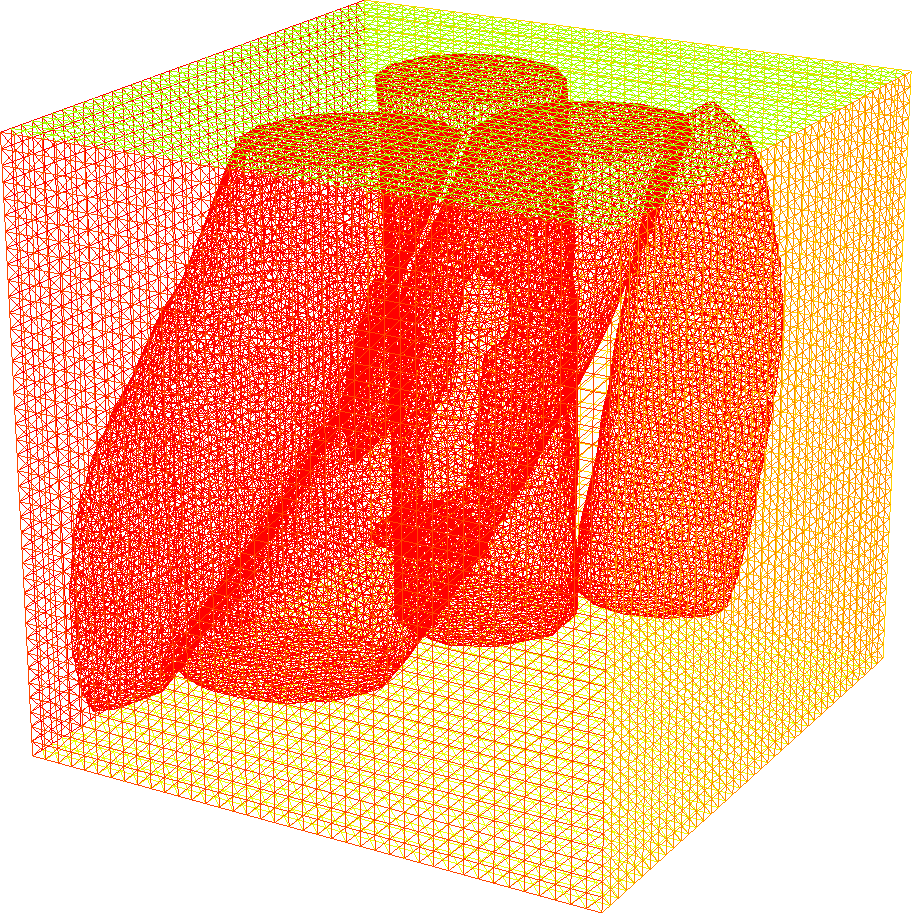}
    \caption{\tiny REV 7 ($\phi \approx 0.71$).}
    \end{subfigure}\hfill
    \begin{subfigure}[b]{0.2\textwidth}
    \includegraphics[width=0.9\linewidth]{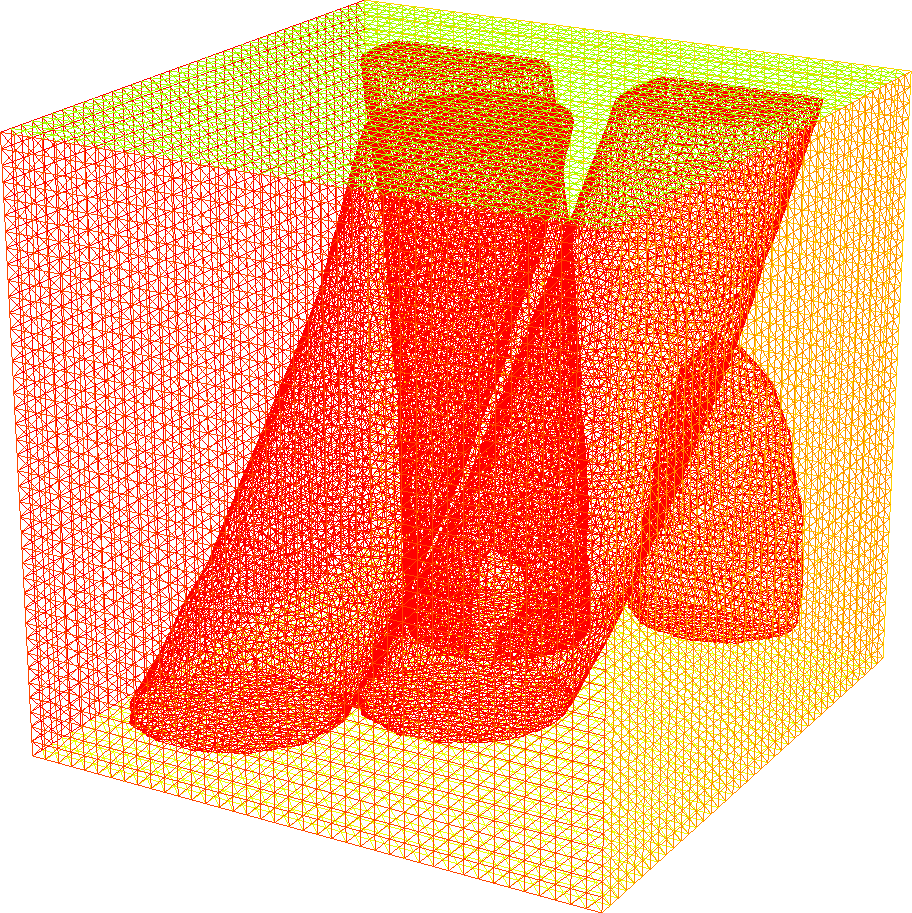}
    \caption{\tiny REV 8 ($\phi \approx 0.70$).}
    \end{subfigure}
\hspace{0.5\linewidth}
 
    \begin{subfigure}[b]{0.2\textwidth}
    \includegraphics[width=0.9\linewidth]{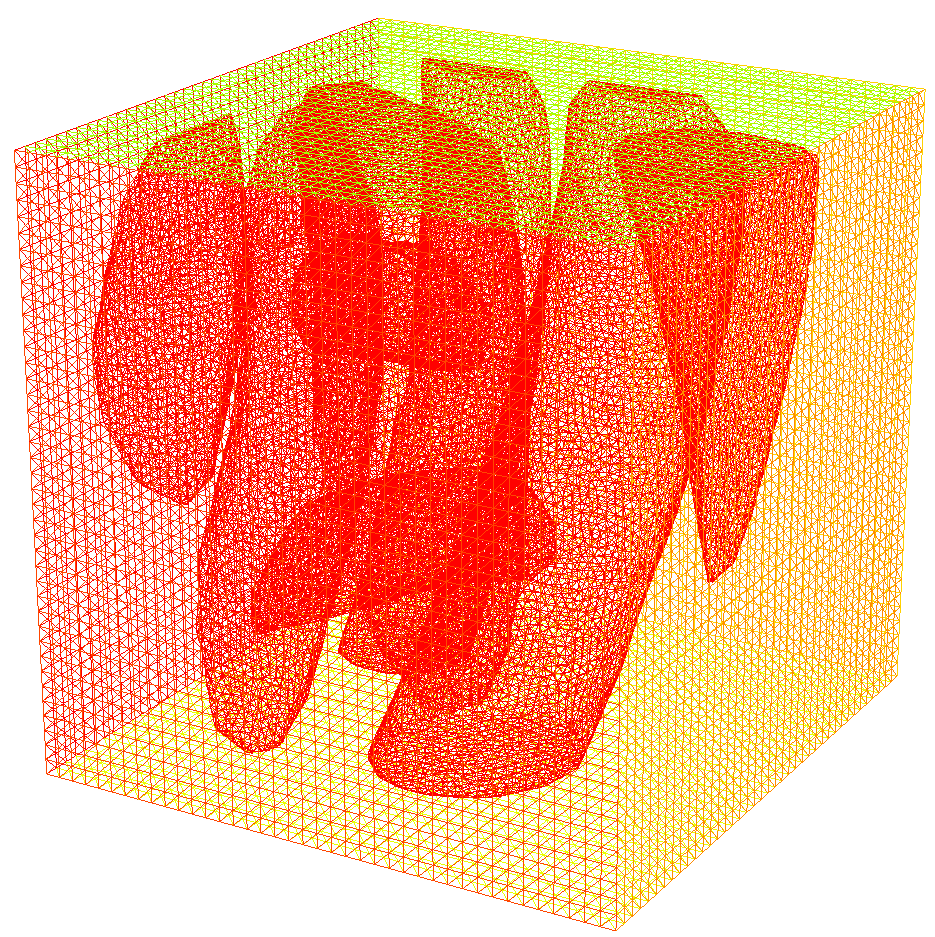}
    \caption{\tiny REV 9 ($\phi \approx 0.69$).}
    \end{subfigure}\hfill
    \begin{subfigure}[b]{0.2\textwidth}
     \includegraphics[width=0.9\linewidth]{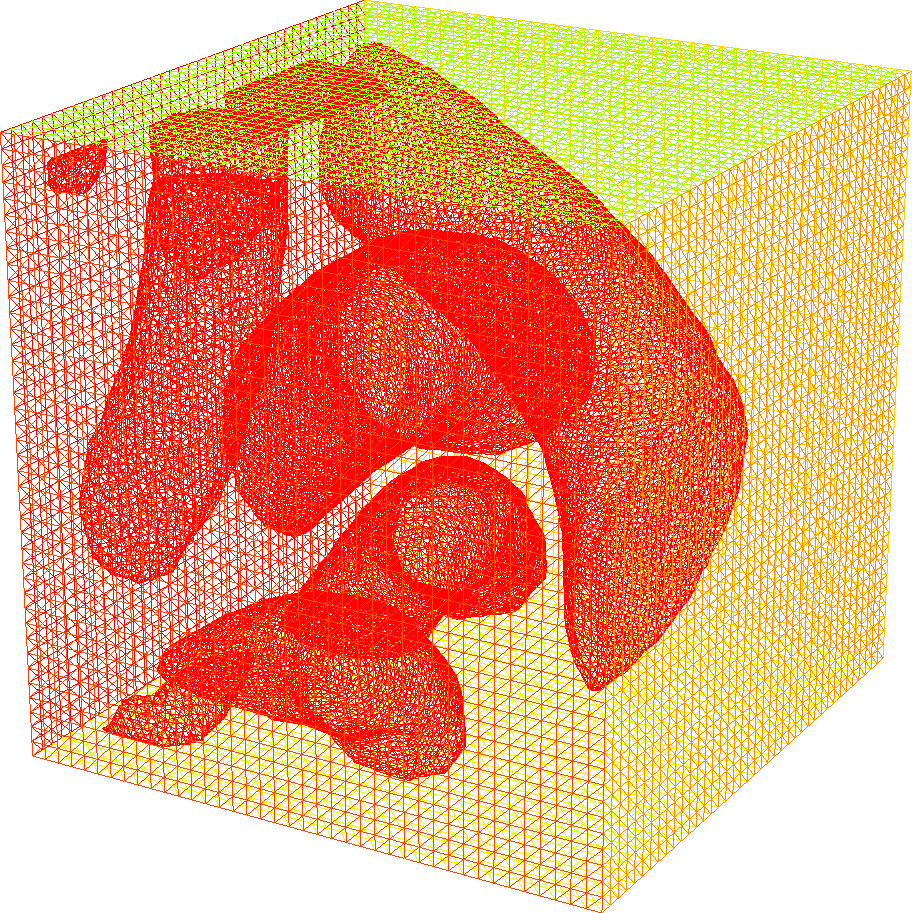}
    \caption{\tiny REV 10 ($\phi \approx 0.69$).}
    \end{subfigure}\hfill
    \begin{subfigure}[b]{0.2\textwidth}
    \includegraphics[width=0.9\linewidth]{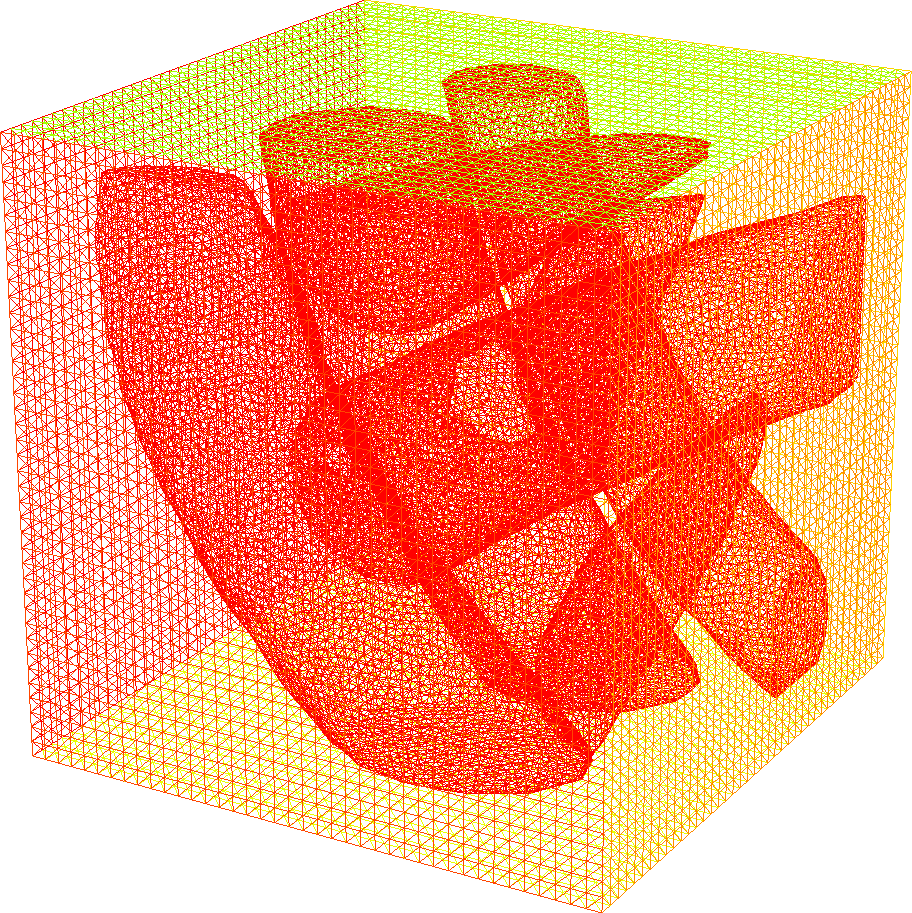}
    \caption{\tiny REV 11 ($\phi \approx 0.65$).}
    \end{subfigure} \hfill
    \begin{subfigure}[b]{0.2\textwidth}
    \includegraphics[width=0.9\linewidth]{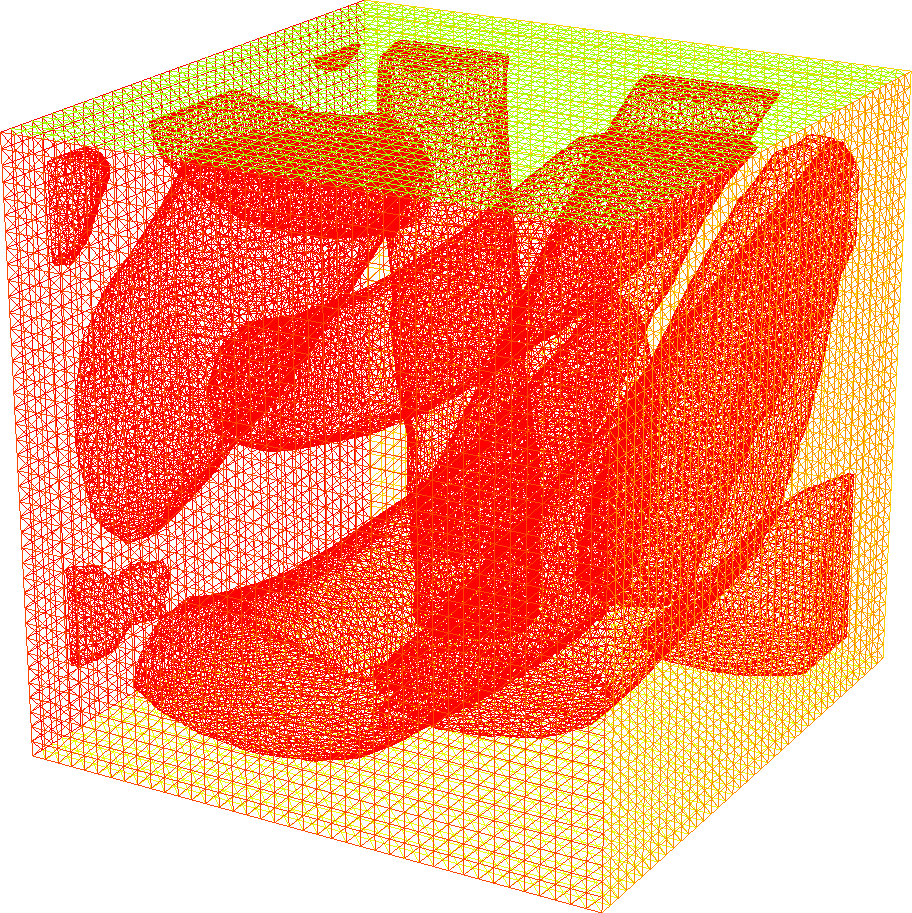}
    \caption{\tiny REV 12 ($\phi \approx 0.63$).}
    \end{subfigure}    
    
    \caption{The 12 chosen REVs and their respective porosities $\phi$.}
    \label{fig:AllREVs}
\end{figure}

For each REV, we compute the cell problems \eqref{eq:cellpb} in the cell $Y^\kappa$ (we recall that the applications of the periodic boundary conditions are made possible using the oversampled cell $Y^\kappa$, with $\kappa=1.1$ in this work), and compute the permeability using \eqref{eq:MatrixPermeabilityOver}.
As a result, we obtain a semi-definite positive permeability matrix, $\mathbf{K}^\kappa_{\rm coil}$, which we diagonalise to find the effective permeabilities, i.e.
\begin{equation}
    \mathbf{K}^\kappa_\mathrm{coil} =  \mathbf{O} \mathbf{K}^\kappa_\mathrm{eff, coil} \mathbf{O}^\top
\label{eq:matrixM}
\end{equation}
with $\mathbf{O}$ an orthogonal matrix and 
\begin{equation*}
    \mathbf{K}^\kappa_\mathrm{eff,coil} = \begin{pmatrix}
        K_{1, \rm coil}^\kappa & 0 & 0 \\
        0 & K_{2, \rm coil}^\kappa & 0  \\
        0 & 0 & K_{3, \rm coil}^\kappa  \\
    \end{pmatrix}
\end{equation*}
in which $0<K_{1, \rm coil}^\kappa \leq K_{2, \rm coil}^\kappa  \leq K_{3, \rm coil}^\kappa$.
Given this ordering, the two smaller effective permeabilities are assumed to be transversal permeabilities and the largest one is assumed to be the longitudinal permeability.
Then, to recover the true permeabilities, we apply the scaling \eqref{eq:relation_over}, i.e.
\begin{equation*}
    \widetilde{\mathbf{K}}_{\rm eff, coil} \approx \frac{\phi}{\kappa^2} \mathbf{K}^\kappa_{\rm eff, coil},
\end{equation*}
where, as previously, $\phi$ is the porosity in the unit cell. We note that the scaling \eqref{eq:relation_over} is valid in this case, as the porosity in the REVs is greater than 0.6 and each REV contains at most approximately five fibres. All corresponding results are presented in \cref{fig:PermBoutin1}.

\begin{figure}[tbp]
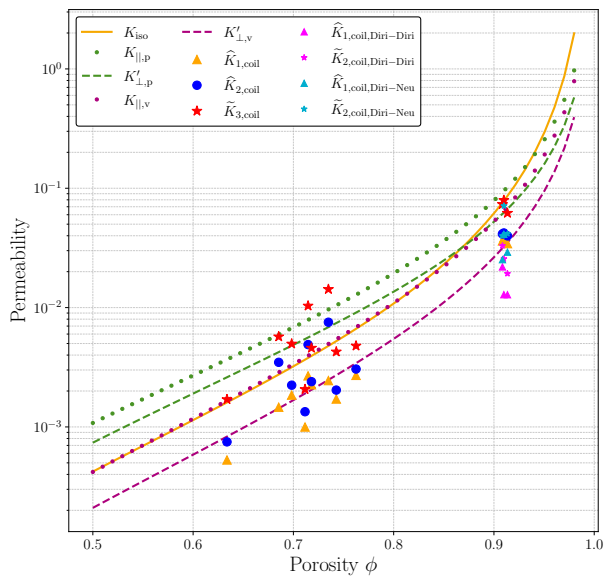

    \centering
    \include{figures/PLOTS/PermRVEsPaper0.178}
    \caption{Comparison of permeabilities for REVs from the real coil geometry. The analytical permeabilities are computed with $R = R_{\rm coil} = 0.178$ (Framing Coil).}
    \label{fig:PermBoutin1}
\end{figure}

From \cref{fig:PermBoutin1}, we clearly see that the analytical permeabilities \eqref{eq:Klong} and \eqref{eq:Kperp_mod} accurately capture the effective permeabilities of the coil packing, namely $\widetilde{K}_{1, \rm coil}$, $\widetilde{K}_{2, \rm coil}$, $\widetilde{K}_{3, \rm coil}$.

In particular, the analytical permeabilities $\mathbf{K}_{\rm v}$ based on kinematically continuous fields match well for most of the REVs.
\medskip

The final step of our analysis consists in assessing whether the relation \eqref{eq:Rotation} from \cref{lemma:Rotation} remains valid for the coil REVs, even though the assumptions of the lemma are not satisfied. For each REV in \cref{fig:AllREVs}, we compute the corresponding rotation matrix $\mathbf{R}_\theta$  using  \cref{alg:direction_wire} (see \cref{app.Algo}), with a sampling rate $N_s=40$, which provides the average direction of the coil packing inside the REV.
Note that a sensitivity analysis of \cref{alg:direction_wire} with respect to the sampling rate $N_s \in \{30,40,50,60\}$ is provided in \cref{sec.SensWire}. We then compare the computed permeability matrix $\mathbf{K}^\kappa_\mathrm{coil}=  \mathbf{O} \mathbf{K}^\kappa_\mathrm{eff, coil} \mathbf{O}^\top $ defined in \eqref{eq:matrixM} to
\begin{equation}
\mathbf{K}^\kappa_\mathrm{coil,algo} :=\mathbf{R}_\theta \mathbf{K}^\kappa_\mathrm{eff, coil} \mathbf{R}^\top_\theta
\label{eq:Kalgo}
\end{equation}
in order to evaluate how accurately the matrix $\mathbf{O}$ can be approximated by $\mathbf{R}_\theta$, noting that a sensitivity analysis with respect to perturbations of the rotation matrix $\mathbf{R}_\theta$ given by \cref{alg:direction_wire} is given in \cref{sec.Sens}. 
To quantify the discrepancy, we use the Frobenius norm, defined as $\| A \|_F = (\sum_{i,j} A_{ij}^2)^{1/2}$, and introduce the relative error
\begin{equation*}
{\rm err}_{\mathbf{K}} = \frac{\lVert \mathbf{K}^\kappa_\mathrm{coil} -  \mathbf{K}^\kappa_\mathrm{coil,algo} \rVert_F} {\lVert \mathbf{K}^\kappa_\mathrm{coil} \rVert_F}.    
\end{equation*}
The relative errors obtained for all considered REVs are reported in \cref{table:Comparison}.

\begin{table}[htbp]
    \caption{Comparison of the action of the orthogonal matrix $\mathbf{O}$ of \eqref{eq:matrixM} and the action of the rotation matrix $\mathbf{R}_\theta$ obtained with \cref{alg:direction_wire} (using a sampling rate $N_s = 40$) for the 12 coil REVs presented in \cref{fig:AllREVs}.}
    \label{table:Comparison}
    \centering
    \setlength{\tabcolsep}{3pt} 
    \begin{tabular}{c|ccccccccccccc}
      \toprule
    REV & 1 & 2 & 3 & 4 & 5 & 6 & 7 & 8 & 9 & 10 & 11 & 12 & Mean\\
    \midrule
    Porosity $\phi$ 
    & 0.91 & 0.89 & 0.84 &  0.80 &  0.75  & 0.73  &  0.71 &  0.70 &  0.69  & 0.69 &  0.65 &  0.63 & -- \\
    \midrule
    ${\rm err}_{\mathbf{K}}$ [\%] 
    & 11.63 & 13.59 & 24.35 & 18.82 & 17.36 & 16.86 & 45.69 & 23.77 & 10.90 & 11.77 & 13.43 & 20.63& 19.07  \\
    \bottomrule
    \end{tabular}
\end{table}

From \cref{table:Comparison}, it can be observed that the permeability matrices $\mathbf{K}^\kappa_\mathrm{coil}$ and 
$\mathbf{K}^\kappa_\mathrm{coil,algo}$ are reasonably close. Although the relative error is on average about 19\%, this level of discrepancy remains acceptable for our purposes and supports the use of $\mathbf{R}_\theta$ as a suitable approximation of $\mathbf{O}$, even though the assumption of \cref{lemma:preservation_symmetry} is not fulfilled. The largest error is measured for REV 7. Upon closer inspection, one observes that the dominant fibre direction in $\mathbf{K}^\kappa_\mathrm{coil}$ and $\mathbf{K}^\kappa_\mathrm{coil, algo}$ agrees reasonably well, whereas the two remaining minor directions do not; their eigenvectors are approximately exchanged. Since the two smallest eigenvalues of $\mathbf{K}^\kappa_\mathrm{coil}$ are not equal, this mismatch contributes significantly to the relative error. However, in the permeability model employed later, the two transverse permeabilities are assumed equal, resulting in a permeability tensor that is independent of the orientation of the two minor eigenvectors. It should also be noted that the oversampling procedure inevitably results in the loss of certain directional information, which may partly explain these discrepancies.

\medskip

In the following section, we employ the permeability model established here and evaluate it across various test cases within the framework of the Darcy equation \eqref{eq:DarcyEquations}.

\section{Numerical validation of the upscaled macroscopic model}
\label{sec:Numerical}

In this section, we consider three settings of increasing complexity to asses the quality of the effective Darcy flow problem using the permeability parametrisations discussed above.
We begin with a cube diagonally cut by cylinders, which is a typical geometry considered in homogenisation, followed by a cylinder cut by spirals, and finally an artificial aneurysm occluded by a coil.
All settings contain two main challenges present in most practical applications: the effects of finite size and of boundary conditions.
The latter reduce the flow in an aneurysm due to no slip at close-by walls, while the former is due to the fact that the ratio of aneurysm radius to coil radius is quite small (often in the range below 10), which requires caution since direct application of homogenisation results might result in large errors.

In all three settings, the boundary conditions for Darcy flow are $p = 1$ at $x_1=0$ and $p = 0$ at $x_1 = L$, where $L$ denotes the length of the domain (in $x_1$-direction), and homogeneous Neumann conditions (no flow) at the remaining boundary $\partial\Omega \setminus \{ x \in \partial\Omega \text{ such that } x_1 = 0,L \}$.
In the first two cases, we compare to the resolved Stokes flow with the corresponding boundary conditions being Neumann conditions $(\nabla v + p\mathbf{I}_3) \nu = \nu$ at $x_1 = 0$ (inflow) and $(\nabla v + p\mathbf{I}_3) \nu = 0$ at $x_1 = L$ (outflow), where $\nu$ denotes the outward normal vector, and a homogeneous Dirichlet condition (no-slip) at the remaining boundary $\partial\Omega_\varepsilon \setminus \{ x \in \partial\Omega_\varepsilon \ : \ x_1 = 0,L \}$.
Then, we compare the resulting average velocities $\bar{v} = |\Omega|^{-1} \int_{\Omega_\varepsilon} v_\varepsilon \,\mathrm{d}x$ (Stokes) and $\bar{v} = |\Omega|^{-1} \int_{\Omega} v \,\mathrm{d}x$ (Darcy), respectively.
Note that we set $\mu = 1$ for simplicity as all problems are linear.

Both the Stokes and Darcy flow problems are discretised using divergence-free discontinuous finite elements of degree $k=1$.
In case of Stokes flow, we use the $\mathbb{BDM}_1$-$\mathbb{P}_1$ element with symmetric interior penalty approach, see e.g.~\cite{Rhebergen2025} for details.
For the Darcy flow problems, we use the mixed discretisation by $\mathbb{RT}_1$-$\mathbb{P}_1$ elements.
The implementation is done in Python using the FEniCSx library \cite{DOLFINx} and is available in the GitHub repository \url{https://github.com/s-lunowa/fenicsx-homogenization}.

\begin{figure}[htbp]
    \centering
    \includegraphics[width=0.49\linewidth]{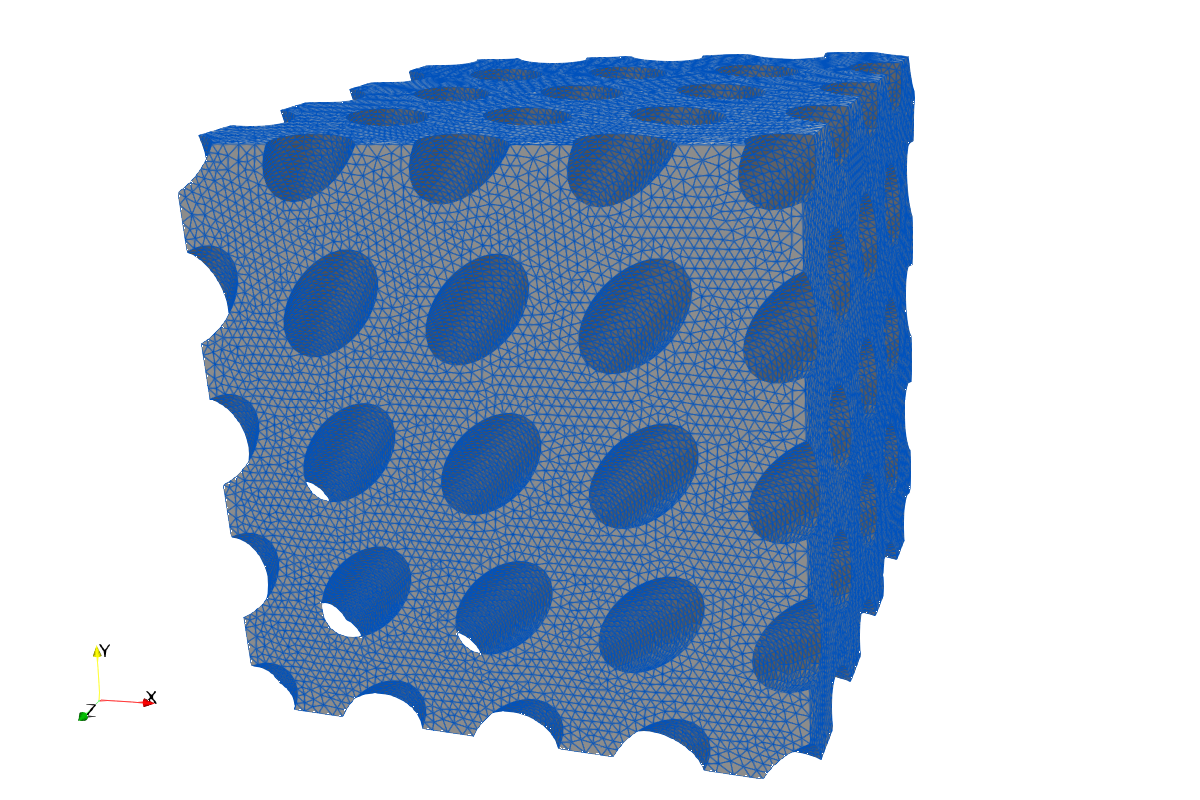}
    \includegraphics[width=0.49\linewidth]{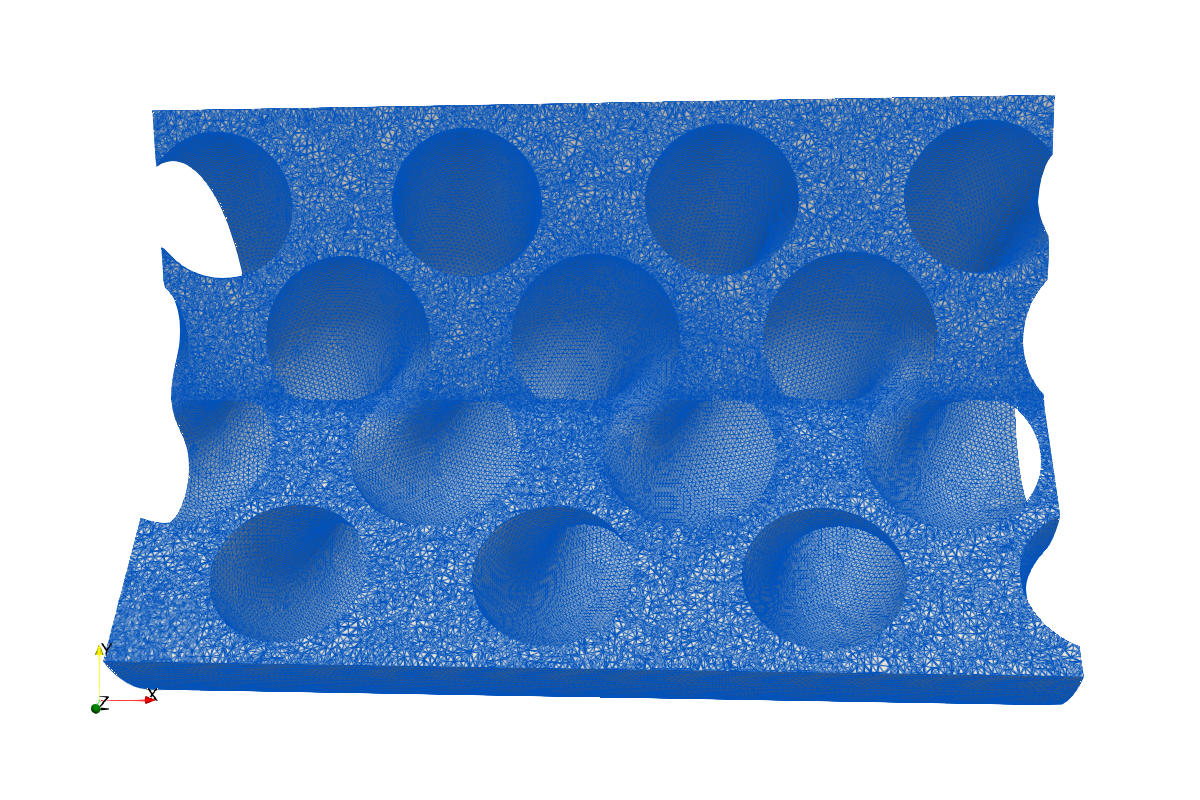}
    \caption{Mesh for the Stokes flow simulation of setting 1 (left) and cut through the mesh in setting 2 (right).}
    \label{fig:stokes:mesh}
\end{figure}

\begin{table}[htpb]
    \caption{Mesh sizes $h$, number of mesh elements $|\mathcal{T}_h|$, number of degrees of freedom (DOFs) in velocity and pressure for the Stokes and Darcy flow simulations of the different settings in \cref{sec:Numerical}.}
    \label{tab:mesh}
    \centering
    \def\z{\phantom{0}}
    \small
    \begin{tabular}{l|cccc|cccc}
        \toprule
        & \multicolumn{4}{c|}{\textbf{Stokes}} & \multicolumn{4}{c}{\textbf{Darcy}} \\
        Setting & $\boldsymbol{h}$ & $\boldsymbol{|\mathcal{T}_h|}$ & Vel. DOFs & Pr. DOFs & $\boldsymbol{h}$ & $\boldsymbol{|\mathcal{T}_h|}$ & Vel. DOFs & Pr. DOFs \\
        \midrule
        1) Cube     & 0.035  & 711,583 & 11,001,444 & 2,364,452 & 0.03125    & 196,608 & 1,787,904 &\z\,786,432 \\
        2) Cylinder & 0.035  & 734,118 & 11,459,484 & 2,468,568 & 0.05\z\z\z & 364,710 & 3,105,465 &  1,364,480 \\
        3) Aneurysm & \multicolumn{4}{c|}{---}                  & 0.15\z\z\z & 255,190 & 2,143,995 &\z\,939,312 \\
        \bottomrule
    \end{tabular}
\end{table}

\subsection{Setting 1: Cube cut by cylinders}

We consider a cube of side length $L=1$. The perforated domain $\Omega_\varepsilon$ is constructed from a $4 \times 4 \times 4$ array of identical subdomains, each containing  a diagonally oriented cylindrical obstacle $\mathcal{O}$ of radius $R_\varepsilon = 1/4$, cf.~\cref{fig:stokes:mesh} (left).
Hence, the cylinder's length is $L_\varepsilon = \sqrt{3}$, leading to a solid fraction of $\rho = \pi R_\varepsilon^2 L_\varepsilon \approx 0.34$ and a porosity of $\phi = 1 - \rho \approx 0.66$.
The effective permeability matrices are obtained using the cylinder direction $t_{\rm c} = 3^{-1/2} (1, 1, 1)^\top$ and effective radius $R = \varepsilon R_\varepsilon$ with $\varepsilon=1/4$:
\begin{align*}
    \mathbf{K}_{\rm iso} &= K_{\rm iso} \mathbf{I}_3 \approx 2.59 \cdot 10^{-4} \, \mathbf{I}_3, \\
    \mathbf{K}_{\rm p} &= K_{\parallel,\rm p} t_{\rm c} \otimes t_{\rm c} + K_{\perp, \rm p}' \left(\mathbf{I}_3 - t_{\rm c} \otimes t_{\rm c}\right)
        \approx 10^{-4} \begin{pmatrix} 4.66 & 0.55 & 0.55 \\ 0.55 & 4.66 & 0.55 \\ 0.55 & 0.55 & 4.66 \end{pmatrix},\\
    \mathbf{K}_{\rm v} &= K_{\parallel,\rm v} t_{\rm c} \otimes t_{\rm c} + K_{\perp, \rm v}' \left(\mathbf{I}_3 - t_{\rm c} \otimes t_{\rm c}\right)
        \approx 10^{-4} \begin{pmatrix} 1.79 & 0.45 & 0.45 \\ 0.45 & 1.79 & 0.45 \\ 0.45 & 0.45 & 1.79 \end{pmatrix}.
\end{align*}
By computing the permeability in the unit cell for this setting (cf.~\cref{sec.TrueCoil}), we find
\begin{equation*}
    \mathbf{K}_{\rm hom}^\varepsilon = \varepsilon^2 \mathbf{K}_{\rm hom}
        \approx 5.9 \cdot 10^{-4} \, t_{\rm c} \otimes t_{\rm c} + 2.9 \cdot 10^{-4} \left(\mathbf{I}_3 - t_{\rm c} \otimes t_{\rm c}\right)
        \approx 10^{-4} \begin{pmatrix} 3.9 & 0.99 & 0.99 \\ 0.99 & 3.9 & 0.99 \\ 0.99 & 0.99 & 3.9 \end{pmatrix}.
\end{equation*}
Finally, based on the results of \cref{sec:permeability}, especially \cref{fig:PermBoutin01}, we also consider the weighted average
\begin{align*}
    \mathbf{K}_{\phi} &= \frac{3 K_{\parallel,\rm p} + K_{\parallel,\rm v}}{4} t_{\rm c} \otimes t_{\rm c} + \frac{K_{\perp, \rm p}' + K_{\perp, \rm v}'}{2} \left(\mathbf{I}_3 - t_{\rm c} \otimes t_{\rm c}\right)
        \approx 10^{-4} \begin{pmatrix} 3.48 & 0.76 & 0.76 \\ 0.76 & 3.48 & 0.76 \\ 0.76 & 0.76 & 3.48 \end{pmatrix}.
\end{align*}

The tetrahedral meshes for the Stokes problem are generated using gmsh, while those for the Darcy problems are uniform, see \cref{tab:mesh} for the details.
The resulting average velocities are summarised in \cref{tab:case1} and the diagonal cut through the pressure solutions is depicted in \cref{fig:case1:pressure}.
The parametric cases ($\mathbf{K}_{\rm p}$ and $\mathbf{K}_{\rm v}$) significantly over- and under-predict the values (similar to the findings in \cref{sec:permeability}).
As expected, the isotropic Darcy flow ($\textbf{K}_{\rm iso}$) cannot replicate the transversal flow, while the homogenised case ($\mathbf{K}_{\rm hom}^\varepsilon$) over-predicts the transversal flow.
However, the values coincide almost exactly in the averaged case ($\mathbf{K}_{\phi}$).
Also, visually, the pressures of the homogenised and averaged cases closely correspond to the Stokes pressure.

\begin{table}[htpb]
    \caption{Average velocity components $\bar{v}_{i}$ ($i=1,2,3$) in setting 1 (cube cut by cylinders) for the Stokes simulation and for the Darcy simulations with different permeability parametrizations $\textbf{K}_{\bullet}$ (and the relative deviation from Stokes).}
    \label{tab:case1}
    \centering\small

    \begin{tabular}{c|c
        |c@{\ }c
        |c@{\ }c
        |c@{\ }c
        |c@{\ }c
        |c@{\ }c
    }
        \toprule
        Quantity & Stokes
        & \multicolumn{2}{c|}{$\mathbf{K}_{\rm iso}$}
        & \multicolumn{2}{c|}{$\mathbf{K}_{\rm p}$}
        & \multicolumn{2}{c|}{$\mathbf{K}_{\rm v}$}
        & \multicolumn{2}{c|}{$\mathbf{K}_{\phi}$}
        & \multicolumn{2}{c}{$\mathbf{K}_{\rm hom}^{\varepsilon}$}
        \\
        \midrule
        $\bar{v}_{1}$ [{\small$\times10^{-4}$}]
        & 3.285
        & 2.588 & {\small(-21\%)}
        & 4.597 & {\small(+40\%)}
        & 1.693 & {\small(-48\%)}
        & 3.333 & {\small(+1\%)}
        & 3.673 & {\small(+12\%)}
        \\
        $\bar{v}_{2}$ [{\small$\times10^{-4}$}]
        & 0.334
        & 0.000 & {\small(-100\%)}
        & 0.262 & {\small(-22\%)}
        & 0.197 & {\small(-41\%)}
        & 0.339 & {\small(+2\%)}
        & 0.435 & {\small(+30\%)}
        \\
        $\bar{v}_{3}$ [{\small$\times10^{-4}$}]
        & 0.334
        & 0.000 & {\small(-100\%)}
        & 0.262 & {\small(-21\%)}
        & 0.197 & {\small(-41\%)}
        & 0.339 & {\small(+2\%)}
        & 0.435 & {\small(+30\%)}
        \\
        \bottomrule
    \end{tabular}

\end{table}

\begin{figure}[htpb]
    \small
    \parbox{0.3\textwidth}{\centering Stokes}
    \parbox{0.3\textwidth}{\centering Isotropic Darcy ($\mathbf{K}_{\rm iso}$)}
    \parbox{0.3\textwidth}{\centering Anisotropic Darcy ($\mathbf{K}_{\phi}$)}\\    
    \includegraphics[width=0.3\linewidth, trim=70mm 40mm 70mm 40mm, clip]{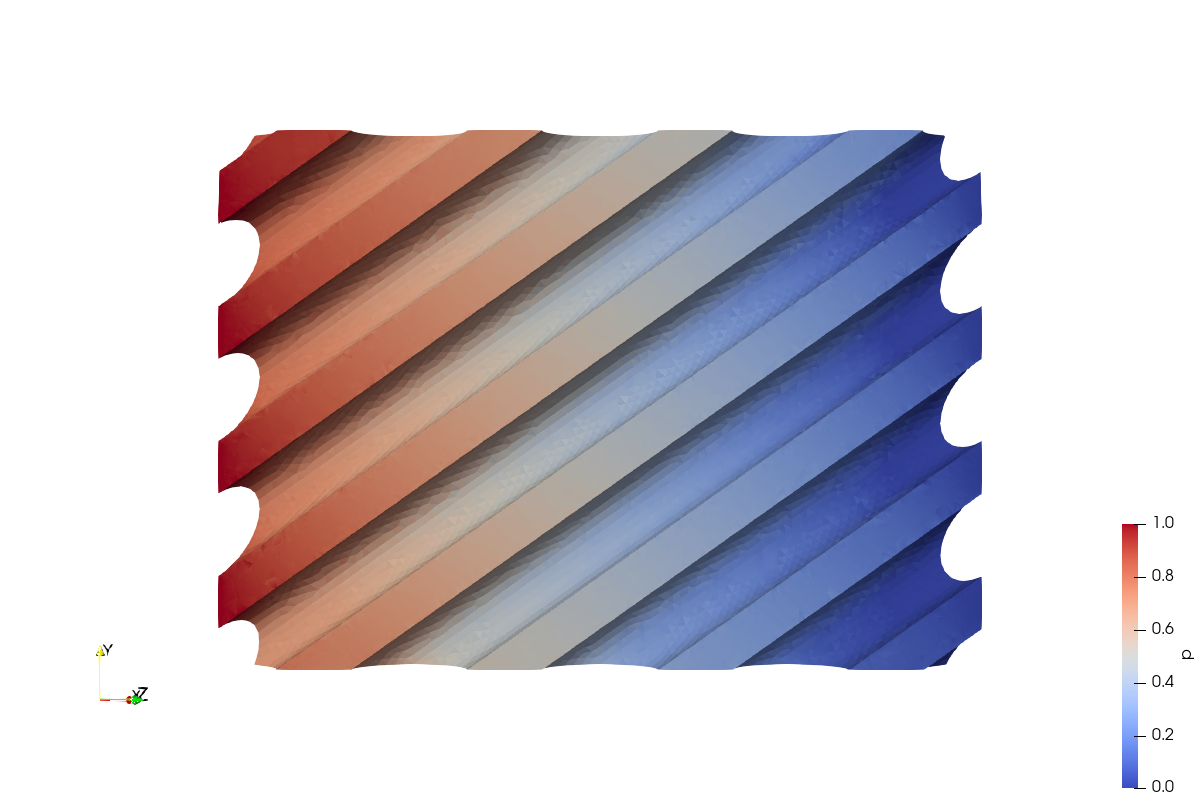}
    \includegraphics[width=0.3\linewidth, trim=70mm 40mm 70mm 40mm, clip]{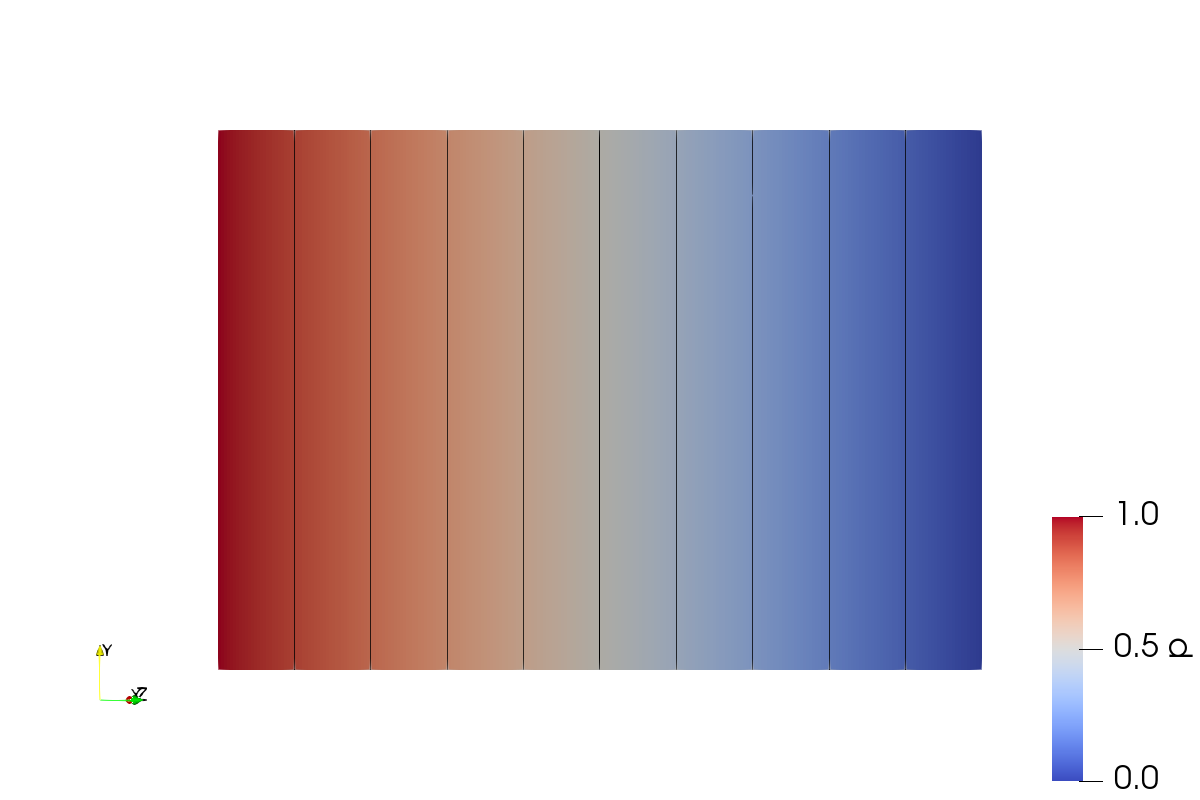}
    \includegraphics[width=0.3\linewidth, trim=70mm 40mm 70mm 40mm, clip]{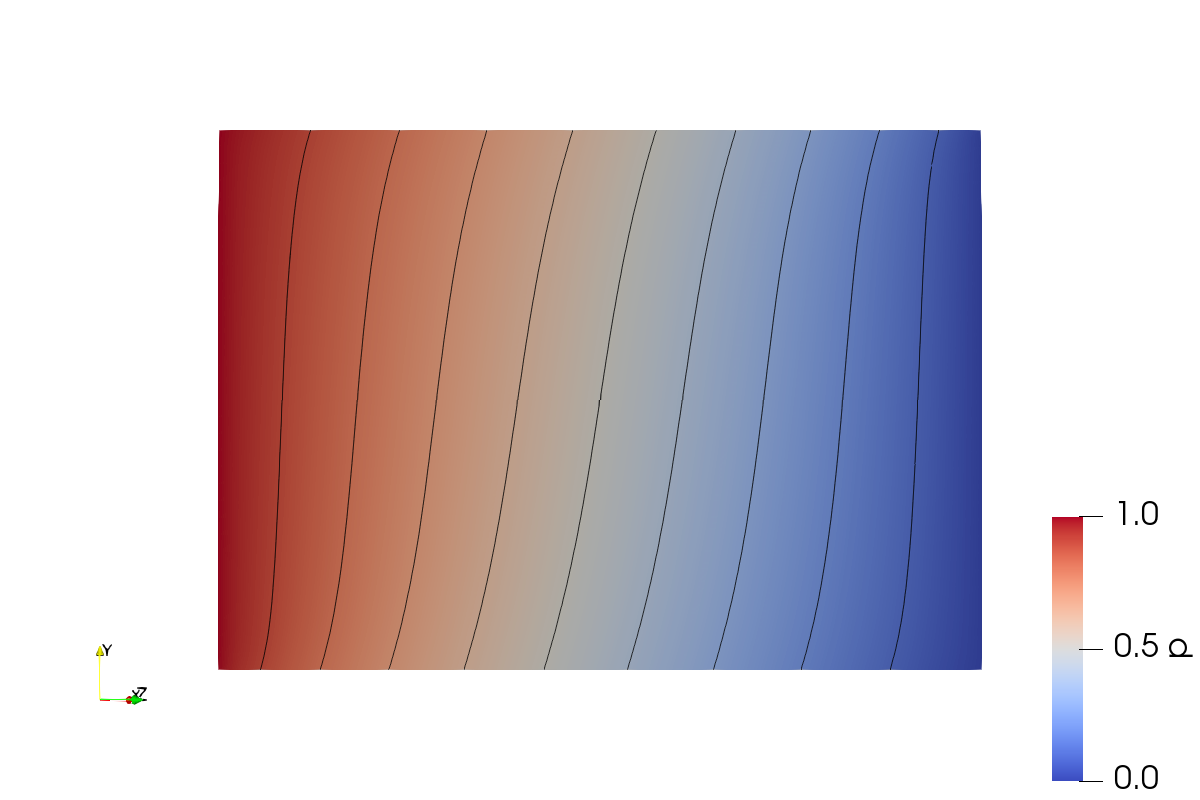}
    \hspace*{0.05\linewidth}

    \parbox{0.3\textwidth}{\centering Anisotropic Darcy ($\mathbf{K}_{\rm p}$)}
    \parbox{0.3\textwidth}{\centering Anisotropic Darcy ($\mathbf{K}_{\rm v}$)}
    \parbox{0.3\textwidth}{\centering Anisotropic Darcy ($\mathbf{K}_{\rm hom}^\varepsilon$)}\\
    \includegraphics[width=0.3\linewidth, trim=70mm 40mm 70mm 40mm, clip]{  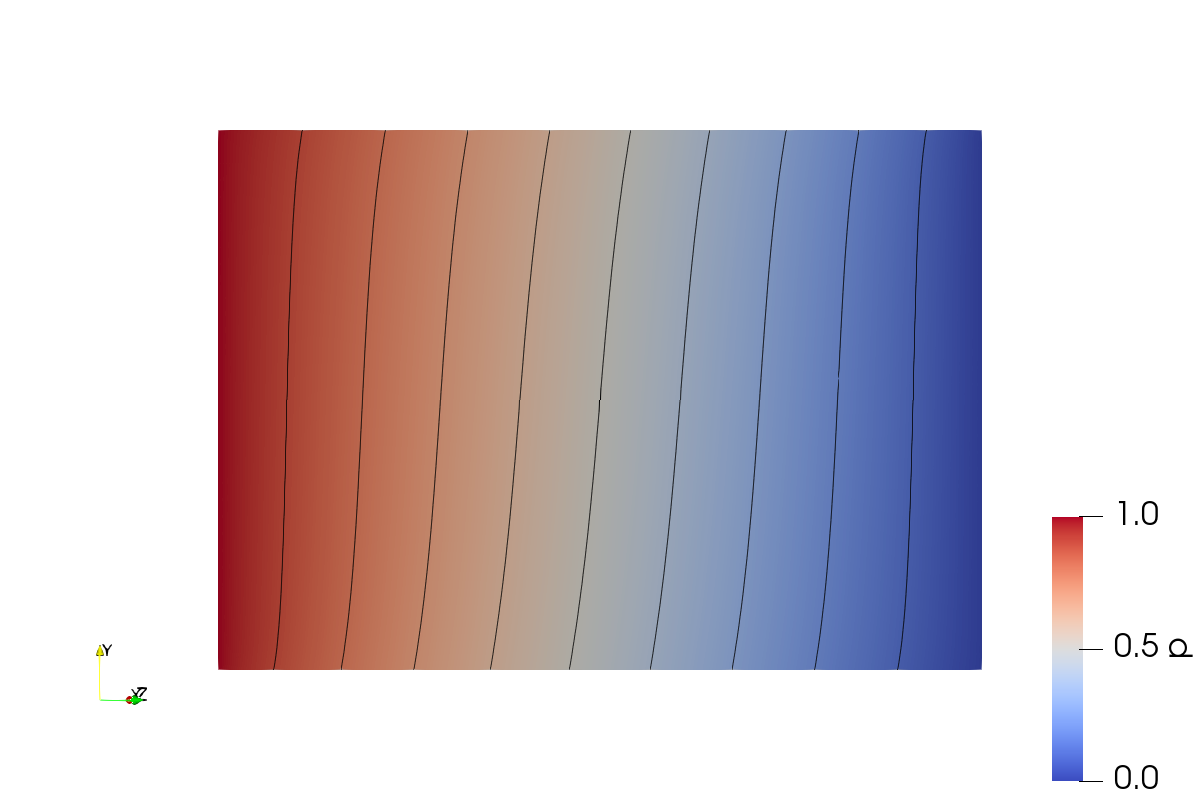}
    \includegraphics[width=0.3\linewidth, trim=70mm 40mm 70mm 40mm, clip]{  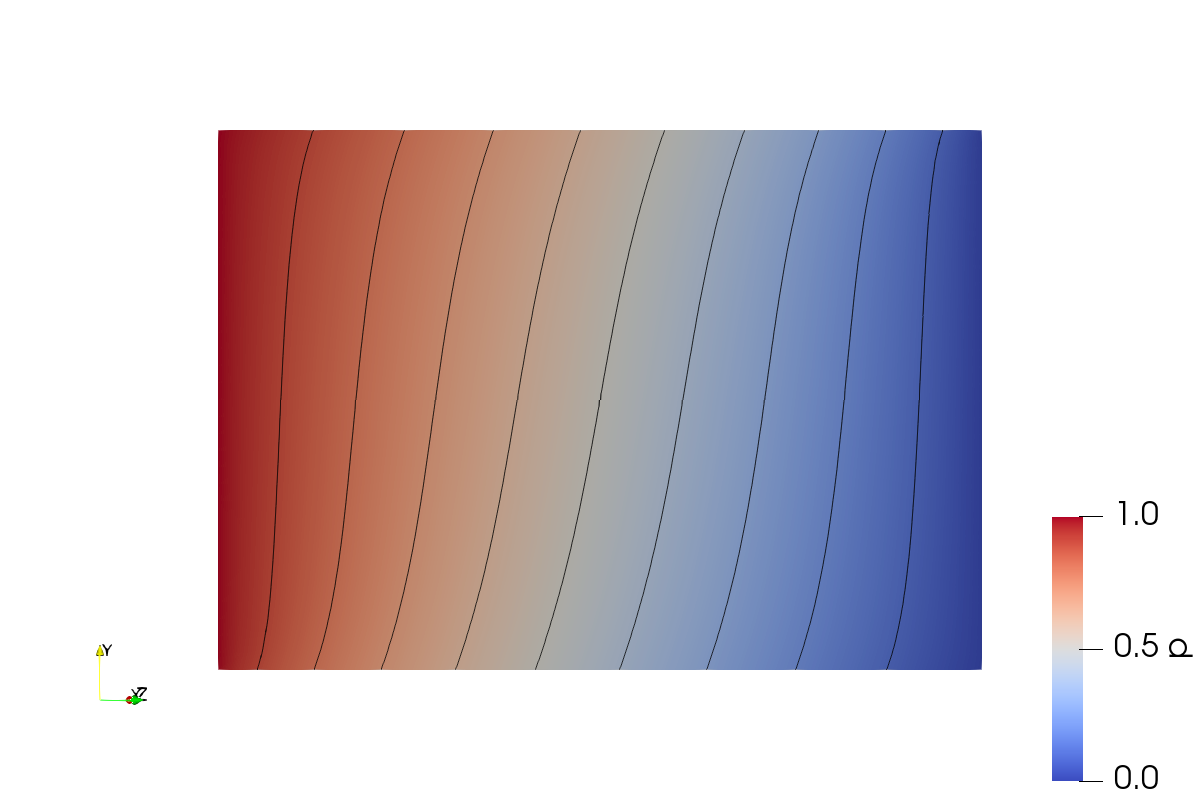}
    \includegraphics[width=0.3\linewidth, trim=70mm 40mm 70mm 40mm, clip]{  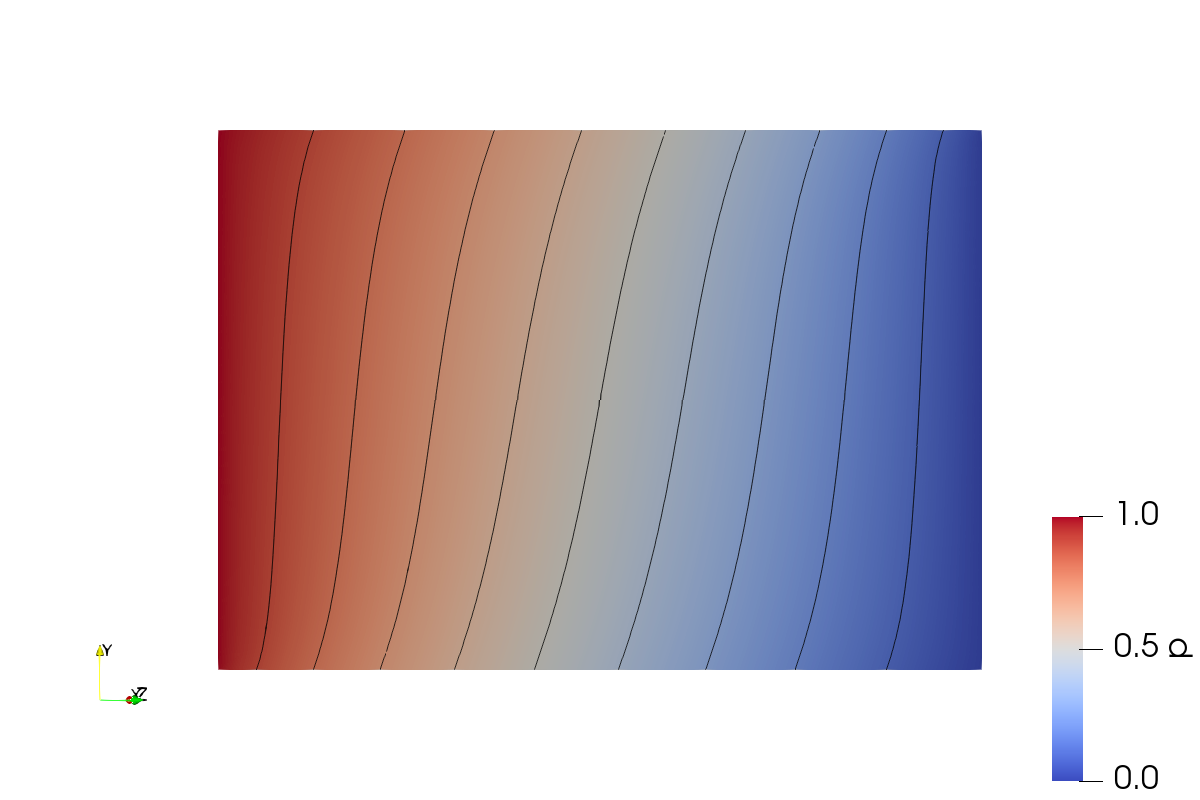}
    \includegraphics[width=0.05\linewidth, trim=37cm 0cm 12mm 17cm, clip]{  figures/darcy_cube_eps0.25_homogenized_pressure.png}

    \caption{Diagonal cut through the pressure solutions (with iso-lines $0.1, 0.2, \dots, 0.9$) for setting 1 (cube cut by cylinders).}
    \label{fig:case1:pressure}
\end{figure}

\subsection{Setting 2: Cylinder cut by spirals}

Next, we consider a domain based on a cylinder of radius one and length $L=3$ where two spirals are cut out, cf. \cref{fig:stokes:mesh} (right).
Both spirals have wire radius $R=0.25$, and the rotation radii $r_o = 0.65$, $r_i = 0.25$ with $3.5$ rotations within the cylinder.
The resulting total porosity (w.r.t.~the cylinder) is $\phi \approx 0.559$.
The tetrahedral meshes for the Stokes and Darcy problems are generated using gmsh, see \cref{tab:mesh} for details.

In a preprocessing step for the calculation of the averaged permeability and porosity, we generate a discretisation of the domain according to \cref{alg:rev_determination} (see \cref{app.Algo}).
First, a structured Cartesian grid is generated with respect to the bounding box of the domain $\mathcal{M}_{\rm a}$.
From this, a constant spacing $\Delta x$ is then calculated, depending on the smallest bounding box side length and a sampling rate $N_s = 40$.
We then generate mesh nodes in the bounding box, which is extended in length by the REV radius $r_{\rm REV}$, and centre them around the bounding box by adding $\Delta x / 2$.
The extension by $r_{\rm REV}$ ensures that each of the nodes contained in $\mathcal{M}_{\rm a}$ has neighbour points in the box with radius $r_{\rm REV}$.
For the averaging of porosity, we employ \cref{alg:porosity_wire} (see \cref{app.Algo}).
We assume that the coil $\mathcal{M}_{\rm c}$ is fully contained in the aneurysm sack $\mathcal{M}_{\rm a}$ (the domain to be homogenised), which itself is fully contained in the vessel $\mathcal{M}_{\rm v}$.
As a first step, we define the volume fraction of fluid and solid $\rho'(x)$ in $\mathcal{M}$.
For each point in the the grid, we assign a value $0\leq\rho_w\leq 1$ if it is contained in the vessel wall, a value of 1 (solid) if it s contained in the coil and a value of 0 (fluid) otherwise.
We then average the volume fraction $\rho'(x)$ in the grid points to $\rho(x)$ by either the mean of the grid-points values contained in the box of radius $r_{\rm REV}$, or by applying a Gaussian filter with a standard deviation of $r_{\rm REV}/2$.
The final porosity is then given by $1-\rho(x)$ in the points $x \in \mathcal{M}\cap\mathcal{M}_{\rm a}$.
For the averaging of the wire direction, we proceed by \cref{alg:direction_wire} (see \cref{app.Algo}).
We first extract the tangent vector $t_{{\rm c},k}$ of coil $k$ from its centreline.
This vector is then transformed into a sign-invariant shape tensor by taking the outer product $t_{{\rm c},k}\otimes t_{{\rm c},k}$.
The shape-tensor field $\mathbf{T}'(x)$ is then defined for each grid point by assigning either the closest shape-tensor field on the spline if the point is within the radius of coil $k$ or zero otherwise.
We then average each component of $\mathbf{T}'(x)$ in the same way as in \cref{alg:porosity_wire}, by employing a Gaussian of standard deviation of $r_{\rm REV}/2$ to obtain the averaged field $\mathbf{T}(x)$.
By extracting for each grid point the eigenvector of the largest absolute eigenvalue from $\mathbf{T}(x)$, we reconstruct the averaged wire direction.

\begin{figure}[tbp]
    \centering\small%
    \parbox{0.3\textwidth}{\centering Box filter with $r_{\rm REV} = 0.25$}
    \parbox{0.3\textwidth}{\centering Gaussian filter with $r_{\rm REV} = 0.25$}
    \parbox{0.3\textwidth}{\centering Gaussian filter with $r_{\rm REV} = 0.5$}
    \parbox{0.03\textwidth}{~}\\
    \includegraphics[width=0.3\linewidth, trim=28mm 30mm 56mm 30mm, clip]{  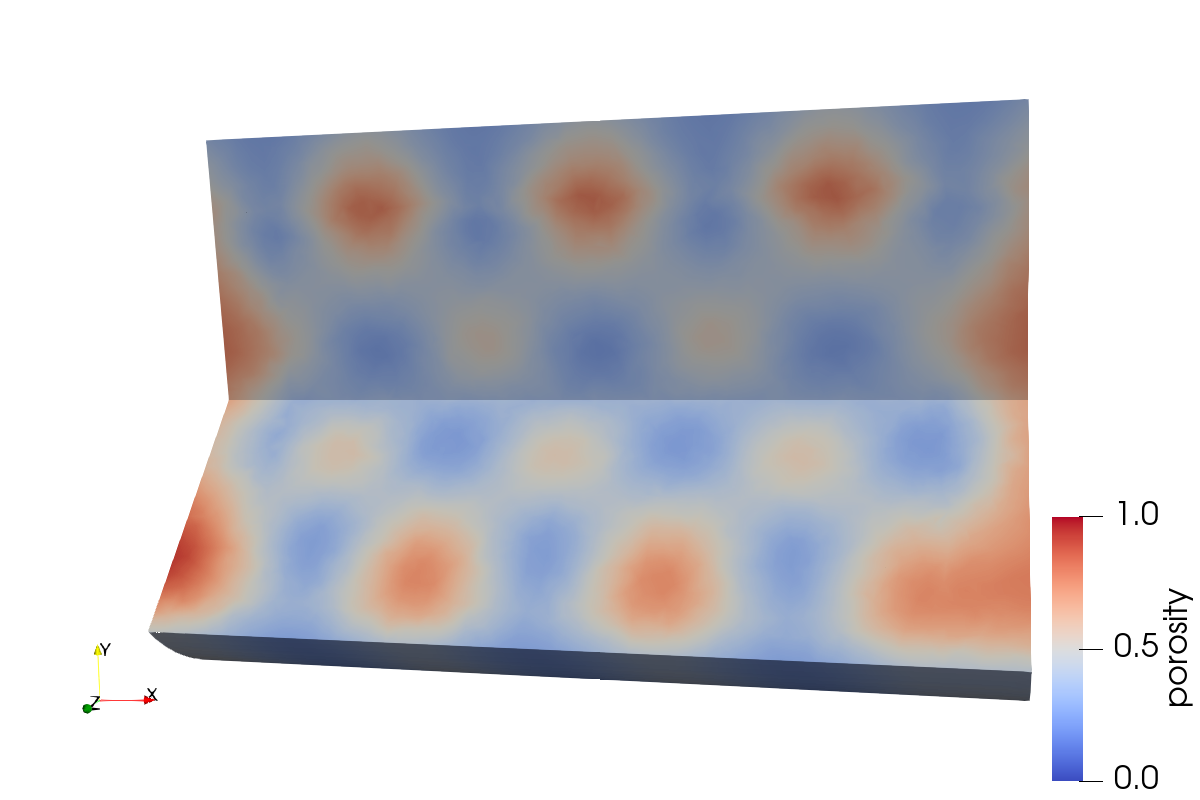}
    \includegraphics[width=0.3\linewidth, trim=28mm 30mm 56mm 30mm, clip]{  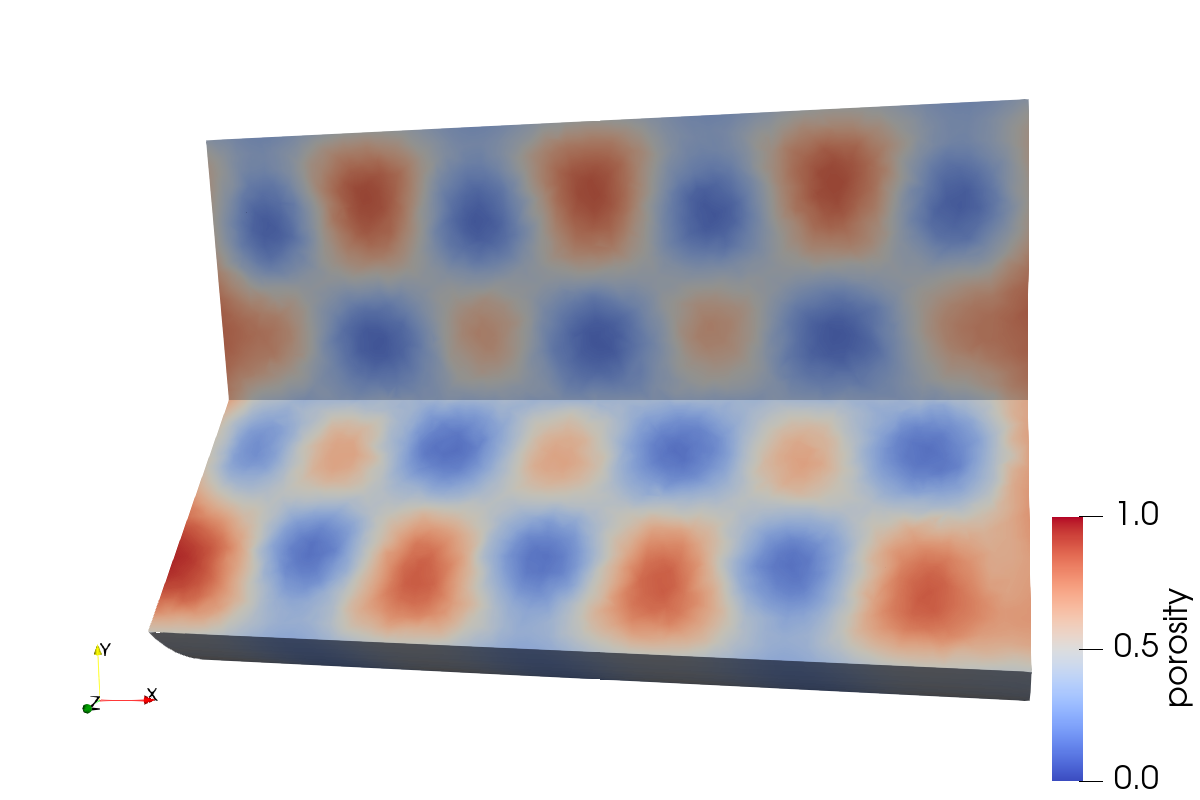}
    \includegraphics[width=0.3\linewidth, trim=28mm 30mm 56mm 30mm, clip]{  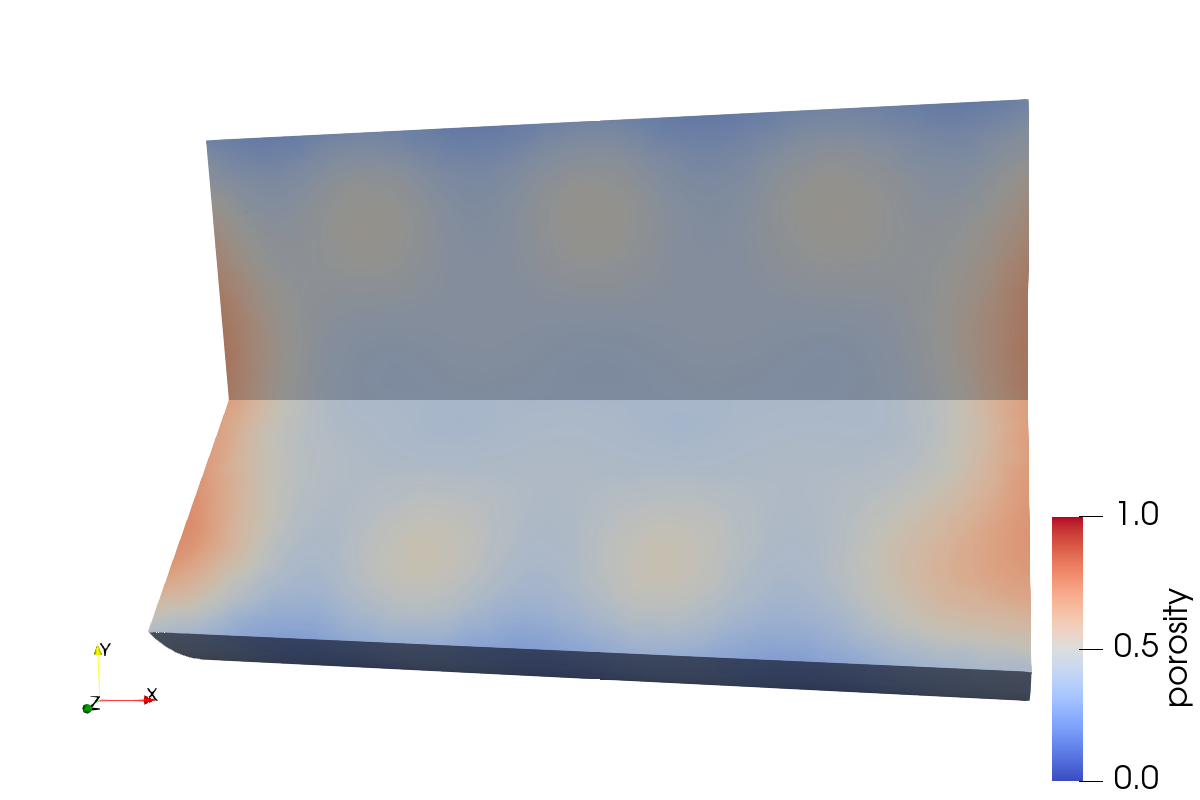}
    \includegraphics[width=0.05\linewidth, trim=37cm 0cm 14mm 17cm, clip]{  figures/cylinder_spiral_porosity_box_0.25_wall.png}
    \caption{Porosity fields of setting 2 generated by \cref{alg:porosity_wire} using $N_s=40$, $\rho_{\rm w} = 1$ and different filters.}
    \label{fig:case2:porosity}
\end{figure}

\begin{figure}[tbp]
    \centering\small%
    \parbox{0.03\textwidth}{~}
    \parbox{0.38\textwidth}{\centering $r_{\rm REV} = 0.25$}
    \parbox{0.38\textwidth}{\centering $r_{\rm REV} = 0.5$}
    \parbox{0.06\textwidth}{~}\\
    \includegraphics[width=0.035\linewidth, trim=28mm 20mm 365mm 150mm, clip]{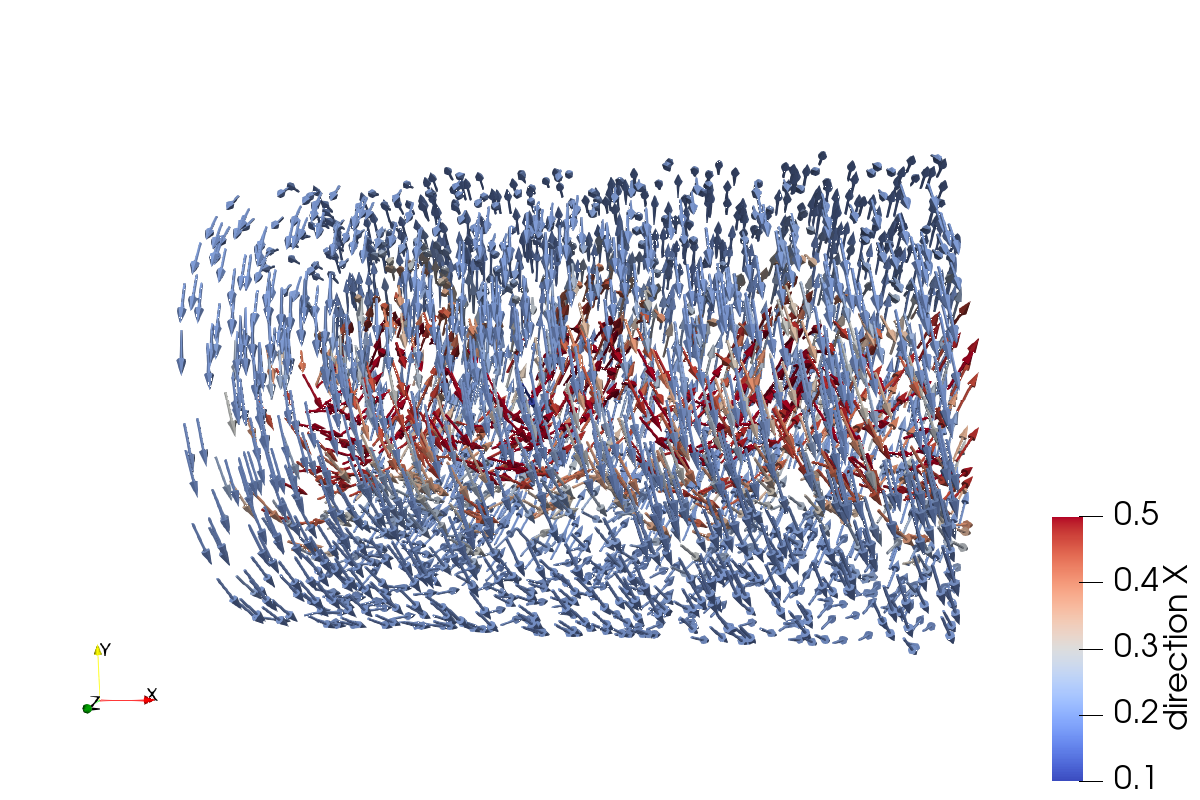}
    \includegraphics[width=0.38\linewidth, trim=60mm 45mm 56mm 50mm, clip]{figures/cylinder_spiral_direction_0.25.png}
    \includegraphics[width=0.38\linewidth, trim=60mm 45mm 56mm 50mm, clip]{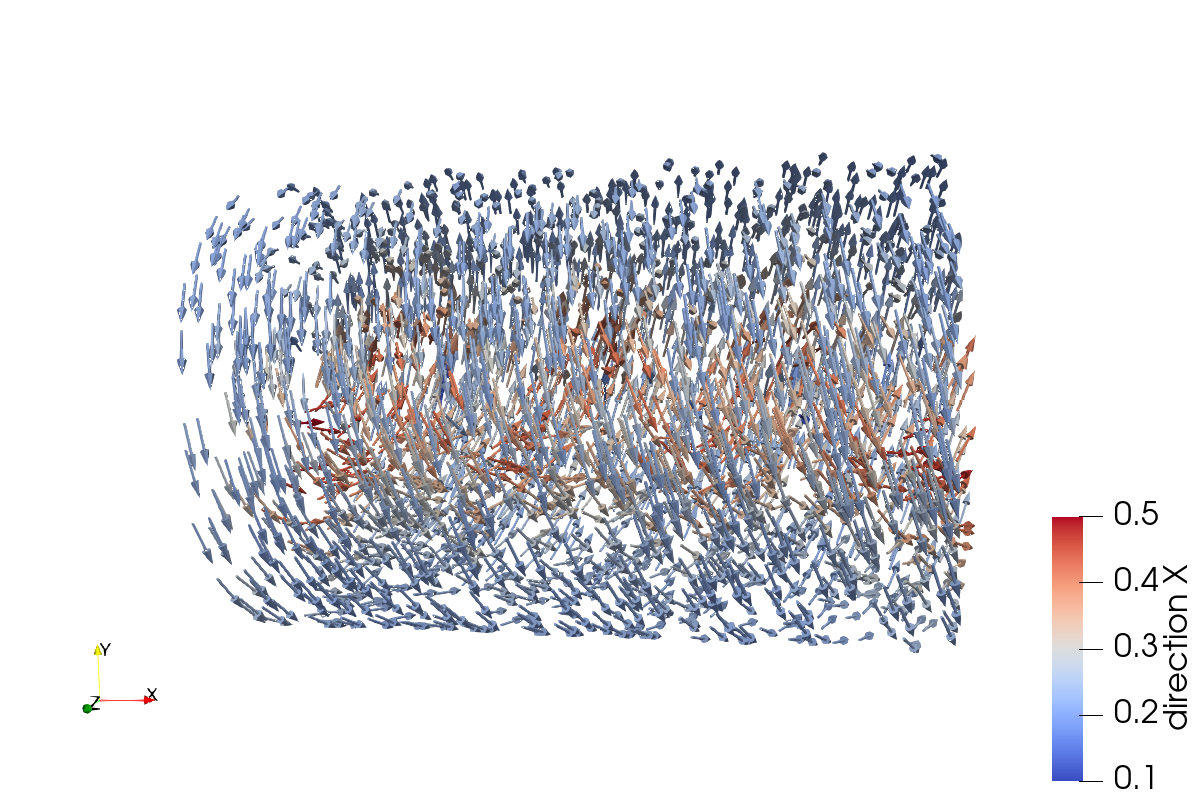}
    \includegraphics[width=0.06\linewidth, trim=37cm 0cm 4mm 17cm, clip]{figures/cylinder_spiral_direction_0.25.png}
    \caption{Wire-direction fields of setting 2 generated by \cref{alg:direction_wire} using $N_s=40$ and different averaging REV radii.}
    \label{fig:case2:wireDir}
\end{figure}

\begin{figure}[tbp]
    \centering
    \includegraphics[page=2, width=0.55\linewidth]{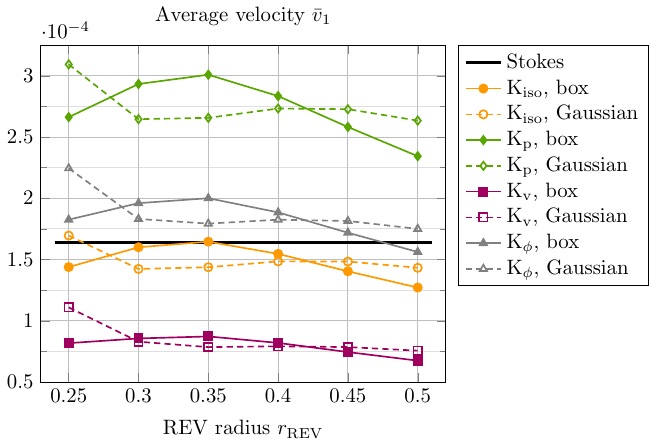}
    \includegraphics[page=4, width=0.4\linewidth, trim=0 0 33mm 0, clip]{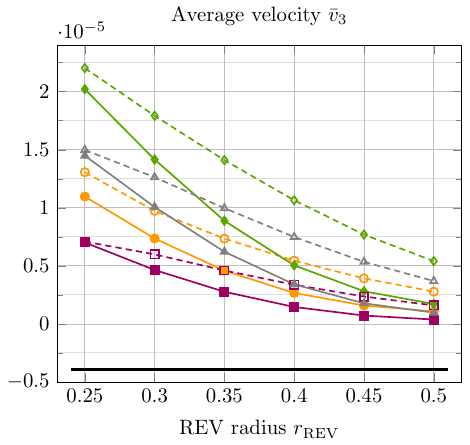}
    \caption{Average velocity components $\bar{v}_1$ (left) and $\bar{v}_3$ (right) in setting 2 for Stokes and Darcy flow problem (with different permeability parametrisations $\textbf{K}_{\bullet}$ and averaging using different filters with $\rho_{\rm w} = 1$ and $N_s=40$).}
    \label{fig:case2:average}
\end{figure}

The resulting porosity and wire-direction fields are exemplarily depicted in \cref{fig:case2:porosity,fig:case2:wireDir}.
With a small averaging radius of $r_{\rm REV} = 0.25$, i.e.~the wire radius of the spirals, microscale details are resolved, while larger averaging radii lead to a more uniform distribution.
Likewise, the Gaussian filter yields more pronounced porosity gradients than the box filter.
We observed that the averaging using $\rho_{\rm w} = 1$, i.e.~padding by solid outside the cylinder, yields the best results, while padding by the average porosity $\rho_{\rm w} = 1 - \bar{\phi}$ or fluid ($\rho_{\rm w} = 0$) over-predicts the porosity close to the boundary, hence leading to a large Darcy flow along the cylinder wall.

\begin{figure}[tbp]
    \centering\small%
    \parbox{0.3\textwidth}{\centering Stokes}
    \parbox{0.3\textwidth}{\centering isotropic Darcy ($\mathbf{K}_{\rm iso}$)}
    \parbox{0.3\textwidth}{\centering anisotropic Darcy ($\mathbf{K}_\phi$)}
    \parbox{0.05\textwidth}{}\\
    \includegraphics[width=0.3\linewidth, trim=28mm 30mm 56mm 35mm, clip]{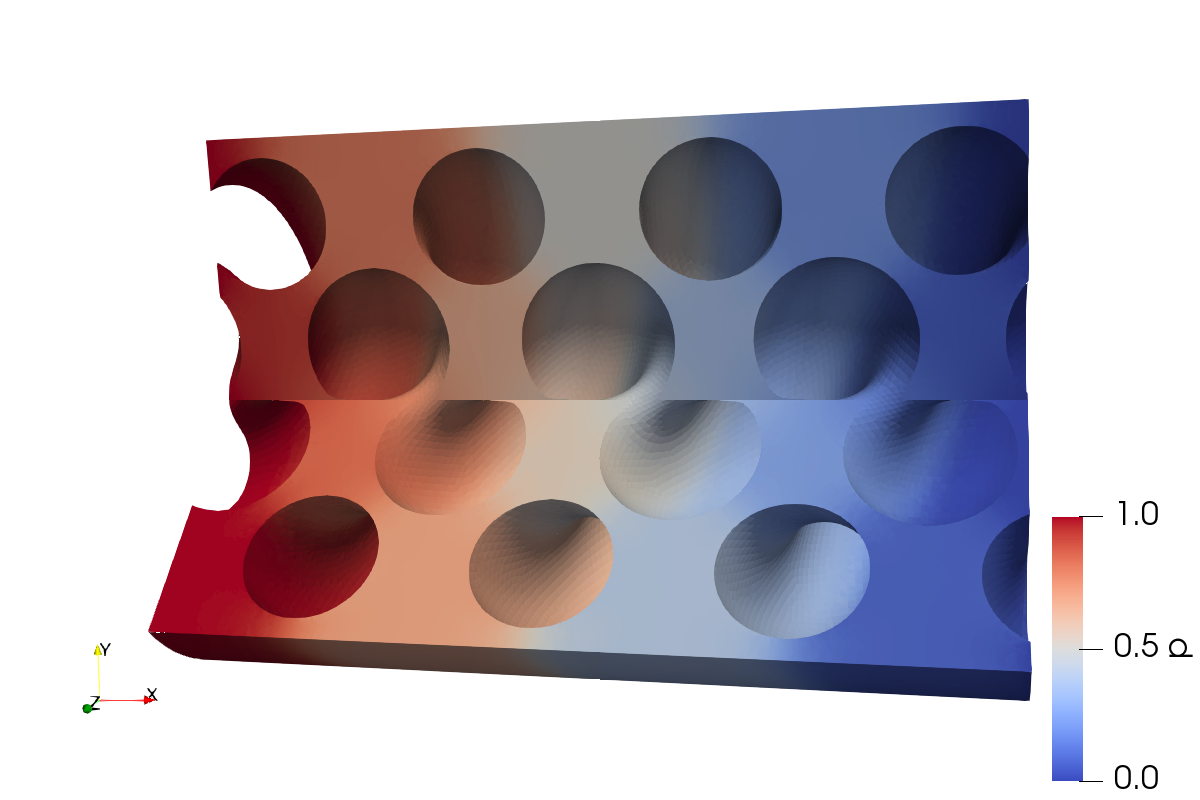}
    \includegraphics[width=0.3\linewidth, trim=28mm 30mm 56mm 35mm, clip]{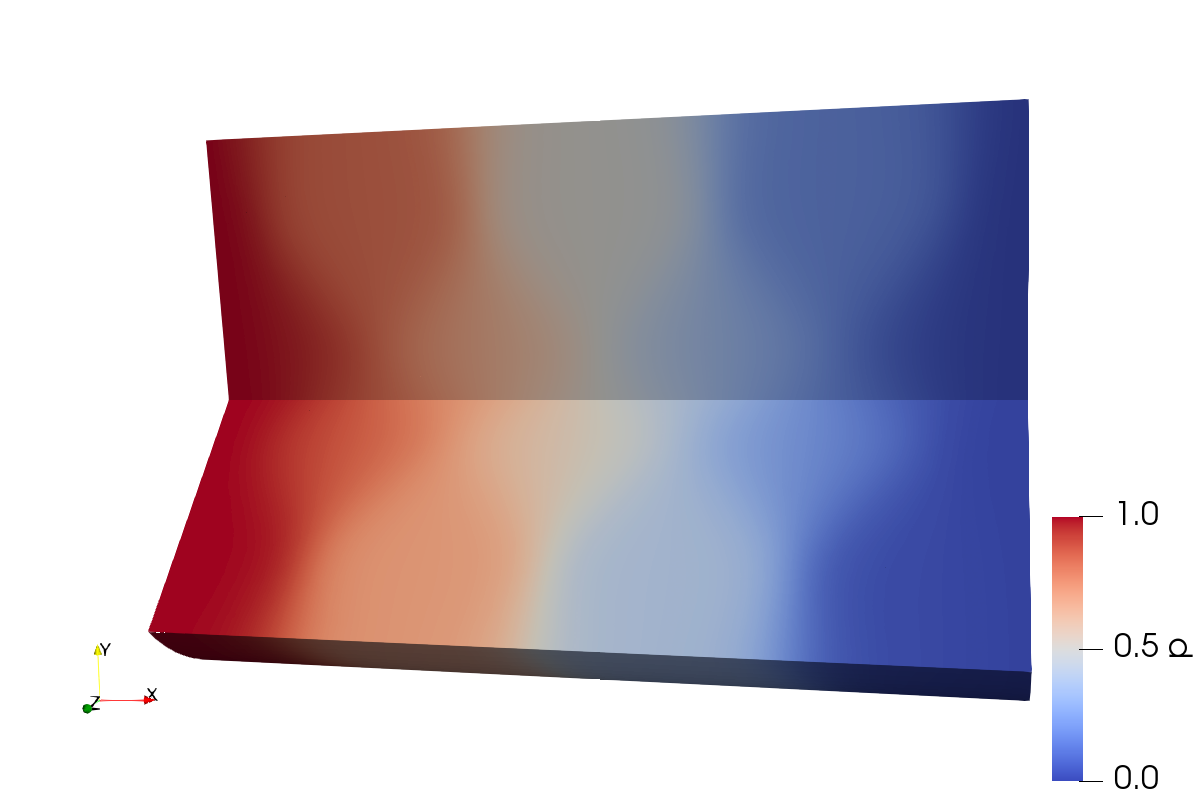}
    \includegraphics[width=0.3\linewidth, trim=28mm 30mm 56mm 35mm, clip]{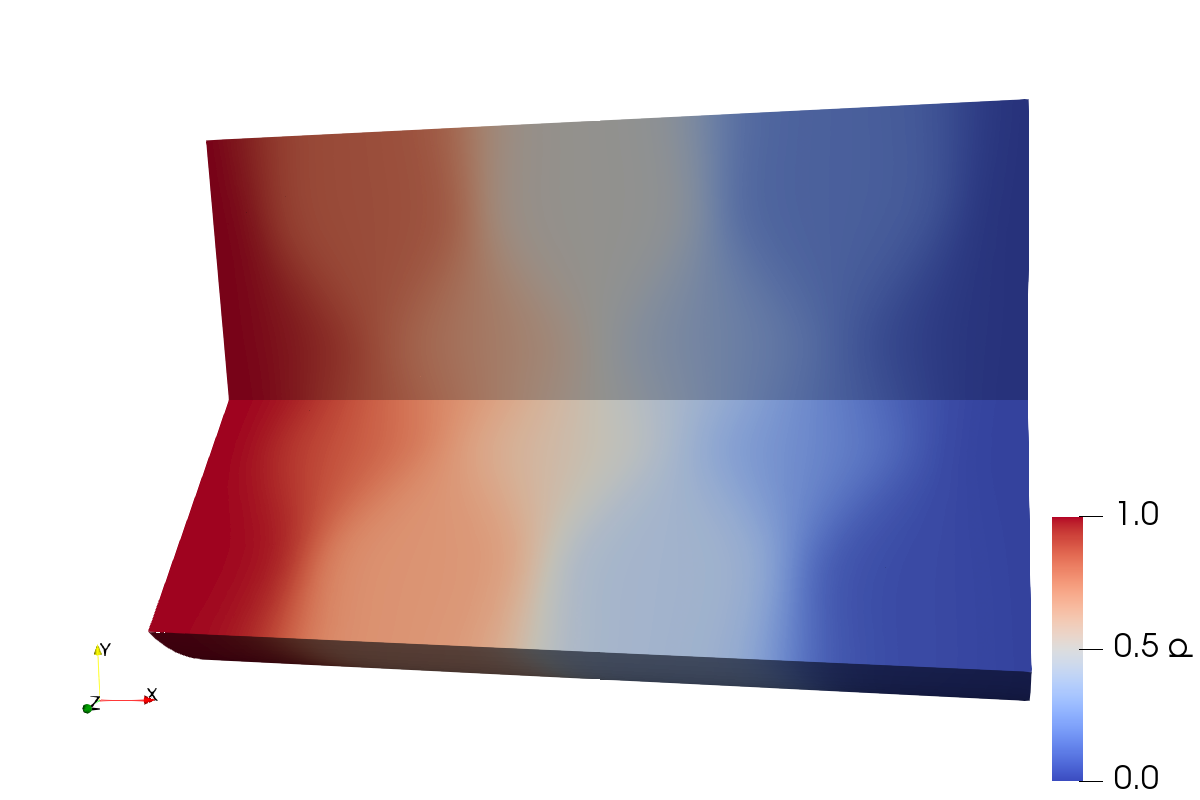}
    \includegraphics[width=0.05\linewidth, trim=37cm 0cm 14mm 17cm, clip]{figures/stokes-cylinder_spiral_pressure.png}

    \includegraphics[width=0.03\linewidth, trim=28mm 30mm 435mm 150mm, clip]{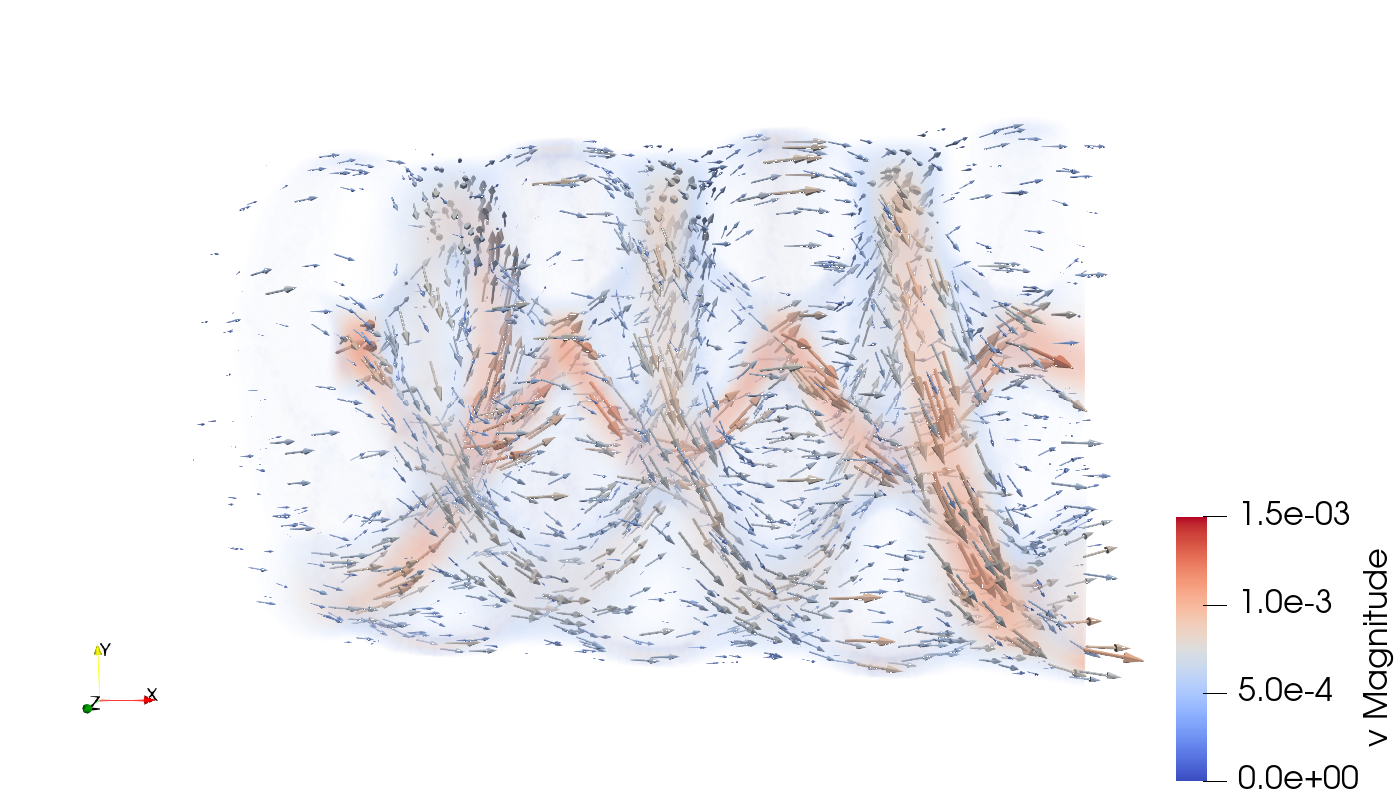}
    \hspace{-5mm}
    \includegraphics[width=0.305\linewidth, trim=65mm 35mm 93mm 35mm, clip]{figures/stokes-cylinder_spiral_velocity2.png}
    \hspace{-3mm}
    \includegraphics[width=0.305\linewidth, trim=65mm 35mm 93mm 35mm, clip]{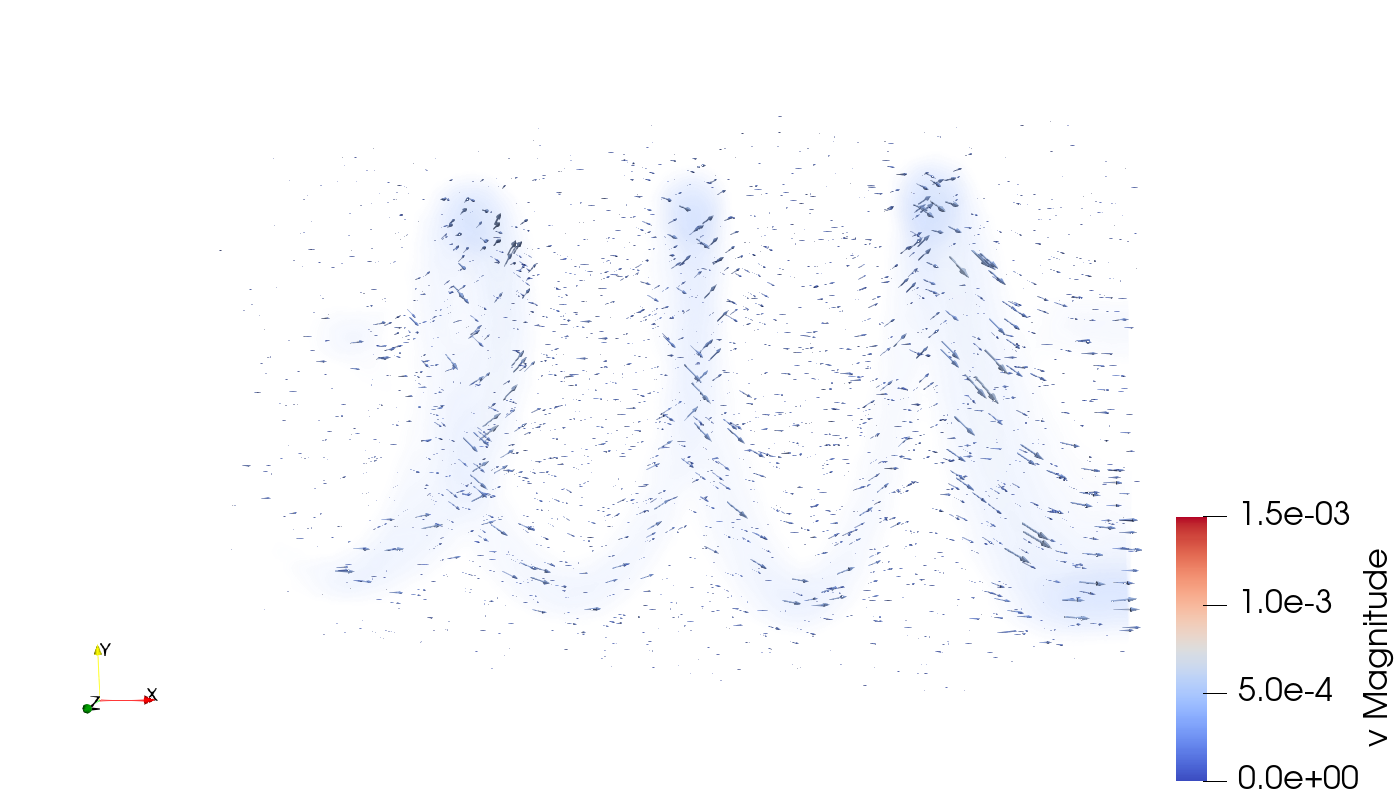}
    \includegraphics[width=0.305\linewidth, trim=65mm 35mm 93mm 35mm, clip]{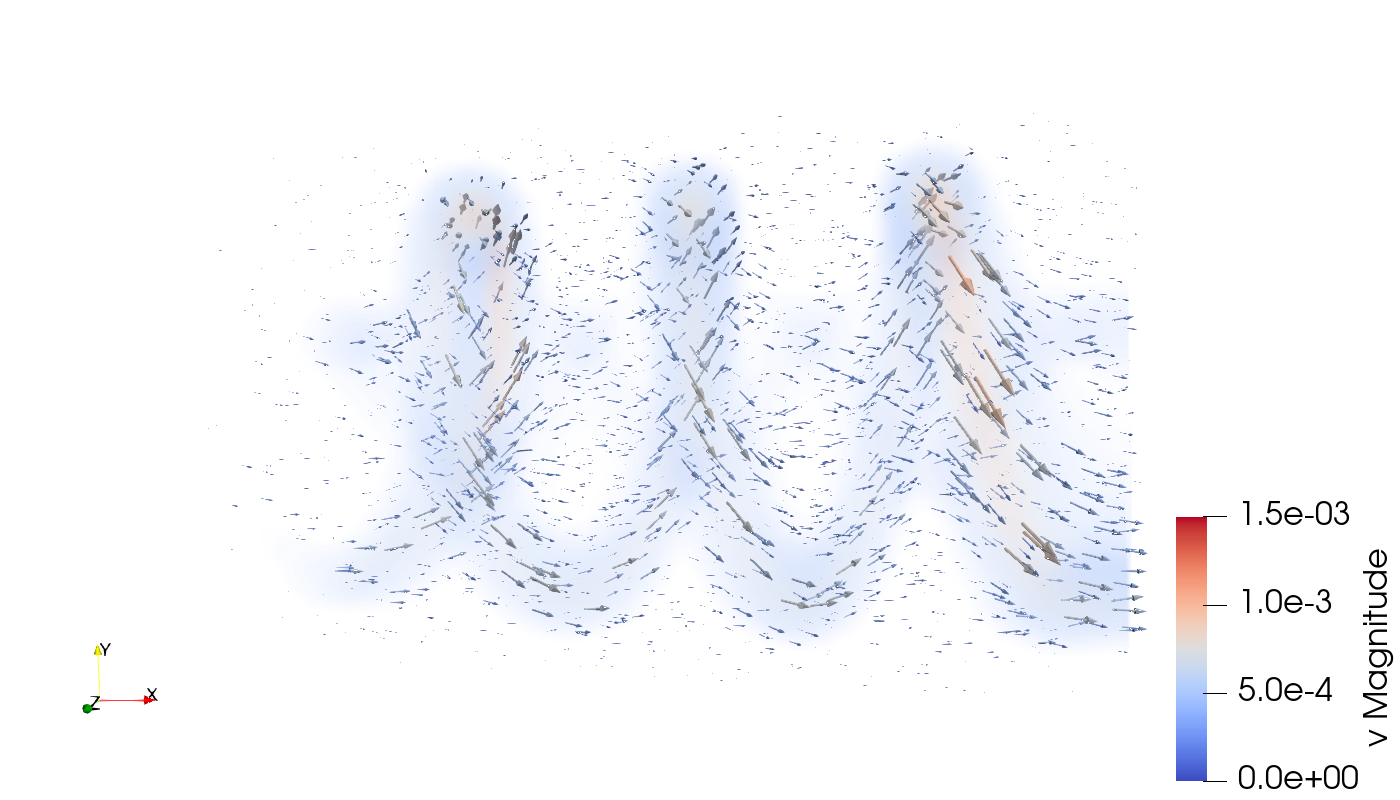}
    \includegraphics[width=0.06\linewidth, trim=415mm 0cm 14mm 17cm, clip]{figures/stokes-cylinder_spiral_velocity2.png}
    \caption{Pressure (top line) and velocity (bottom line) solutions for setting 2. For the averaging in the Darcy cases, \cref{alg:porosity_wire,alg:direction_wire} are used with a Gaussian filter of REV radius of $r_{\rm REV} = 0.35$, $N_s=40$ and $\rho_{\rm w} = 1$.}
    \label{fig:case2:solution}
\end{figure}

Overall, the effects of the different filters and their REV sizes on the average velocity is minor, as visible in \cref{fig:case2:average} for the average velocity components $\bar{v}_1$ and $\bar{v}_3$, while the average velocity component $\bar{v}_2$ is always below $10^{-7}$.
Note that the large deviations in the third component are likely due to the finite-size effects of the Neumann boundary conditions for Stokes flow combined with the boundary layers at in- and out-flow of the averaging (cf. \cref{fig:case2:porosity}), which decrease for increasing REV radius.
As before, the parametric cases ($\mathbf{K}_{\rm p}$ and $\mathbf{K}_{\rm v}$) significantly over- and under-predict the values, while the isotropic ($\mathbf{K}_{\rm iso}$) and the weighted average ($\mathbf{K}_{\phi}$) cases match quite well.
Comparing the pressure and velocity solutions of the Stokes problem and the latter two Darcy cases, as depicted in \cref{fig:case2:solution}, one observes close correspondence between all pressure fields, but a more localised flow pattern for the anisotropic Darcy case, similar to the Stokes flow.

\subsection{Setting 3: Realistic coil in an artificial aneurysm}

Finally, we consider an artificial aneurysm occluded by the coil presented in \cref{fig:coil}, resulting in the geometry depicted in \cref{fig:case3:param} (left).
This geometry does not allow for a fully resolved Stokes simulation on a typical computer due to the excessive number of DOFs required for an appropriate three-dimensional mesh. In particular, automated meshing tools, such as gmsh used here, easily fail due to coil--coil and coil--aneurysm intersections for triangulations at relatively ``low'' resolutions ($h\approx 0.03$, aiming for about 10 million mesh elements).
Even if they succeed, the occurrence of any bad-shaped elements impedes the convergence of iterative solvers necessary at that resolution, thus significantly increasing the computational time.
Nevertheless, we can apply the averaging strategy and resolve the Darcy flow,
as the domain only consists of the aneurysm.
To this end, we apply the averaging \cref{alg:porosity_wire,alg:direction_wire} (see \cref{app.Algo}) using $N_s=40$, a Gaussian filter with REV radius $r_{\rm REV} = 0.4$ and solid padding ($\rho_{\rm w} = 1$), see \cref{fig:case3:param} (right) for the result.
For the permeability parametrisation, we choose $\mathbf{K}_\phi$ based on the previously observed good results, together with $R = 0.1524$, which is the average of the two coil radii.
The tetrahedral mesh of the aneurysm for the Darcy problem is generated using gmsh, see \cref{tab:mesh} for details.
The resulting solution is depicted in \cref{fig:case3:darcy}.
The main features of the coil structure are recovered by the porosity field, and thus also the Darcy solution, visible in the heterogeneous pressure field and low velocities in densely packed coil regions.

\begin{figure}[htbp]
    \centering
    \includegraphics[width=0.35\linewidth, trim=33mm 3mm 51mm 5mm, clip]{  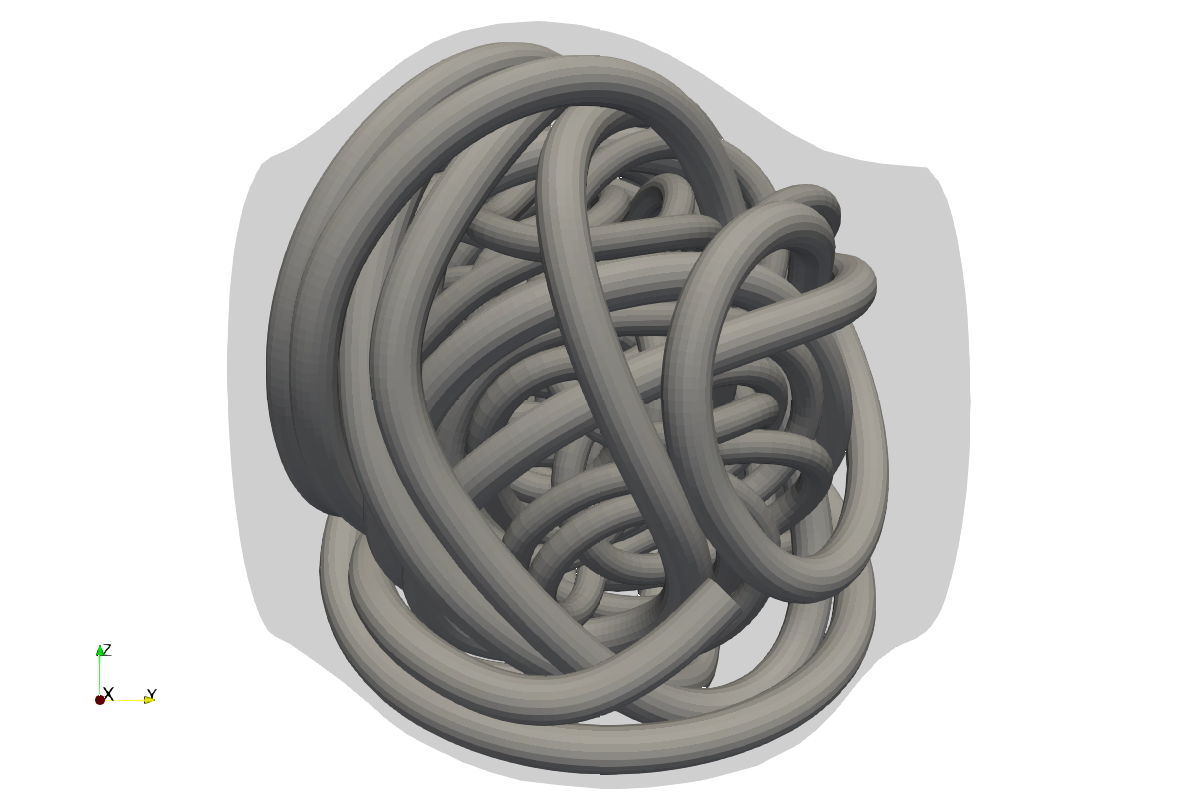}
    \includegraphics[width=0.35\linewidth, trim=80mm 3mm 4mm 5mm, clip]{  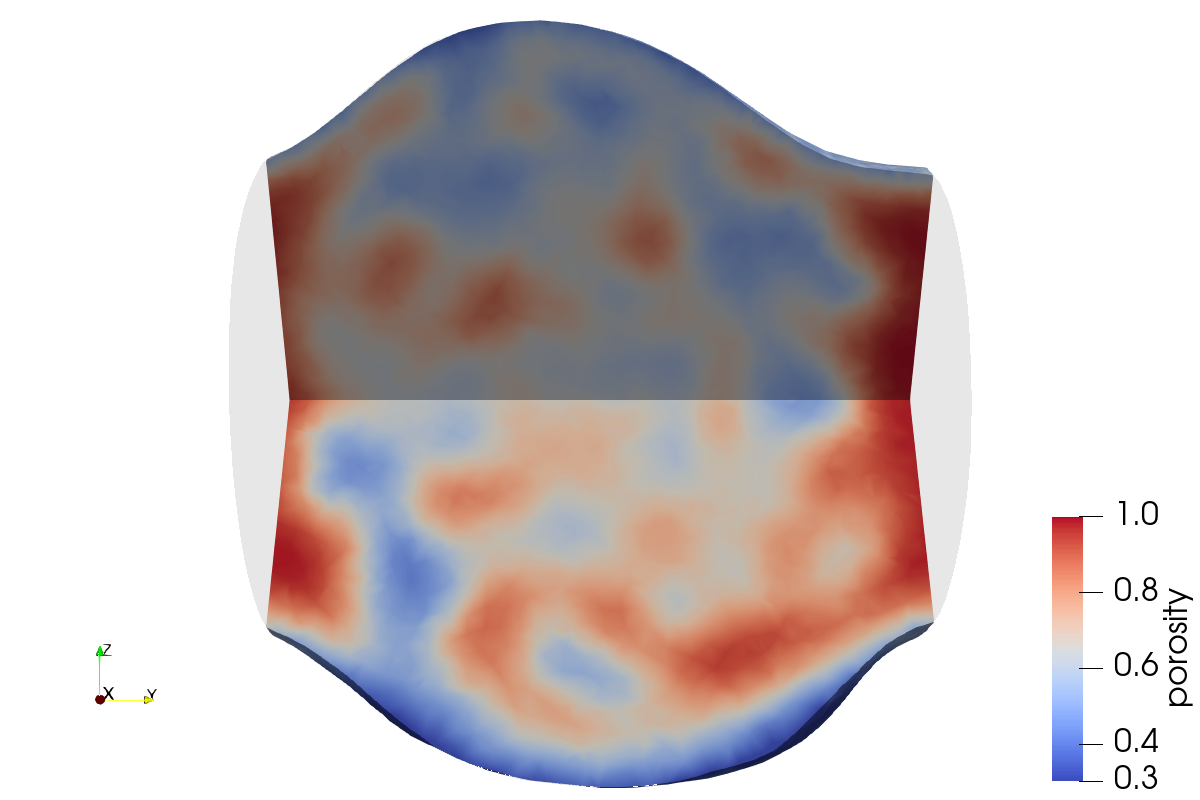}
    \caption{Left: Geometry of a coil in an artificial aneurysm. Right: Averaged porosity field generated by \cref{alg:porosity_wire}.}
    \label{fig:case3:param}
\end{figure}

\begin{figure}[htbp]
    \centering
    \includegraphics[width=0.04\linewidth, trim=33mm 23mm 360mm 100mm, clip]{  figures/artificial_vessel_geometry.png}
    \hspace{-6mm}
    \includegraphics[width=0.32\linewidth, trim=80mm 3mm 4mm 5mm, clip]{  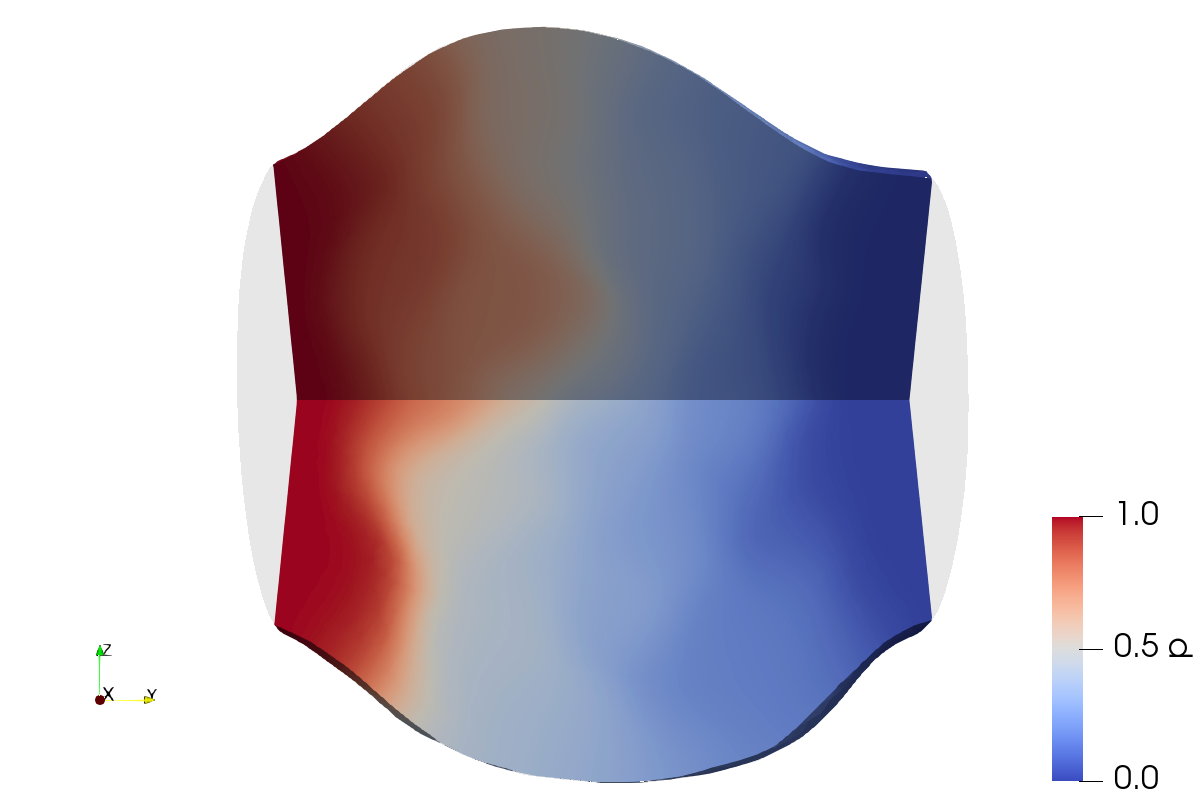}
    \includegraphics[width=0.32\linewidth, trim=80mm 3mm 4mm 5mm, clip]{  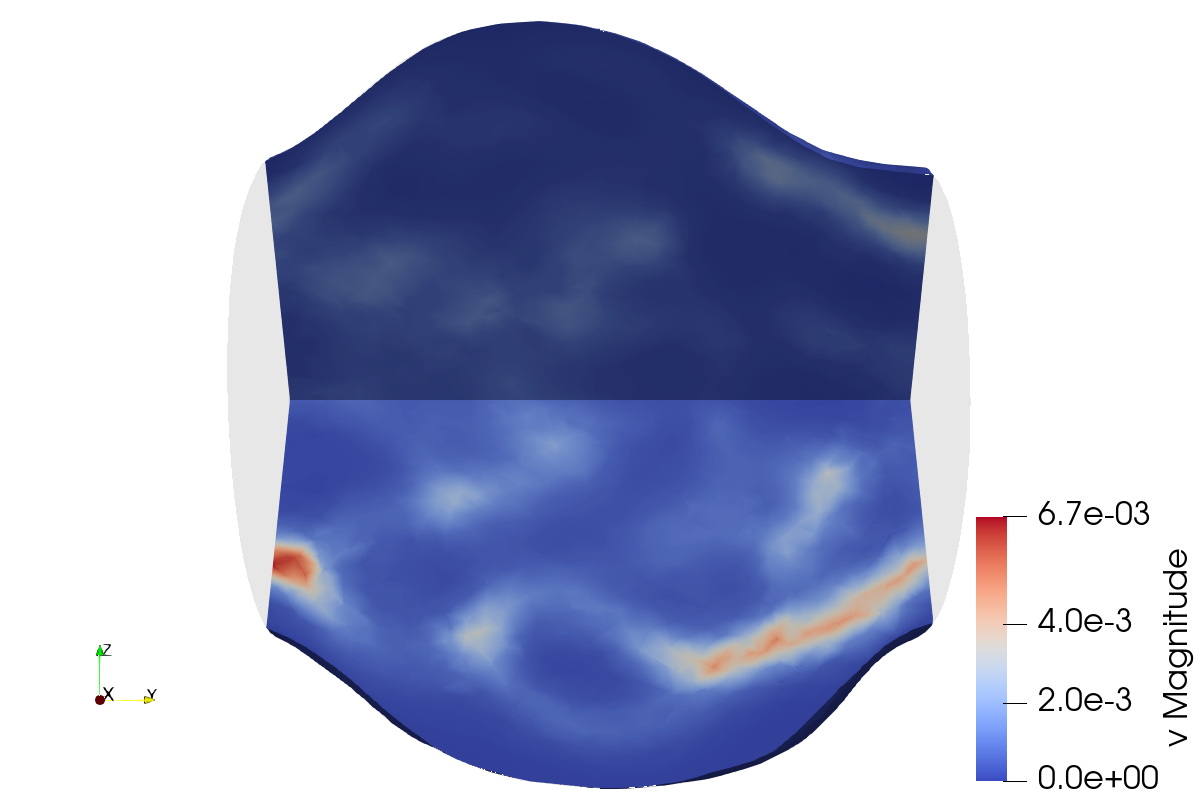}
    \includegraphics[width=0.32\linewidth, trim=80mm 3mm 4mm 5mm, clip]{  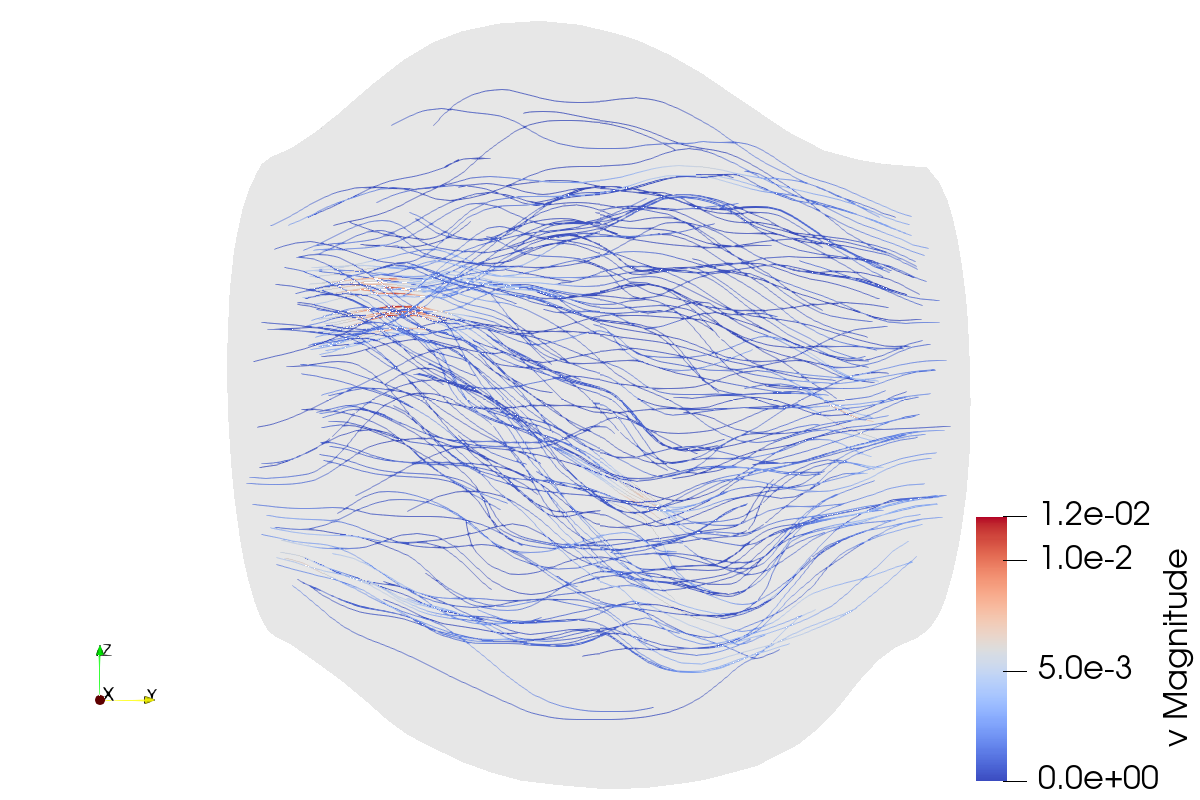}
    \caption{Pressure (left), velocity magnitude (centre) and streamlines (right) of the anisotropic Darcy simulation.}
    \label{fig:case3:darcy}
\end{figure}

\section{Conclusion}

This study introduces a comprehensive multiscale framework for determining effective permeabilities in anisotropic microscopic geometries containing dense, fibre-like obstacles, with high relevance to flow prediction in coiled aneurysm. By combining homogenisation theory with simulation in Representative Elementary Volumes, we obtain permeability tensors that capture directional effects imposed by the underlying microstructures. While commonly used isotropic
models can result in significant inaccuracies, anisotropic permeability tensors
can improve macroscopic simulation results significantly.  Darcy flow simulations show that anisotropy 
substantially influences macroscopic flow patterns and pressure distributions.
Applying the methodology to a realistic coil  geometry further demonstrates its flexibility and predictive capability. Overall, our approach provides a robust pathway for linking microscopic coil arrangements to clinically meaningful flow characterisations, and offers a flexible foundation which can be extended to other application fields such as fibre-reinforced materials, fractured porous media and biomedical flow configurations.

\appendix

\section{Algorithms for homogenisation by spatial averaging}
\label{app.Algo}
In this section, we present the \cref{alg:rev_determination,alg:rev_generatiom,alg:porosity_wire,alg:direction_wire} employed in this study and cited in the text.

\begin{algorithm}[H]
    \caption{Representative Elementary Volumes (REVs) generation}
    \label{alg:rev_generatiom}
    \begin{algorithmic}[1]
    \Require Surface mesh \texttt{.obj} of the coil.  \Comment{Made of quadrilateral elements.}
    \Ensure Tetrahedral volume mesh of an REV.  \Comment{Readable by a finite element software e.g. FreeFEM.}
    \State Read the surface mesh \texttt{.obj} of the coil with \texttt{blender}.
    \State Select an REV of size 1 by computing the intersection of the coil with a cube of size 1. Re-mesh the REV with appropriate mesh size. \Comment{At this stage, we have a surface mesh (made of quadrilateral elements) of the coil inside a unit cell.}  
    \State Export the REV in format \texttt{.obj}.
    \State Open the REV mesh with \texttt{meshlab}. Turn the mesh into a pure triangular mesh and re-mesh the coil with appropriate mesh size. Remove spurious components of the mesh and export each component in format \texttt{.obj}. \Comment{Surface mesh made of triangles.}
    \State Convert each mesh component \texttt{.obj} in a surface mesh \texttt{.mesh} using the python plugin \texttt{meshio} \cite{meshio}. \newline \Comment{Format \texttt{.mesh} readable by FreeFEM.}
    \State Import the surface mesh of each component in FreeFEM and glue then together to create a global surface mesh $\mathcal{C}_{h,S}$.
    \State Create the surface mesh of the unit box (oversampled with a ratio $\kappa=1.1$) around the obstacles, denoted $\mathcal{B}_{h,S}$. \Comment{Ensure that the mesh is compatible with periodic boundary conditions.}
    \State Combine $\mathcal{C}_{h,S}$ and $\mathcal{B}_{h,S}$ as $\mathcal{T}_{h,S}=\mathcal{B}_{h,S} \cup \mathcal{C}_{h,S}$.
    \State Create the three-dimensional volume mesh $\mathcal{T}_{h}$ (composed of tetrahedra) of $\mathcal{T}_{h,S}$, i.e.~filling the space between the bounding box and the coil with tetrahedra, using \texttt{TetGen} \cite{tetgen}.
    \end{algorithmic}
\end{algorithm}

\begin{algorithm}[H]
\caption{Uniform mesh generation}
\label{alg:rev_determination}
\begin{algorithmic}[1]
\Require Aneurysm mesh $\mathcal{M}_{\rm a}$; REV radius $r_{\rm REV}>0$; sampling rate $N_s\in \mathbb{N}_+$.
\Ensure Uniform mesh $\mathcal{M}$.



\State Compute the bounding box centre $x_{\rm c}$ and length vector $L$ of $\mathcal{M}_{\rm a}$
\State $\Delta x = \min\{L_1, L_2, L_3\} / (2N_s)$ \Comment{mesh spacing}
\ForAll{$i \in \{1,2,3\}$}
    \State $N_i = \left\lceil (L_i/2 + r_{\rm REV}) / \Delta x\right\rceil$
    \State $\mathcal{S}_i = \left\{ x_{{\rm c}, i} + (n+1/2) \Delta x \ \middle|\ n = -N_i, \ldots, N_i-1 \right\}$
\EndFor
\State $\mathcal{M} = \mathcal{S}_1 \times \mathcal{S}_2 \times \mathcal{S}_3$
\end{algorithmic}
\end{algorithm}

\begin{algorithm}[H]
\caption{Averaging procedure for the porosity field}
\label{alg:porosity_wire}
\begin{algorithmic}[1]
\Require Coil mesh $\mathcal{M}_{\rm c}$, aneurysm mesh $\mathcal{M}_{\rm a}$, and vessel mesh $\mathcal{M}_{\rm v}$ with $\mathcal{M}_{\rm c} \subseteq \mathcal{M}_{\rm a} \subseteq \mathcal{M}_{\rm v}$; wall boundary condition $\rho_{\rm w} \in [0, 1]$; REV radius $r_{\rm REV}$; and uniform mesh $\mathcal{M}$ from \cref{alg:rev_determination}.
\Ensure Porosity field $\phi(x)$ for $x\in \mathcal{M} \cap \mathcal{M}_{\rm a}$.

\LComment{\textbf{Initialisation and boundary conditions for solid fraction}}
\ForAll{$x \in \mathcal{M}$}
    \State $\rho'(x) = \begin{cases}
            \rho_{\rm w} & \text{if } x \in \mathcal{M}_{\rm v}^\complement \text{ (wall BC)} \\
            1            & \text{if } x \in \mathcal{M}_{\rm c} \text{ (coil region)} \\
            0            & \text{otherwise (fluid region)}
        \end{cases} $
\EndFor

\LComment{\textbf{Apply averaging by convolution over local REV neighbourhood (box filter)}}
\ForAll{$x \in \mathcal{M} \cap \mathcal{M}_{\rm a}$}
    \State $\mathcal{N}(x) = \big\{ x' \in \mathcal{M} \ \big|\ \|x - x'\|_\infty \le r_{\mathrm{REV}} \big\}$
    \State $\displaystyle\rho(x) = \frac{1}{|\mathcal{N}(x)|}
        \sum_{x' \in \mathcal{N}_{\rm REV}(x)} \rho'(x')$
    \LComment{Alternative Gaussian filter with $\sigma = r_{\rm REV}/2$:\\
     $\displaystyle
         \rho(x) =
         \left(\sum_{x' \in \mathcal{M}} \rho'(x') \exp\left(-\tfrac{\|x - x'\|_2^2}{2\sigma^2}\right) \right) \Big/ \left(\sum_{x' \in \mathcal{M}} \exp\left(-\tfrac{\|x - x'\|_2^2}{2\sigma^2}\right) \right)$}
    \State $\phi(x) = 1 - \rho(x)$
\EndFor
\end{algorithmic}
\end{algorithm}

\begin{algorithm}[H]
\caption{Averaging procedure for the direction field}
\label{alg:direction_wire}
\begin{algorithmic}[1]
\Require $K$ pairs of coil meshes $\mathcal{M}_{{\rm c},k}$ and coil centrelines $\mathcal{C}_k$, and aneurysm mesh $\mathcal{M}_{\rm a}$ such that $\mathcal{M}_{{\rm c},k} \subseteq \mathcal{M}_{\rm a}$ for all $k = 1, \dots, K$; REV radius $r_{\rm REV}$; and uniform mesh $\mathcal{M}$ from \cref{alg:rev_determination}.
\Ensure Tangent field $t(x)$ for $x\in \mathcal{M} \cap \mathcal{M}_{\rm a}$.

\LComment{\textbf{Generate shape tensor field}}
\ForAll{$x \in \mathcal{M}$}
    \State $\mathbf{T}'(x) = \mathbf{0}$
\EndFor
\For{$k = 1, \dots, K$}
    \State Generate tangent vector $t_{{\rm c}, k}(s)$ along the centreline $s \in \mathcal{C}_k$
    \ForAll{$x \in \mathcal{M} \cap \mathcal{M}_{{\rm c},k}$}
        \State Find $s_{x} = \operatorname{argmin}_{s \in \mathcal{C}_k} \|x - s\|_2$
        \State $\mathbf{T}'(x) = t_{{\rm c},k}(s_{x}) \otimes t_{{\rm c},k}(s_{x})$
    \EndFor
\EndFor

\LComment{\textbf{Average shape tensor field}}
\ForAll{$x \in \mathcal{M} \cap \mathcal{M}_{\rm a}$}
    \State Apply Gaussian-weighted averaging with $\sigma = r_{\rm REV} / 2$:
    \Statex $\displaystyle
        \mathbf{T}(x) = \left(\sum_{x' \in \mathcal{M}} \exp\left(-\tfrac{\|x - x'\|_2^2}{2\sigma^2}\right) \mathbf{T}'(x') \right) \Big/ \left(\sum_{x' \in \mathcal{M}} \exp\left(-\tfrac{\|x - x'\|_2^2}{2\sigma^2}\right)\right)$
    \LComment{\textbf{Extract wire direction}}
    \State Compute normalised eigenvector $v(x)$ corresponding to the largest absolute eigenvalue of $\textbf{T}(x)$
    \State $t(x) = v(x)$
\EndFor
\end{algorithmic}
\end{algorithm}

\section{Sensitivity analysis of the wire direction}
\label{sec.SensWire}
A sensitivity analysis of \cref{alg:direction_wire} with respect to the sampling rate $N_s \in \{30,\,40,\,50,\,60\}$ is depicted in \cref{fig:SamplingSensitivity}. 
For each pair of successive rates ($N_s \rightarrow N_s'$), we compute the relative change in porosity and the relative Frobenius norm of the difference in the permeability tensor, both with respect to the coarser grid.
These quantities are evaluated for all REVs, and the mean and standard deviation across all REVs is reported.
\begin{figure}[htbp]
    \centering
    \include{figures/PLOTS/SamplingRateSensitivity}
    \caption{REV resolution sensitivity. Mean relative changes and corresponding standard deviation in porosity $\phi$ and permeability $\mathbf{K}$ between successive refinements ($N_s \rightarrow N_s'$) over all REVs.}
    \label{fig:SamplingSensitivity}
\end{figure}
We attribute the larger variation in permeability mainly to the discrete grid representation of the geometry. While the directional field is derived from a continuous centreline, the geometry itself is approximated on a Cartesian grid, resulting in a staircase-like interface. Consequently, small pores and narrow gaps may open or close depending on the resolution. This has only a minor effect on porosity, but a strong impact on permeability, as the underlying analytical relations exhibit near-cancellation of leading-order terms, making permeability highly sensitive to small geometric changes.

For the considered application in endovascular coiling, the level of variation at $N_s=40$ is acceptable, as it remains within the typical uncertainty range encountered in medical applications, including imaging, measurement and modelling uncertainties.

\section{Sensitivity analysis of the computed permeability tensor}
\label{sec.Sens}

To assess the sensitivity, and thus the robustness, of our approach with respect to inaccuracies in the rotation matrix $\mathbf{R}_\theta$ given by \cref{alg:direction_wire}, we apply small perturbations to this rotation matrix and analyse their influence. To this end, let $\mathbf{R}$  be a rotation matrix and $\mathbf{K}_{\rm eff}$ a diagonal tensor. As is standard in sensitivity analyses of rotations, we model small perturbations using a skew-symmetric matrix representation, 
\begin{equation*}
    \boldsymbol{\widehat{\omega}} = \begin{pmatrix}
        0 & - \omega_3 & \omega_2 \\
        \omega_3 & 0 & -\omega_1 \\
        -\omega_2 & \omega_1 & 0
    \end{pmatrix}.
\end{equation*}
Then, for a small parameter $\eta>0$, define the perturbed rotation matrix $\mathbf{R}_\eta$ as
\begin{equation*}
    \mathbf{R}_\eta = \mathbf{R}(\mathbf{I}_3+\eta  \boldsymbol{\widehat{\omega}}).
\end{equation*}
It should be noted that $\mathbf{R}_\eta$ remains orthogonal up to first order in $\eta$. Now, let 
\begin{equation*}
  \mathbf{K}=\mathbf{R} \mathbf{K}_{\rm eff} \mathbf{R}^\top,
\end{equation*}
and consider the perturbed permeability tensor,
\begin{equation*}
    \mathbf{K}_\eta = \mathbf{R}_\eta  \mathbf{K}_{\rm eff}  \mathbf{R}_\eta^\top  =  \mathbf{R} (\mathbf{I}_3+\eta \boldsymbol{\widehat{\omega}} )  \mathbf{K}_{\rm eff}  (\mathbf{I}_3+\eta \boldsymbol{\widehat{\omega}} )^\top \mathbf{R}^\top .
\end{equation*}
Since $ \boldsymbol{\widehat{\omega}}^\top = -  \boldsymbol{\widehat{\omega}} $, this becomes 
\begin{equation*}
    \mathbf{K}_\eta =\mathbf{R} (\mathbf{I}_3+\eta \boldsymbol{\widehat{\omega}} )  \mathbf{K}_{\rm eff}  (\mathbf{I}_3-\eta \boldsymbol{\widehat{\omega}} ) \mathbf{R}^\top .
\end{equation*}
Now expand, 
\begin{equation*}
   \mathbf{K}_\eta  = \mathbf{K} + \eta\mathbf{R}( \boldsymbol{\widehat{\omega}}\mathbf{K}_{\rm eff}-\mathbf{K}_{\rm eff}\boldsymbol{\widehat{\omega}})\mathbf{R}^\top+ O(\eta^2).
\end{equation*}
At the end, the first-order change is
\begin{equation*}
   \delta \mathbf{K} = \eta \boldsymbol{R}[\boldsymbol{\widehat{\omega}},\mathbf{K}_{\rm eff}  ]\boldsymbol{R}^\top,
\end{equation*}
where $[\boldsymbol{\widehat{\omega}},\mathbf{K}_{\rm eff}  ] = \boldsymbol{\widehat{\omega}}\mathbf{K}_{\rm eff} - \mathbf{K}_{\rm eff}\boldsymbol{\widehat{\omega}}$ is the commutator. By denoting  $\mathbf{K}_{\rm eff}={\rm diag}(K_1, K_2, K_3)$, it follows
\begin{equation*}
   [\boldsymbol{\widehat{\omega}},\mathbf{K}_{\rm eff}  ]_{ij} = (K_j - K_i) \boldsymbol{\widehat{\omega}}_{ij}.
\end{equation*}
This last equality shows that the sensitivity depends on the difference $K_i-K_j$. In other words, if the eigenvalues are close then the perturbation has small effect; on the contrary, if the eigenvalues are widely separated then the perturbation has a larger effect. 

For a numerical illustration, we incorporate the small parameter $\eta$ directly into the matrix $\boldsymbol{\widehat{\omega}}$, and denote the resulting matrix by $\boldsymbol{\widehat{\omega}}_\eta$. Let $\omega_\eta = (\omega_1, \omega_2, \omega_3)$. We choose $\omega_\eta=\eta u$, with $u$ a direction unit vector. Also instead, of considering $\mathbf{R}_\eta = \mathbf{R}(\mathbf{I}_3+ \boldsymbol{\widehat{\omega}}_\eta)$, we use the exponential map
\begin{equation*}
    \mathbf{R}_\eta = \mathbf{R}\exp(\boldsymbol{\widehat{\omega}}_\eta),
\end{equation*}
which preserves orthogonality exactly and is more stable numerically. For small perturbation, this quantity is easy to compute using the Rodrigues formula. The full procedure is summarised in \cref{alg:sens}.

\begin{algorithm}[htbp]
\caption{Sensitivity analysis of the permeability computation}
\label{alg:sens}
\begin{algorithmic}[1]
\Require Diagonal tensor $\mathbf{K}_{\rm eff}$, Rotation matrix $\mathbf{R}$.
\State Choose the small parameter $\eta$.
\State Sample randomly unit vector $u$, and set $\omega_\eta=\eta u$.
\State Build $\boldsymbol{\widehat{\omega}}_\eta$.
\State Compute $\mathbf{R}_\eta = \mathbf{R} \exp(\boldsymbol{\widehat{\omega}}_\eta)$. \Comment{Using the Rodrigues Formula}
\State Evaluate $\mathbf{K}_\eta = \mathbf{R}_\eta \mathbf{K}_{\rm eff} \mathbf{R}_\eta^\top$.
\State Evaluate $\mathbf{K} = \mathbf{R}\mathbf{K}_{\rm eff} \mathbf{R}^\top$.
\State Measure $\lVert \mathbf{K}_\eta - \mathbf{K} \rVert / \lVert \mathbf{K} \rVert$.
\end{algorithmic}
\end{algorithm}
Using the same notations as in \eqref{eq:Kalgo}, and denoting by $\mathbf{K}^\kappa_\mathrm{coil,algo,\eta}$ the perturbed permeability tensor, we compute the following relative error
\begin{equation*}
{\rm err}_{\eta} = \frac{ \lVert  \mathbf{K}^\kappa_\mathrm{coil,algo,\eta} - \mathbf{K}^\kappa_\mathrm{coil,algo}\rVert_F} {\lVert \mathbf{K}^\kappa_\mathrm{coil,algo} \rVert_F}.    
\end{equation*} 
To ensure robustness with respect to rotational direction, sensitivity measures have been averaged over 100 random perturbation directions. The results are presented in \cref{fig:Sensitivity}. They indicate that the proposed approach is robust with respect to rotational perturbations. This robustness is explained by the relatively clustered eigenvalues of the permeability tensors, as shown in \cref{fig:PermBoutin1}, which limits the sensitivity of the transformed tensor to changes in orientation. Furthermore, the approximate trend of the relative error with respect to the perturbation magnitude $\eta$ in log--log scale confirms a first-order dependence of $\mathbf{K}$ on rotational perturbations.

\begin{figure}[H]
    \centering
    \begin{subfigure}[b]{0.49\textwidth}
        \include{figures/PLOTS/SensitivityErrorsIndex.tex}
    \end{subfigure}\hfill
    \begin{subfigure}[b]{0.49\textwidth}
        \include{figures/PLOTS/SensitivityErrorsEta}
    \end{subfigure}\hfill
\caption{Mean sensitivity with variability over 100 random perturbation directions.}
  \label{fig:Sensitivity}
\end{figure}

\bibliographystyle{plain}
\bibliography{referencesPaper}

 \end{document}

%% file: figures/PLOTS/PermBoutinPaper_TRANS_RAND.tex
\begin{tikzpicture}[scale=0.9]

\begin{axis}[
name = plot1, 
log basis y={10},
tick align=outside,
tick pos=left,
xlabel={Porosity \(\displaystyle \phi\)},
xmajorgrids,
xmin=0.31, xmax=1,
xminorgrids,
xtick style={color=black},
ylabel={Permeability},
ymin=1e-5, ymax=1,
ymajorgrids,
yminorgrids,
axis on top,
minor y grid style={
    opacity=0.5,    
    line width=0.5pt
},
ymode=log,
ytick style={color=black},
legend cell align={left},
legend style={
anchor=north west,
at={(0.02,0.98)},
fill opacity=1, 
draw opacity=1, 
text opacity=1
}
]

\fill[gray!20, draw=none]
    (axis cs:0.63,1e-5) rectangle (axis cs:0.95,1);

\addplot [unia-orange, line width=1]
table {%
0.328163265306122 4.69779055612832e-05
0.342244897959184 5.55947455867969e-05
0.356326530612245 6.55185851583973e-05
0.370408163265306 7.69262782985158e-05
0.384489795918367 9.0019208764214e-05
0.398571428571429 0.00010502716156452
0.41265306122449 0.000122212898673989
0.426734693877551 0.0001418775859734
0.440816326530612 0.000164367253428091
0.454897959183673 0.000190080514904766
0.468979591836735 0.000219477830672125
0.483061224489796 0.000253092668331955
0.497142857142857 0.000291545011806375
0.511224489795918 0.000335557790011198
0.52530612244898 0.000385976956473804
0.539387755102041 0.000443796161500189
0.553469387755102 0.000510187237765207
0.567551020408163 0.000586538093936535
0.581632653061224 0.000674500115334652
0.595714285714286 0.000776047857472847
0.609795918367347 0.000893554762553819
0.623877551020408 0.00102988994031398
0.637959183673469 0.00118854289554844
0.652040816326531 0.00137378569927984
0.666122448979592 0.00159088586051406
0.680204081632653 0.00184638863547035
0.694285714285714 0.00214849561408981
0.708367346938775 0.00250757859382163
0.722448979591837 0.0029368863427724
0.736530612244898 0.00345353078748012
0.750612244897959 0.00407988510779337
0.76469387755102 0.00484560083002051
0.778775510204082 0.00579057517390748
0.792857142857143 0.00696941226431119
0.806938775510204 0.00845829625655892
0.821020408163265 0.01036587771792
0.835102040816327 0.0128510710274503
0.849183673469388 0.0161532323828661
0.863265306122449 0.0206455599856338
0.87734693877551 0.0269344779083657
0.891428571428571 0.0360561867827464
0.905510204081633 0.0498955109654998
0.919591836734694 0.07216660732704
0.933673469387755 0.111009908223644
0.947755102040816 0.1871336432956
0.961836734693877 0.366576209821991
};
\addlegendentry{$K_{\rm iso}$}
\addplot [unia-green, dashed, line width =2]
table {%
0.328163265306122 8.25804521053592e-05
0.342244897959184 9.79191716080374e-05
0.356326530612245 0.000115581253714309
0.370408163265306 0.000135867925837757
0.384489795918367 0.000159117751744345
0.398571428571429 0.00018571165576078
0.41265306122449 0.000216078730277413
0.426734693877551 0.0002507029706084
0.440816326530612 0.000290131112184859
0.454897959183673 0.000334981783572674
0.468979591836735 0.00038595623704123
0.483061224489796 0.000443850979139423
0.497142857142857 0.000509572700629004
0.511224489795918 0.000584156003062158
0.52530612244898 0.000668784544808238
0.539387755102041 0.000764816391261372
0.553469387755102 0.000873814564312004
0.567551020408163 0.000997584061436967
0.581632653061224 0.00113821697780527
0.595714285714286 0.00129814784757004
0.609795918367347 0.00148022196817733
0.623877551020408 0.00168778034856415
0.637959183673469 0.0019247661217269
0.652040816326531 0.00219585892062236
0.666122448979592 0.00250664603574545
0.680204081632653 0.0028638424567212
0.694285714285714 0.00327557661243555
0.708367346938775 0.00375176548396123
0.722448979591837 0.00430461290811836
0.736530612244898 0.00494928014699151
0.750612244897959 0.00570480118039264
0.76469387755102 0.00659535175179397
0.778775510204082 0.00765203971510073
0.792857142857143 0.00891548023503732
0.806938775510204 0.0104395813813917
0.821020408163265 0.0122972477227459
0.835102040816327 0.0145892185289877
0.849183673469388 0.017458213980829
0.863265306122449 0.0211124481238763
0.87734693877551 0.0258664930307445
0.891428571428571 0.0322162063893561
0.905510204081633 0.0409854391094471
0.919591836734694 0.0536380049661095
0.933673469387755 0.0730162505293543
0.947755102040816 0.105365044197498
0.961836734693877 0.16720197372748
0.975918367346939 0.318541335148184
};
\addlegendentry{$K^\prime_{\perp, \mathrm{p}}$}
\addplot [unia-pink, dashed, line width=2]
table {%
0.328163265306122 2.16087542724322e-05
0.342244897959184 2.57450894557819e-05
0.356326530612245 3.05439379347899e-05
0.370408163265306 3.61000326104448e-05
0.384489795918367 4.25216273771818e-05
0.398571428571429 4.99326113208234e-05
0.41265306122449 5.84749958314748e-05
0.426734693877551 6.83118493174542e-05
0.440816326530612 7.96307711623967e-05
0.454897959183673 9.2648017867896e-05
0.468979591836735 0.000107613421218043
0.483061224489796 0.000124816272440605
0.497142857142857 0.000144592389917912
0.511224489795918 0.000167332643958732
0.52530612244898 0.000193493284446055
0.539387755102041 0.000223608511215378
0.553469387755102 0.000258305850179976
0.567551020408163 0.000298325060712216
0.581632653061224 0.000344541515849274
0.595714285714286 0.000397995286562703
0.609795918367347 0.000459927553125825
0.623877551020408 0.000531826501528319
0.637959183673469 0.000615485600557934
0.652040816326531 0.000713078183415831
0.666122448979592 0.000827253707605246
0.680204081632653 0.000961263136687323
0.694285714285714 0.00111912388217951
0.708367346938775 0.00130583913968623
0.722448979591837 0.00152769300791871
0.736530612244898 0.00179265272130847
0.750612244897959 0.0021109246919263
0.76469387755102 0.00249573529556149
0.778775510204082 0.00296444645318153
0.792857142857143 0.00354018079429615
0.806938775510204 0.00425424137897051
0.821020408163265 0.00514980454602214
0.835102040816327 0.00628771698568626
0.849183673469388 0.00775589696309406
0.863265306122449 0.00968517009645217
0.87734693877551 0.0122771676317478
0.891428571428571 0.0158561977963938
0.905510204081633 0.0209722805584064
0.919591836734694 0.028623554084953
0.933673469387755 0.0407911909729228
0.947755102040816 0.0619303862610072
0.961836734693877 0.104127539334098
0.975918367346939 0.212601152609326
};
\addlegendentry{$K^\prime_{\perp, \mathrm{v}}$}
\addplot[
    color=black,
    mark=triangle,
    only marks,
    mark size=3pt,
    every mark/.append style={solid},
    error bars/.cd,
        y dir=both,
        y explicit
] table[
    x=PHI ,
    y=mean_kperp,
    y error=std_kperp
]{figures/PLOTS/outputPertubation2.dat};
\addlegendentry{$K_{\perp, \mathrm{exp}}$}
\addplot[
    color=red,
    mark=triangle,
    only marks,
    mark size=3pt,
    every mark/.append style={solid},
    error bars/.cd,
        y dir=both,
        y explicit
] table[
    x=PHI ,
    y=mean_kperp_over,
    y error=std_kperp_over
]{figures/PLOTS/outputPertubation2.dat};
\addlegendentry{$\widetilde{K}_{\perp, \mathrm{exp}}$}
\end{axis}

\node at (plot1.south) [yshift=-45pt] {(a) Transversal Permeability};

\end{tikzpicture}

%% file: figures/PLOTS/PermBoutinPaper_LONG_RAND.tex
\begin{tikzpicture}[scale=0.9]

\begin{axis}[
name = plot2, 
log basis y={10},
tick align=outside,
tick pos=left,
xlabel={Porosity \(\displaystyle \phi\)},
xmajorgrids,
xmin=0.31, xmax=1,
xminorgrids,
xtick style={color=black},
ylabel={Permeability},
ymin=1e-5, ymax=1,
ymajorgrids,
yminorgrids,
axis on top,
minor y grid style={
    opacity=0.5,    
    line width=0.5pt
},
ymode=log,
ytick style={color=black},
legend cell align={left},
legend style={
anchor=north west,
at={(0.02,0.98)},
fill opacity=1, 
draw opacity=1, 
text opacity=1
}
]

\fill[gray!20, draw=none]
  (axis cs:0.63,1e-5) rectangle (axis cs:0.95,1);

\addplot [unia-orange, line width=1]
table {%
0.328163265306122 4.69779055612832e-05
0.342244897959184 5.55947455867969e-05
0.356326530612245 6.55185851583973e-05
0.370408163265306 7.69262782985158e-05
0.384489795918367 9.0019208764214e-05
0.398571428571429 0.00010502716156452
0.41265306122449 0.000122212898673989
0.426734693877551 0.0001418775859734
0.440816326530612 0.000164367253428091
0.454897959183673 0.000190080514904766
0.468979591836735 0.000219477830672125
0.483061224489796 0.000253092668331955
0.497142857142857 0.000291545011806375
0.511224489795918 0.000335557790011198
0.52530612244898 0.000385976956473804
0.539387755102041 0.000443796161500189
0.553469387755102 0.000510187237765207
0.567551020408163 0.000586538093936535
0.581632653061224 0.000674500115334652
0.595714285714286 0.000776047857472847
0.609795918367347 0.000893554762553819
0.623877551020408 0.00102988994031398
0.637959183673469 0.00118854289554844
0.652040816326531 0.00137378569927984
0.666122448979592 0.00159088586051406
0.680204081632653 0.00184638863547035
0.694285714285714 0.00214849561408981
0.708367346938775 0.00250757859382163
0.722448979591837 0.0029368863427724
0.736530612244898 0.00345353078748012
0.750612244897959 0.00407988510779337
0.76469387755102 0.00484560083002051
0.778775510204082 0.00579057517390748
0.792857142857143 0.00696941226431119
0.806938775510204 0.00845829625655892
0.821020408163265 0.01036587771792
0.835102040816327 0.0128510710274503
0.849183673469388 0.0161532323828661
0.863265306122449 0.0206455599856338
0.87734693877551 0.0269344779083657
0.891428571428571 0.0360561867827464
0.905510204081633 0.0498955109654998
0.919591836734694 0.07216660732704
0.933673469387755 0.111009908223644
0.947755102040816 0.1871336432956
0.961836734693877 0.366576209821991
};
\addlegendentry{$K_{\rm iso}$}
\addplot [unia-green, dotted, line width=2]
table {%
0.328163265306122 0.000131709672264658
0.342244897959184 0.00015488978402181
0.356326530612245 0.000181357472647725
0.370408163265306 0.000211514542688948
0.384489795918367 0.000245812148628047
0.398571428571429 0.000284757764256261
0.41265306122449 0.000328923299219363
0.426734693877551 0.000378954581571623
0.440816326530612 0.000435582473112364
0.454897959183673 0.000499635944171741
0.468979591836735 0.000572057509722019
0.483061224489796 0.000653921523644522
0.497142857142857 0.000746455948543542
0.511224489795918 0.000851068372485331
0.52530612244898 0.000969377241955137
0.539387755102041 0.00110324953638514
0.553469387755102 0.00125484644321117
0.567551020408163 0.00142667903026885
0.581632653061224 0.00162167649152501
0.595714285714286 0.00184327031466032
0.609795918367347 0.00209549875854005
0.623877551020408 0.0023831374406206
0.637959183673469 0.00271186377165753
0.652040816326531 0.00308846566182986
0.666122448979592 0.00352110869165131
0.680204081632653 0.00401968129486503
0.694285714285714 0.00459624520733915
0.708367346938775 0.00526562969145211
0.722448979591837 0.00604622474350101
0.736530612244898 0.00696105369105987
0.750612244897959 0.00803924433371849
0.76469387755102 0.00931807859310571
0.778775510204082 0.0108458982710718
0.792857142857143 0.0126863052553367
0.806938775510204 0.0149243666892921
0.821020408163265 0.0176760112784753
0.835102040816327 0.0211026644402702
0.849183673469388 0.0254347956199764
0.863265306122449 0.0310112669128112
0.87734693877551 0.0383480953948307
0.891428571428571 0.048265251960511
0.905510204081633 0.062136403582282
0.919591836734694 0.082423289300833
0.933673469387755 0.113949206263445
0.947755102040816 0.167413986632705
0.961836734693877 0.271421376864006
0.975918367346939 0.531203925121178
};
\addlegendentry{$K_{\parallel, \mathrm{p}}$}
\addplot [unia-pink, dotted, line width=2]
table {%
0.328163265306122 4.32175085448644e-05
0.342244897959184 5.14901789115637e-05
0.356326530612245 6.10878758695797e-05
0.370408163265306 7.22000652208895e-05
0.384489795918367 8.50432547543635e-05
0.398571428571429 9.98652226416468e-05
0.41265306122449 0.00011694999166295
0.426734693877551 0.000136623698634908
0.440816326530612 0.000159261542324793
0.454897959183673 0.000185296035735792
0.468979591836735 0.000215226842436086
0.483061224489796 0.00024963254488121
0.497142857142857 0.000289184779835824
0.511224489795918 0.000334665287917464
0.52530612244898 0.000386986568892109
0.539387755102041 0.000447217022430756
0.553469387755102 0.000516611700359952
0.567551020408163 0.000596650121424433
0.581632653061224 0.000689083031698549
0.595714285714286 0.000795990573125407
0.609795918367347 0.00091985510625165
0.623877551020408 0.00106365300305664
0.637959183673469 0.00123097120111587
0.652040816326531 0.00142615636683166
0.666122448979592 0.00165450741521049
0.680204081632653 0.00192252627337465
0.694285714285714 0.00223824776435901
0.708367346938775 0.00261167827937246
0.722448979591837 0.00305538601583742
0.736530612244898 0.00358530544261693
0.750612244897959 0.00422184938385261
0.76469387755102 0.00499147059112298
0.778775510204082 0.00592889290636305
0.792857142857143 0.00708036158859231
0.806938775510204 0.00850848275794102
0.821020408163265 0.0102996090920443
0.835102040816327 0.0125754339713725
0.849183673469388 0.0155117939261881
0.863265306122449 0.0193703401929043
0.87734693877551 0.0245543352634956
0.891428571428571 0.0317123955927876
0.905510204081633 0.0419445611168128
0.919591836734694 0.057247108169906
0.933673469387755 0.0815823819458457
0.947755102040816 0.123860772522014
0.961836734693877 0.208255078668195
0.975918367346939 0.425202305218653
};
\addlegendentry{$K_{\parallel, \mathrm{v}}$}

\addplot[
    color=black,
    mark=o,
    only marks,
    mark size=2pt,
    every mark/.append style={solid},
    error bars/.cd,
        y dir=both,
        y explicit
] table[
    x=PHI ,
    y=mean_klong,
    y error=std_klong
]{figures/PLOTS/outputPertubation2.dat};
\addlegendentry{$K_{\parallel, \mathrm{exp}}$}
\addplot[
    color=red,
    mark=o,
    only marks,
    mark size=2pt,
    every mark/.append style={solid},
    error bars/.cd,
        y dir=both,
        y explicit
] table[
    x=PHI ,
    y=mean_klong_over,
    y error=std_klong_over,
]{figures/PLOTS/outputPertubation2.dat};
\addlegendentry{$\widetilde{K}_{\parallel, \mathrm{exp}}$}
\end{axis}

\node at (plot2.south) [yshift=-45pt] {(b) Longitudinal Permeability};

\end{tikzpicture}

%% file: figures/PLOTS/PermBoutinRadius_TRANS.tex
\begin{tikzpicture}[scale=0.9]
\begin{axis}[
name=plot1,
log basis y={10},
tick align=outside,
tick pos=left,
enlarge x limits = true,
xlabel={Radius \(\displaystyle R\)},
xmajorgrids,
xmin=0.1, xmax=0.4,
xminorgrids,
xtick style={color=black},
ylabel={Permeability},
ymin=8e-4, ymax=0.4,
ymajorgrids,
yminorgrids,
minor y grid style={
    opacity=0.5,    
    line width=0.5pt
},
ymode=log,
ytick style={color=black},
legend cell align={left},
legend style={
anchor=south west,
at={(0.02,0.02)},
fill opacity=1, 
draw opacity=1, 
text opacity=1
}
]
\addplot [unia-orange, line width=1pt]
table {%
0.09 0.308745395347919
0.0982051282051282 0.255455424433351
0.106410256410256 0.214047663411465
0.114615384615385 0.181247467411453
0.122820512820513 0.154835963143059
0.131025641025641 0.133266421145126
0.139230769230769 0.115433928190816
0.147435897435897 0.100532145248355
0.155641025641026 0.0879614566700972
0.163846153846154 0.077268470142316
0.172051282051282 0.068105208645072
0.18025641025641 0.060200995881195
0.188461538461538 0.0533427154066235
0.196666666666667 0.0473607098335565
0.204871794871795 0.0421185509488224
0.213076923076923 0.0375055122750276
0.221282051282051 0.0334309579362632
0.229487179487179 0.0298201099109335
0.237692307692308 0.0266108198512992
0.245897435897436 0.0237510819530587
0.254102564102564 0.0211970986490752
0.262307692307692 0.0189117630274671
0.27051282051282 0.0168634584409181
0.278717948717949 0.0150251017416996
0.286923076923077 0.0133733752295138
0.295128205128205 0.0118881059413051
0.303333333333333 0.0105517608427423
0.311538461538461 0.00934903383201354
0.31974358974359 0.00826650595631058
0.327948717948718 0.00729236437544183
0.336153846153846 0.00641616874477596
0.344358974358974 0.00562865608906646
0.352564102564103 0.00492157708638864
0.360769230769231 0.00428755811381908
0.368974358974359 0.00371998452401359
0.377179487179487 0.00321290149899123
0.385384615384615 0.00276092951988275
0.393589743589744 0.00235919204104894
0.401794871794872 0.00200325339555887
0.41 0.00168906531073422
};
\addlegendentry{$K_{\rm iso}$}
\addplot [unia-green, dashed, line width=2pt]
table {%
0.09 0.106333629051352
0.0982051282051282 0.0994120803192762
0.106410256410256 0.09305409037392
0.114615384615385 0.0871778061270903
0.122820512820513 0.0817186723412519
0.131025641025641 0.0766249671911025
0.139230769230769 0.0718546948688095
0.147435897435897 0.067373367752987
0.155641025641026 0.0631523882543093
0.163846153846154 0.0591678448832086
0.172051282051282 0.0553996006112697
0.18025641025641 0.0518305914077351
0.188461538461538 0.0484462784470351
0.196666666666667 0.045234214357862
0.204871794871795 0.0421836952405297
0.213076923076923 0.0392854779715284
0.221282051282051 0.0365315477570091
0.229487179487179 0.0339149247621077
0.237692307692308 0.0314295014306019
0.245897435897436 0.0290699041492873
0.254102564102564 0.0268313744248247
0.262307692307692 0.0247096658780928
0.27051282051282 0.0227009542258919
0.278717948717949 0.0208017580843775
0.286923076923077 0.0190088689436002
0.295128205128205 0.0173192890639691
0.303333333333333 0.0157301763590305
0.311538461538461 0.0142387955731416
0.31974358974359 0.0128424752508469
0.327948717948718 0.0115385701369326
0.336153846153846 0.0103244287496888
0.344358974358974 0.00919736594069587
0.352564102564103 0.00815464029722391
0.360769230769231 0.00719343626224952
0.368974358974359 0.00631085084593948
0.377179487179487 0.00550388478488995
0.385384615384615 0.00476943797504498
0.393589743589744 0.00410430896465875
0.401794871794872 0.00350519824849212
0.41 0.0029687150571383
};
\addlegendentry{$K^\prime_{\perp, \mathrm{p}}$}
\addplot [color=unia-pink, dashed, line width=2pt]
table {%
0.09 0.0704428941806166
0.0982051282051282 0.0642306784238219
0.106410256410256 0.0586318752688871
0.114615384615385 0.0535617735123482
0.122820512820513 0.0489527886504484
0.131025641025641 0.0447500111744279
0.139230769230769 0.0409081126593468
0.147435897435897 0.0373891421493384
0.155641025641026 0.0341609229382107
0.163846153846154 0.0311958642792872
0.172051282051282 0.0284700660857339
0.18025641025641 0.0259626344919558
0.188461538461538 0.0236551517585902
0.196666666666667 0.0215312608732127
0.204871794871795 0.0195763365472621
0.213076923076923 0.0177772220909985
0.221282051282051 0.016122017077081
0.229487179487179 0.0145999045513034
0.237692307692308 0.0132010093163678
0.245897435897436 0.0119162808315467
0.254102564102564 0.0107373957593835
0.262307692307692 0.00965667630138324
0.27051282051282 0.00866702130248745
0.278717948717949 0.00776184774232016
0.286923076923077 0.00693504072171881
0.295128205128205 0.00618091043328789
0.303333333333333 0.00549415490172391
0.311538461538461 0.00486982751333968
0.31974358974359 0.00430330853927792
0.327948717948718 0.0037902800043508
0.336153846153846 0.00332670337156095
0.344358974358974 0.00290879960745685
0.352564102564103 0.00253303127038031
0.360769230769231 0.00219608632610351
0.368974358974359 0.0018948634462306
0.377179487179487 0.00162645858632092
0.385384615384615 0.00138815267477168
0.393589743589744 0.00117740027149192
0.401794871794872 0.000991819078443221
0.41 0.000829180203120002
};
\addlegendentry{$K^\prime_{\perp, \mathrm{v}}$}
\addplot [black, mark=triangle*, mark size=4, mark options={solid}, only marks]
table {%
0.1 0.0814749
0.15 0.052092
0.2 0.0330274
0.25 0.0199697
0.3 0.0110303
0.35 0.00521688
0.4 0.00185181
};
\addlegendentry{$K_{\perp, \mathrm{exp}}$}
\addplot [red, mark=triangle*, mark size=4, mark options={solid}, only marks]
table {%
0.1 0.0885525987957792
0.15 0.0564019199594727
0.2 0.0353635791633911
0.25 0.0210193117156016
0.3 0.0114795446686115
0.35 0.00548967283057093
0.4 0.00212954572506773
};
\addlegendentry{$\widetilde{K}_{\perp, \mathrm{exp}}$}
\draw (axis cs:0.11,0.08) node[
  scale=0.75,
  anchor=west,
  text=black,
  rotate=0.0,
  font=\bfseries
]{$\phi=$0.97};
\draw (axis cs:0.16,0.056) node[
  scale=0.75,
  anchor=west,
  text=black,
  rotate=0.0,
  font=\bfseries
]{$\phi=$0.93};

\draw (axis cs:0.205,0.035) node[
  scale=0.75,
  anchor=west,
  text=black,
  rotate=0.0,
  font=\bfseries
]{$\phi=$0.87};

\draw (axis cs:0.26,0.021) node[
  scale=0.75,
  anchor=west,
  text=black,
  rotate=0.0,
  font=\bfseries
]{$\phi=$0.80};

\draw (axis cs:0.31,0.011) node[
  scale=0.75,
  anchor=west,
  text=black,
  rotate=0.0,
  font=\bfseries
]{$\phi=$0.72};

\draw (axis cs:0.35,0.006) node[
  scale=0.75,
  anchor=west,
  text=black,
  rotate=0.0,
  font=\bfseries
]{$\phi=$0.62};

\draw (axis cs:0.37,0.0014) node[
  scale=0.75,
  anchor=west,
  text=black,
  rotate=0.0,
  font=\bfseries
]{$\phi=$0.50};
\end{axis}

\node at (plot1.south) [yshift=-45pt] {(a) Transversal Permeability};

\end{tikzpicture}

%% file: figures/PLOTS/PermBoutinRadius_LONG.tex
\begin{tikzpicture}[scale=0.9]
\begin{axis}[
name=plot2,
log basis y={10},
tick align=outside,
tick pos=left,
enlarge x limits = true,
xlabel={Radius \(\displaystyle R\)},
xmajorgrids,
xmin=0.1, xmax=0.4,
xminorgrids,
xtick style={color=black},
ylabel={Permeability},
ymin=8e-4, ymax=0.4,
ymajorgrids,
yminorgrids,
minor y grid style={
    opacity=0.5,    
    line width=0.5pt
},
ymode=log,
ytick style={color=black},
legend cell align={left},
legend style={
anchor=south west,
at={(0.02,0.02)},
fill opacity=1, 
draw opacity=1, 
text opacity=1
}
]
\addplot [unia-orange, line width=1]
table {%
0.09 0.308745395347919
0.0982051282051282 0.255455424433351
0.106410256410256 0.214047663411465
0.114615384615385 0.181247467411453
0.122820512820513 0.154835963143059
0.131025641025641 0.133266421145126
0.139230769230769 0.115433928190816
0.147435897435897 0.100532145248355
0.155641025641026 0.0879614566700972
0.163846153846154 0.077268470142316
0.172051282051282 0.068105208645072
0.18025641025641 0.060200995881195
0.188461538461538 0.0533427154066235
0.196666666666667 0.0473607098335565
0.204871794871795 0.0421185509488224
0.213076923076923 0.0375055122750276
0.221282051282051 0.0334309579362632
0.229487179487179 0.0298201099109335
0.237692307692308 0.0266108198512992
0.245897435897436 0.0237510819530587
0.254102564102564 0.0211970986490752
0.262307692307692 0.0189117630274671
0.27051282051282 0.0168634584409181
0.278717948717949 0.0150251017416996
0.286923076923077 0.0133733752295138
0.295128205128205 0.0118881059413051
0.303333333333333 0.0105517608427423
0.311538461538461 0.00934903383201354
0.31974358974359 0.00826650595631058
0.327948717948718 0.00729236437544183
0.336153846153846 0.00641616874477596
0.344358974358974 0.00562865608906646
0.352564102564103 0.00492157708638864
0.360769230769231 0.00428755811381908
0.368974358974359 0.00371998452401359
0.377179487179487 0.00321290149899123
0.385384615384615 0.00276092951988275
0.393589743589744 0.00235919204104894
0.401794871794872 0.00200325339555887
0.41 0.00168906531073422
};
\addlegendentry{$K_{\rm iso}$}
\addplot [unia-green, dotted, line width=2]
table {%
0.09 0.17679976408872
0.0982051282051282 0.163675054803915
0.106410256410256 0.151729524102943
0.114615384615385 0.140796835225739
0.122820512820513 0.130745048696289
0.131025641025641 0.121467698598036
0.139230769230769 0.112877584120719
0.147435897435897 0.104902341568621
0.155641025641026 0.0974812172050581
0.163846153846154 0.0905626701685706
0.172051282051282 0.0841025617398436
0.18025641025641 0.0780627668444716
0.188461538461538 0.0724100948867721
0.196666666666667 0.067115440734387
0.204871794871795 0.0621531093550107
0.213076923076923 0.0575002731537523
0.221282051282051 0.0531365319017779
0.229487179487179 0.0490435528276361
0.237692307692308 0.0452047739627104
0.245897435897436 0.0416051578522231
0.254102564102564 0.0382309857067012
0.262307692307692 0.0350696842782634
0.27051282051282 0.0321096794107062
0.278717948717949 0.0293402714787812
0.286923076923077 0.0267515289043429
0.295128205128205 0.0243341966898701
0.303333333333333 0.0220796174974555
0.311538461538461 0.0199796632634003
0.31974358974359 0.0180266757044658
0.327948717948718 0.0162134143635288
0.336153846153846 0.0145330110764018
0.344358974358974 0.0129789299304125
0.352564102564103 0.0115449319385977
0.360769230769231 0.0102250437783964
0.368974358974359 0.0090135300462607
0.377179487179487 0.00790486856408816
0.385384615384615 0.00689372834331388
0.393589743589744 0.0059749498706427
0.401794871794872 0.00514352742794808
0.41 0.00439459319955273
};
\addlegendentry{$K_{\parallel, \mathrm{p}}$}
\addplot [unia-pink, dotted, line width=2]
table {%
0.09 0.140885788361233
0.0982051282051282 0.128461356847644
0.106410256410256 0.117263750537774
0.114615384615385 0.107123547024696
0.122820512820513 0.0979055773008969
0.131025641025641 0.0895000223488558
0.139230769230769 0.0818162253186936
0.147435897435897 0.0747782842986769
0.155641025641026 0.0683218458764215
0.163846153846154 0.0623917285585743
0.172051282051282 0.0569401321714678
0.18025641025641 0.0519252689839116
0.188461538461538 0.0473103035171805
0.196666666666667 0.0430625217464253
0.204871794871795 0.0391526730945242
0.213076923076923 0.0355544441819971
0.221282051282051 0.0322440341541621
0.229487179487179 0.0291998091026068
0.237692307692308 0.0264020186327357
0.245897435897436 0.0238325616630933
0.254102564102564 0.0214747915187671
0.262307692307692 0.0193133526027665
0.27051282051282 0.0173340426049749
0.278717948717949 0.0155236954846403
0.286923076923077 0.0138700814434376
0.295128205128205 0.0123618208665758
0.303333333333333 0.0109883098034478
0.311538461538461 0.00973965502667935
0.31974358974359 0.00860661707855583
0.327948717948718 0.0075805600087016
0.336153846153846 0.0066534067431219
0.344358974358974 0.0058175992149137
0.352564102564103 0.00506606254076062
0.360769230769231 0.00439217265220702
0.368974358974359 0.00378972689246121
0.377179487179487 0.00325291717264185
0.385384615384615 0.00277630534954336
0.393589743589744 0.00235480054298383
0.401794871794872 0.00198363815688644
0.41 0.00165836040624
};
\addlegendentry{$K_{\parallel, \mathrm{v}}$}
\addplot [black, mark=*, mark size=3, mark options={solid}, only marks]
table {%
0.1 0.16304
0.15 0.104584
0.2 0.0670975
0.25 0.0418986
0.3 0.0249384
0.35 0.0138552
0.4 0.00701791
};
\addlegendentry{$K_{\parallel, \mathrm{exp}}$}
\addplot [red, mark=*, mark size=3, mark options={solid}, only marks]
table {%
0.1 0.172737763951211
0.15 0.110092239429699
0.2 0.0694273598760476
0.25 0.0421351279036947
0.3 0.0240523426101228
0.35 0.0125865776451545
0.4 0.00582761126586956
};
\addlegendentry{$\widetilde{K}_{\parallel, \mathrm{exp}}$}
\draw (axis cs:0.11,0.165787) node[
  scale=0.75,
  anchor=west,
  text=black,
  rotate=0.0,
   font=\bfseries
]{$\phi=$0.97};
\draw (axis cs:0.16,0.106449) node[
  scale=0.75,
  anchor=west,
  text=black,
  rotate=0.0,
   font=\bfseries
]{$\phi=$0.93};
\draw (axis cs:0.21,0.0683593) node[
  scale=0.75,
  anchor=west,
  text=black,
  rotate=0.0,
   font=\bfseries
]{$\phi=$0.87};
\draw (axis cs:0.26,0.0425507) node[
  scale=0.75,
  anchor=west,
  text=black,
  rotate=0.0,
   font=\bfseries
]{$\phi=$0.80};
\draw (axis cs:0.31,0.0252411) node[
  scale=0.75,
  anchor=west,
  text=black,
  rotate=0.0,
   font=\bfseries
]{$\phi=$0.72};
\draw (axis cs:0.36,0.0139536) node[
  scale=0.75,
  anchor=west,
  text=black,
  rotate=0.0,
   font=\bfseries
]{$\phi=$0.62};
\draw (axis cs:0.37,0.004) node[
  scale=0.75,
  anchor=west,
  text=black,
  rotate=0.0,
   font=\bfseries
]{$\phi=$0.50};
\end{axis}

\node at (plot2.south) [yshift=-45pt] {(b) Longitudinal Permeability};

\end{tikzpicture}

%% file: figures/PLOTS/PermRVEsPaper0.178.tex
\begin{tikzpicture}[scale=1.2]
\begin{axis}[
legend cell align={left},
legend columns=1,
legend style={
fill opacity=1, 
draw opacity=1, 
text opacity=1,
at={(1.1,1)}, 
anchor=north west},
log basis y={10},
tick align=outside,
tick pos=left,
xlabel={Porosity \(\displaystyle \phi\)},
xmajorgrids,
xmin=0.6, xmax=1.0,
minor x tick num=1,
xminorgrids,
xtick style={color=black},
ylabel={Permeability},
ymajorgrids,
ymin=4e-4, ymax=0.2,
yminorgrids,
minor grid style={
    opacity=0.5,    
    line width=0.5pt
},
ymode=log,
ytick style={color=black}
]
\addplot [unia-orange, line width=1]
table {%
0.607755102040816 0.00123276407694661
0.617551020408163 0.0013604427006662
0.62734693877551 0.00150218250584566
0.637142857142857 0.00165976972954601
0.646938775510204 0.001835262696576
0.656734693877551 0.00203104148059251
0.666530612244898 0.00224986828449488
0.676326530612245 0.00249496124021012
0.686122448979592 0.0027700851126107
0.695918367346939 0.00307966343751507
0.705714285714286 0.003428918027601
0.715510204081633 0.00382404368233837
0.72530612244898 0.00427242853984691
0.735102040816326 0.00478293410271174
0.744897959183674 0.00536625398770068
0.75469387755102 0.0060353775350255
0.764489795918367 0.00680619454088514
0.774285714285714 0.00769829204445563
0.784081632653061 0.00873601564333698
0.793877551020408 0.00994989993856011
0.803673469387755 0.0113786214203438
0.813469387755102 0.0130717022992988
0.823265306122449 0.0150933121682345
0.833061224489796 0.0175277048158505
0.842857142857143 0.0204871412388823
0.85265306122449 0.0241236805400957
0.862448979591837 0.0286471447133836
0.872244897959184 0.0343532273786122
0.882040816326531 0.0416688262954867
0.891836734693878 0.0512277437653413
0.901632653061224 0.0640023197224398
0.911428571428571 0.081543501087637
0.921224489795918 0.106444373881923
0.931020408163265 0.14329963992478
};
\addlegendentry{$K_{\rm iso}$}
\addplot [unia-green, dotted, line width=2]
table {%
0.607755102040816 0.0028964411461472
0.617551020408163 0.00316725627562662
0.62734693877551 0.00346421910412861
0.637142857142857 0.00379021067696423
0.646938775510204 0.00414849620903457
0.656734693877551 0.00454278736066446
0.666530612244898 0.00497731660180829
0.676326530612245 0.005456926430624
0.686122448979592 0.00598717694851685
0.695918367346939 0.00657447625520374
0.705714285714286 0.00722623939278297
0.715510204081633 0.00795108324691874
0.72530612244898 0.00875906706062797
0.735102040816326 0.00966199125187944
0.744897959183674 0.0106737713672398
0.75469387755102 0.0118109097121061
0.764489795918367 0.0130930951557709
0.774285714285714 0.0145439728378244
0.784081632653061 0.0161921415528588
0.793877551020408 0.0180724598589336
0.803673469387755 0.020227776205897
0.813469387755102 0.0227112496387406
0.823265306122449 0.0255895057468982
0.833061224489796 0.0289469939624113
0.842857142857143 0.0328921052552489
0.85265306122449 0.0375659233727216
0.862448979591837 0.043155008098283
0.872244897959184 0.0499105146432941
0.882040816326531 0.0581775689383224
0.891836734693878 0.0684418156007728
0.901632653061224 0.0814058701806243
0.911428571428571 0.0981202952596534
0.921224489795918 0.120219573640991
0.931020408163265 0.150374113792454
};
\addlegendentry{$K_{\parallel, \mathrm{p}}$}
\addplot [unia-green, dashed, line width=2]
table {%
0.607755102040816 0.00204515350221135
0.617551020408163 0.00224060793737456
0.62734693877551 0.00245483307870805
0.637142857142857 0.00268982437911517
0.646938775510204 0.00294783171911509
0.656734693877551 0.00323139929421652
0.666530612244898 0.00354341308115588
0.676326530612245 0.00388715759137711
0.686122448979592 0.00426638406947272
0.695918367346939 0.00468539287988941
0.705714285714286 0.00514913359419206
0.715510204081633 0.00566332730924726
0.72530612244898 0.0062346170861724
0.735102040816326 0.00687075423203056
0.744897959183674 0.00758083063964727
0.75469387755102 0.00837557082994478
0.764489795918367 0.00926770211010762
0.774285714285714 0.0102724279736814
0.784081632653061 0.0114080394411929
0.793877551020408 0.0126967128831726
0.803673469387755 0.0141655631938247
0.813469387755102 0.0158480515271332
0.823265306122449 0.0177858929306149
0.833061224489796 0.0200316807249171
0.842857142857143 0.0226525578046472
0.85265306122449 0.025735449024323
0.862448979591837 0.0293946757033262
0.872244897959184 0.0337833008149399
0.882040816326531 0.0391104918032308
0.891836734693878 0.0456689234411691
0.901632653061224 0.0538795999034901
0.911428571428571 0.0643683158391127
0.921224489795918 0.0781028031592667
0.931020408163265 0.0966542198523091
0.940816326530612 0.122735124259195
0.950612244897959 0.161419899129861
};
\addlegendentry{$K^\prime_{\perp, \mathrm{p}}$}
\addplot [unia-pink,dotted, line width=2]
table {%
0.607755102040816 0.00126840921185508
0.617551020408163 0.00140306394485239
0.62734693877551 0.00155257126202741
0.637142857142857 0.00171876064107249
0.646938775510204 0.00190372495261738
0.656734693877551 0.00210986535009829
0.666530612244898 0.00233994517671601
0.676326530612245 0.00259715500444652
0.686122448979592 0.00288519149360418
0.695918367346939 0.0032083535132015
0.705714285714286 0.00357165995525292
0.715510204081633 0.00398099499830422
0.72530612244898 0.00444328835127853
0.735102040816326 0.0049667404159345
0.744897959183674 0.00556110560151887
0.75469387755102 0.00623805158380238
0.764489795918367 0.00701161867830141
0.774285714285714 0.00789881252847132
0.784081632653061 0.00892037626565262
0.793877551020408 0.0101018071479995
0.803673469387755 0.0114747105342213
0.813469387755102 0.0130786258790986
0.823265306122449 0.0149635234224225
0.833061224489796 0.0171932700811108
0.842857142857143 0.0198505222906826
0.85265306122449 0.0230437637626387
0.862448979591837 0.0269176430410806
0.872244897959184 0.0316685219561577
0.882040816326531 0.0375685006206272
0.891836734693878 0.045003707697086
0.901632653061224 0.0545375607932428
0.911428571428571 0.0670197965672722
0.921224489795918 0.0837841235856364
0.931020408163265 0.107029255639376
0.940816326530612 0.140611950271557
};
\addlegendentry{$K_{\parallel, \mathrm{v}}$}
\addplot [unia-pink, dashed, line width=2]
table {%
0.607755102040816 0.000634204605927541
0.617551020408163 0.000701531972426195
0.62734693877551 0.000776285631013705
0.637142857142857 0.000859380320536243
0.646938775510204 0.000951862476308689
0.656734693877551 0.00105493267504914
0.666530612244898 0.00116997258835801
0.676326530612245 0.00129857750222326
0.686122448979592 0.00144259574680209
0.695918367346939 0.00160417675660075
0.705714285714286 0.00178582997762646
0.715510204081633 0.00199049749915211
0.72530612244898 0.00222164417563926
0.735102040816326 0.00248337020796725
0.744897959183674 0.00278055280075943
0.75469387755102 0.00311902579190119
0.764489795918367 0.0035058093391507
0.774285714285714 0.00394940626423566
0.784081632653061 0.00446018813282631
0.793877551020408 0.00505090357399973
0.803673469387755 0.00573735526711066
0.813469387755102 0.0065393129395493
0.823265306122449 0.00748176171121125
0.833061224489796 0.00859663504055538
0.842857142857143 0.00992526114534132
0.85265306122449 0.0115218818813194
0.862448979591837 0.0134588215205403
0.872244897959184 0.0158342609780789
0.882040816326531 0.0187842503103136
0.891836734693878 0.022501853848543
0.901632653061224 0.0272687803966214
0.911428571428571 0.0335098982836361
0.921224489795918 0.0418920617928182
0.931020408163265 0.0535146278196881
0.940816326530612 0.0703059751357783
0.950612244897959 0.0959341300736121
0.960408163265306 0.138192770934046
};
\addlegendentry{$K^\prime_{\perp, \mathrm{v}}$}
\addplot [unia-orange, mark=triangle, mark size=2, mark options={solid, line width=.7}, only marks]
table {%
0.633153    0.0007989792330497659
0.886856    0.023276806685805628
0.728251    0.0018133979975360425
0.709509    0.0021193025786059904
0.649907    0.0014664395035315634
0.703568    0.0032920035491120422
0.694623    0.002061257277726681
0.694224    0.002480357357415495
0.749029    0.002357958296099473
0.906124    0.030732707370684818
0.836649    0.009119274083112354
0.797412    0.004124664534540671

};
\addlegendentry{$\widetilde{K}_{1, \mathrm{coil}}$}

\addplot [unia-blue, mark=o, mark size=2, mark options={solid, line width=.7}, only marks]
table {%
0.633153    0.0010873172404165647
0.886856    0.028408106772277134
0.728251    0.0021633979394464647
0.709509    0.005666875250663364
0.649907    0.003041131084089365
0.703568    0.005350487645883514
0.694623    0.002378977614335368
0.694224    0.0027212862071277744
0.749029    0.002768382212827113
0.906124    0.04099940969185677
0.836649    0.01869293815177061
0.797412    0.011634268808767174
}; 
\addlegendentry{$\widetilde{K}_{2, \mathrm{coil}}$}

\addplot [unia-red, mark=square , mark size=2, mark options={solid, line width=.7}, only marks]
table {%
0.633153    0.001625018317120446
0.886856    0.04293609858621472
0.728251    0.0029708871239100517
0.709509    0.008165197011185185
0.649907    0.0032286645260319624
0.703568    0.00875855960381436
0.694623    0.0031400940559627413
0.694224    0.0038955365801178854
0.749029    0.0031202645004370475
0.906124    0.06147350274968979
0.836649    0.021586725374207933
0.797412    0.012695310167964878
};
\addlegendentry{$\widetilde{K}_{3, \mathrm{coil}}$}

\end{axis}
\end{tikzpicture}

%% file: figures/PLOTS/SamplingRateSensitivity.tex
\usepgfplotslibrary{fillbetween}
\begin{tikzpicture}

\begin{axis}[
name=rev_sensitivity,
ymode=log,
log basis y={10},
xlabel={Refinement step},
ylabel={Relative change [\%]},
xmin=28, xmax=52,
ymin=1e-1, ymax=1e2,
xtick={30,40,50},
xticklabels={{$30\rightarrow 40$},{$40\rightarrow 50$},{$50\rightarrow 60$}},
xmajorgrids,
xtick style={color=black},
ymajorgrids,
yminorgrids,
minor y grid style={
    opacity=0.5,    
    line width=0.5pt
},
tick align=outside,
tick pos=left,
legend cell align={left},
legend style={
anchor=north east,
at={(0.95,0.95)},
fill opacity=1,
draw opacity=1,
text opacity=1
}
]

\addplot[name path=phi_upper, draw=none, forget plot] table {
N y
30 5.2963
40 2.0812
50 0.9119
};

\addplot[name path=phi_lower, draw=none, forget plot] table {
N y
30 1.0000
40 0.1660
50 0.2457
};

\addplot[
    black,
    fill=unia-orange,
    fill opacity=0.12,
    draw=none,
    forget plot
] fill between[of=phi_upper and phi_lower];

\addplot[name path=K_upper, draw=none, forget plot] table {
N y
30 34.8608
40 16.6030
50 9.6846
};

\addplot[name path=K_lower, draw=none, forget plot] table {
N y
30 10.7326
40 1.7370
50 1.1536
};

\addplot[
    unia-green,
    fill=unia-green,
    fill opacity=0.12,
    draw=none,
    forget plot
] fill between[of=K_upper and K_lower];

\addplot[
    color=unia-orange,
    mark=o,
    line width=1pt
] table {
N mean
30 3.1481
40 1.1236
50 0.5788
};
\addlegendentry{Porosity \(\phi\)}

\addplot[
    color=unia-green,
    mark=square,
    line width=1pt
] table {
N mean
30 22.7967
40 9.1700
50 5.4191
};
\addlegendentry{Permeability \(\mathbf K\)}

\end{axis}

\end{tikzpicture}

%% file: figures/PLOTS/SensitivityErrorsIndex.tex
\begin{tikzpicture}[scale=0.8]

\begin{axis}[
name=plot1, 
tick align=outside,
tick pos=left,
xlabel={REV index},
xmajorgrids,
xminorgrids,
xtick style={color=black},
xtick={1,2,3,4,5,6,7,8,9,10,11,12},
ylabel={Relative Error ${\rm err}_{\eta}$},
ymin=7.e-5, ymax=2e-1,
ymajorgrids,
yminorgrids,
minor y grid style={
    opacity=0.5,    
    line width=0.5pt
},
ymode=log,
ytick style={color=black},
legend cell align={left},
legend style={
anchor=north,
at={(0.5,1.15)},
fill opacity=1, 
draw opacity=1, 
text opacity=1,
legend columns=3,
column sep=5pt
}
]
\addplot[
    color=unia-orange,
    mark=o,
    only marks,
    mark size=2pt,
    every mark/.append style={solid},
    error bars/.cd,
        y dir=both,
        y explicit
] table[
    x=Matrix  ,
    y=mean1e-03,
    y error=std1e-03
]{figures/PLOTS/matrix_errors_index.dat};
\addlegendentry{$\eta=10^{-3}$}
\addplot[
    color=unia-green,
    mark=triangle,
    only marks,
    mark size=3pt,
    every mark/.append style={solid},
    error bars/.cd,
        y dir=both,
        y explicit
] table[
    x=Matrix  ,
    y=mean1e-02,
    y error=std1e-02
]{figures/PLOTS/matrix_errors_index.dat};
\addlegendentry{$\eta=10^{-2}$}
\addplot[
    color=unia-pink,
    mark=square,
    only marks,
    mark size=2pt,
    every mark/.append style={solid},
    error bars/.cd,
        y dir=both,
        y explicit
] table[
    x=Matrix  ,
    y=mean1e-01,
    y error=std1e-01
]{figures/PLOTS/matrix_errors_index.dat};
\addlegendentry{$\eta=10^{-1}$}
\end{axis}

\node at (plot1.south)  [yshift=-45pt, text width=7cm, align=center]  {(a) Relative error for each REV for different perturbation magnitude $\eta$.};

\end{tikzpicture}

%% file: figures/PLOTS/SensitivityErrorsEta.tex
\begin{tikzpicture}[scale=0.8]

\begin{axis}[
name=plot2,
tick align=outside,
tick pos=left,
xlabel={$\eta$},
xmajorgrids,
xtick style={color=black},
ylabel={Relative Error ${\rm err}_{\eta}$},
ymajorgrids,
yminorgrids,
minor y grid style={
    opacity=0.5,    
    line width=0.5pt
},
ymode=log,
xmode=log,
ymin=7.e-5, ymax=2e-1,
ytick style={color=black},
legend cell align={left},
legend style={
anchor=north west,
at={(-0.05,1.15)},
fill opacity=1, 
draw opacity=1, 
text opacity=1,
legend columns=3,
column sep=5pt
}
]
\addplot[
    color=unia-pink, mark=o, mark size=2pt, line width=1pt] table[
    x=eta,
    y=mean1
    ]{figures/PLOTS/matrix_errors_eta.dat};

\addplot[
    color=unia-pink, mark=o, mark size=2pt, line width=1pt] table[
    x=eta,
    y=mean2
    ]{figures/PLOTS/matrix_errors_eta.dat};

\addplot[
    color=unia-pink, mark=o, mark size=2pt, line width=1pt] table[
    x=eta,
    y=mean3
    ]{figures/PLOTS/matrix_errors_eta.dat};

\addplot[
    color=unia-pink, mark=o, mark size=2pt, line width=1pt] table[
    x=eta,
    y=mean4
    ]{figures/PLOTS/matrix_errors_eta.dat};

\addplot[
    color=unia-pink, mark=o, mark size=2pt, line width=1pt] table[
    x=eta,
    y=mean5
    ]{figures/PLOTS/matrix_errors_eta.dat};

\addplot[
    color=unia-pink, mark=o, mark size=2pt, line width=1pt] table[
    x=eta,
    y=mean6
    ]{figures/PLOTS/matrix_errors_eta.dat};

\addplot[
    color=unia-pink, mark=o, mark size=2pt, line width=1pt] table[
    x=eta,
    y=mean7
    ]{figures/PLOTS/matrix_errors_eta.dat};

\addplot[
    color=unia-pink, mark=o, mark size=2pt, line width=1pt] table[
    x=eta,
    y=mean8
    ]{figures/PLOTS/matrix_errors_eta.dat};

\addplot[
    color=unia-pink, mark=o, mark size=2pt, line width=1pt] table[
    x=eta,
    y=mean9
    ]{figures/PLOTS/matrix_errors_eta.dat};

\addplot[
    color=unia-pink, mark=o, mark size=2pt, line width=1pt] table[
    x=eta,
    y=mean10
    ]{figures/PLOTS/matrix_errors_eta.dat};

\addplot[
    color=unia-pink, mark=o, mark size=2pt, line width=1pt] table[
    x=eta,
    y=mean11
    ]{figures/PLOTS/matrix_errors_eta.dat};

\addplot[
    color=unia-pink, mark=o, mark size=2pt, line width=1pt] table[
    x=eta,
    y=mean12
    ]{figures/PLOTS/matrix_errors_eta.dat};

\end{axis}

\node at (plot2.south)  [yshift=-45pt, text width=8cm, align=center]  {(b) Relative error with respect to the perturbation magnitude $\eta$ for all the 12 chosen REVs in log-log scale.};

\end{tikzpicture}